\documentclass[11pt]{article}

\usepackage[T1]{fontenc}
\usepackage[utf8]{inputenc}
\usepackage[a4paper,margin=1in]{geometry}
\usepackage{lmodern}
\usepackage{microtype}
\usepackage{etoolbox}
\usepackage{amsmath,amssymb,amsthm}
\usepackage{bm}
\usepackage{dsfont}
\usepackage{graphicx}
\usepackage{adjustbox}
\usepackage{booktabs}
\usepackage{multirow}
\usepackage{array}
\usepackage{threeparttable}
\usepackage{longtable}
\usepackage{pdflscape}
\usepackage{float}
\usepackage{algorithm}
\usepackage{algpseudocode}
\usepackage{tikz}
\usepackage{subcaption}
\usepackage{caption}
\usepackage{siunitx}
\usepackage{ragged2e}
\usepackage{setspace}
\usepackage{url}
\usepackage{natbib}
\usepackage[dvipsnames,svgnames,x11names]{xcolor}
\usepackage[colorlinks=true,linkcolor=blue,citecolor=blue,urlcolor=blue]{hyperref}

\DeclareUnicodeCharacter{2011}{\nobreakdash-}
\DeclareUnicodeCharacter{2013}{--}
\DeclareUnicodeCharacter{2019}{'}

\newtheorem{theorem}{Theorem}
\newtheorem{lemma}{Lemma}

\newtheorem{corollary}{Corollary}
\theoremstyle{definition}
\newtheorem{definition}{Definition}
\newtheorem{assumption}{Assumption}

\theoremstyle{remark}

\newcommand{\ind}{\mathds{1}}
\newcommand{\bma}{\mathbf{a}}

\newcommand{\bmA}{\mathbf{A}}
\newcommand{\bmx}{\mathbf{x}}

\newcommand{\bmX}{\mathbf{X}}

\newcommand{\bmU}{\mathbf{U}}
\newcommand{\bmu}{\mathbf{u}}
\newcommand{\bmV}{\mathbf{V}}

\newcommand{\bmF}{\mathbf{F}}
\newcommand{\bmv}{\mathbf{v}}
\newcommand{\bmg}{\mathbf{g}}
\newcommand{\bmzeta}{\bm{\zeta}}

\newcommand{\cmH}{\mathcal{H}}
\newcommand{\bmPA}{\mathbf{PA}}
\newcommand{\bmpa}{\mathbf{pa}}
\newcommand{\bmw}{\mathbf{w}}
\newcommand{\bmW}{\mathbf{W}}
\newcommand{\bmT}{\mathbf{T}}

\newcommand{\bmxi}{\bm{\xi}}
\newcommand{\bmtheta}{\bm{\theta}}
\newcommand{\bmB}{\mathcal{B}}
\newcommand{\bmI}{\mathcal{I}}
\newcommand{\prob}{\mathbb{P}}
\newcommand{\texto}{\text{O}}
\newcommand{\nm}[1]{\left\|#1\right\|}
\DeclareMathOperator{\EL}{EL}
\DeclareMathOperator{\UL}{UL}

\newcommand{\zc}[1]{#1}
\newcommand{\smalleq}[2][0.85]{\scalebox{#1}{$\displaystyle #2$}}
\newcommand{\up}{\rule{0pt}{2.6ex}}
\newcommand{\down}{\rule[-1.2ex]{0pt}{1.2ex}}
\newcommand{\TABLE}[3]{\centering\caption{#1}\begin{adjustbox}{max width=\textwidth}#2\end{adjustbox}\ifstrempty{#3}{}{\par\smallskip\noindent\begin{minipage}{0.96\textwidth}\footnotesize #3\end{minipage}}}
\newcommand{\FIGURE}[3]{\centering#1\caption{#2}\ifstrempty{#3}{}{\par\smallskip\noindent\begin{minipage}{0.96\textwidth}\footnotesize #3\end{minipage}}}

\usetikzlibrary{arrows,shapes.arrows,shapes.geometric,shapes.multipart,decorations.pathmorphing,positioning}
\tikzstyle{startstop} = [rectangle,rounded corners,minimum width=3cm,minimum height=1cm,text centered,draw=black]
\tikzstyle{process} = [rectangle,minimum width=3cm,minimum height=1cm,text centered,draw=black]
\tikzstyle{decision} = [diamond,minimum width=3cm,minimum height=1cm,text centered,draw=black,aspect=2]
\tikzstyle{arrow} = [thick,->,>=stealth]

\title{Effort-Centric Fairness in Lending Decisions}
\author{
Shiqi Fang \quad
Zexun Chen\thanks{Corresponding author. Email: \texttt{Zexun.Chen@ed.ac.uk}} \quad
Jake Ansell\\
University of Edinburgh Business School, UK
}
\date{}

\begin{document}
\maketitle

\begin{abstract}
Algorithmic credit scoring must satisfy fairness and explanation requirements, yet prevailing predictive-parity criteria assess only outcomes at the decision point. They can therefore overlook whether rejected applicants face unequal burdens in reaching future approval, a phenomenon we call masked inequality. We develop an effort-centric framework that measures an applicant's effort as the minimum weighted cost of feasible changes required to cross the approval boundary. The framework distinguishes feature-independent actions from additive structural shifts that propagate through a causal model and defines parity by comparing average minimum effort across protected groups. We derive tractable local expressions for general differentiable classifiers and exact expressions for logistic regression, embed them in an in-processing fairness objective, and bound changes in portfolio credit risk. The same optimisation yields actionable pathways to approval. Using mortgage data with continuous and discrete features, we find that rejected female applicants require greater effort even when standard predictive-parity criteria are satisfied. Feature-independent regularisation reduces the effort gap by more than 50\% with modest predictive changes. Causal regularisation yields reductions above 90\% at the tested positive penalty weights, but with larger predictive and risk-return trade-offs. Expected and unexpected losses remain broadly stable under feature-independent regularisation and increase under causal regularisation; RAROC declines but remains positive. These results show that effort parity complements predictive fairness by revealing and mitigating hidden barriers to future credit access while making the associated operational trade-offs explicit.
\end{abstract}


\section{Introduction}
\label{sec:introduction}

Machine learning (ML) is increasingly central to credit scoring, enabling lenders to evaluate default likelihood and make more informed credit decisions. However, growing concerns about fairness have accompanied these advances.
The European Union's Artificial Intelligence Act classifies creditworthiness assessment as a high-risk application, requiring that automated systems undergo rigorous fairness evaluation before deployment and throughout their operational lifecycle \citep{act2024eu}. Similar regulatory pressures are emerging globally; the U.S. Consumer Financial Protection Bureau scrutinises algorithmic lending for disparate impact, whilst the UK Financial Conduct Authority has issued guidance on fair treatment in automated decisions \citep{uk2021guidance}. These developments reflect a fundamental tension in modern credit markets; ML that optimises predictive accuracy may systematically disadvantage protected groups, raising both legal liability and reputational risk for financial institutions.

A parallel and increasingly intertwined regulatory concern is the explainability of automated credit decisions. In the United States, the Equal Credit Opportunity Act (ECOA) requires creditors who deny an application to furnish an adverse action notice stating the specific, principal reasons for the denial \citep{ecoa1974, regb}. The European Union's General Data Protection Regulation (GDPR) imposes a comparable duty, granting data subjects a right to meaningful information about the logic of automated decisions that produce legal or similarly significant effects \citep{gdpr2016}, a provision that has since motivated a substantial legal and technical literature on explanation \citep{ballegeer2025evaluating, tu2025inherently}.

These two regulatory commitments, to fair outcomes and to explainable decisions, do not automatically guarantee one another.
The prevailing approach to algorithmic fairness in credit scoring focuses on \emph{predictive parity}, ensuring that model predictions or outcomes are equitably distributed across demographic groups.
Prominent criteria include statistical parity \citep{kamishima2012fairness} and equalised odds \citep{hardt_equality_2016}. These metrics have been widely adopted in both academic research \citep{kozodoi_fairness_2022, hurlin2024fairness} and industry practice, forming the basis for regulatory compliance assessments.
Recent work in management science has examined the strategic implications of such fairness constraints, demonstrating that these fairness requirements can reduce firms' incentives to invest in model accuracy and ultimately can harm the groups they are intended to protect \citep{fu2022fair}.
Yet predictive parity captures only a static notion of fairness: whether loan acceptances are equitably distributed at the moment of decision.
It neglects a prospective question that the disclosure requirements themselves implicitly raise. By obliging lenders to state the principal reasons for a denial, those rules presuppose that a rejected applicant can identify what would need to change to obtain approval; they are silent on whether that route back into approval is comparably burdensome across groups.
This prospective dimension can give rise to a subtle form of disparate impact: a facially neutral scoring rule, one that makes no explicit use of protected characteristics and may even satisfy predictive parity, can nonetheless impose a systematically more burdensome recourse process on a protected group. Under Regulation B, the implementing regulation of the ECOA, such practices may be challenged under the effects test when the resulting burden falls disproportionately on a prohibited basis and cannot be justified by business necessity, particularly where a less discriminatory alternative exists \citep{regb}. The framework developed in this paper supplies the operational counterpart of this principle by measuring recourse burden across groups and directly searching for less discriminatory models.

\zc{In lending, rejection is often followed by reapplication after applicants adjust financial attributes such as income, leverage, or debt ratios. Consider two otherwise comparable rejected applicants: if one must increase income by \$200 to reach approval while the other must increase it by \$400, a systematic association between that burden and protected-group membership creates what we term \emph{masked inequality}. Credit decisions are also recurrent. Applicants facing larger burdens are less likely to reverse a denial, and denial may restrict the liquidity needed to improve creditworthiness. Because lending decisions affect the observations available for subsequent model development, effort gaps can persist or compound while remaining invisible to static predictive-parity audits.}

Given these dynamics, masked inequality has significant implications for lenders, regulators, and consumers. For \textbf{lenders}, effort disparities create reputational exposure and may suppress conversion among potentially creditworthy borrowers who are discouraged by infeasible improvement requirements. For \textbf{regulators}, a model that passes predictive-parity audits but systematically requires protected groups to exert greater effort to reverse an adverse decision remains, by construction, undetected by the very audits designed to certify its fairness; the group-level effort gap constitutes a distinct dimension of disparate impact that current compliance frameworks are not equipped to measure. For \textbf{consumers}, particularly those from protected groups who already face structural disadvantages in labour and housing markets, disproportionate effort requirements compound existing inequalities and may transform temporary denial into permanent exclusion.

\zc{We propose an \emph{effort-centric} fairness framework that measures the minimum weighted cost of feasible changes required for each rejected applicant to obtain approval. It distinguishes feature-independent changes from causal actions and compares their burdens across protected groups. Because these applicant-level costs depend on the classifier's evolving decision boundary, direct integration into training produces a computationally prohibitive bi-level problem. We derive closed-form local surrogates for general differentiable classifiers and exact weighted-distance expressions for logistic regression, reducing the problem to a tractable single-level objective. The same optimisation identifies concrete pathways that rejected applicants can follow towards approval.}

Using real mortgage data with a mix of continuous and discrete features, our empirical analysis yields four main findings.
First, substantial effort disparities persist even when models satisfy standard predictive parity criteria; rejected female applicants are systematically farther from the approval boundary and therefore face higher weighted effort.
Second, feature-independent effort regularisation reduces disparities by more than 50\% with modest predictive changes, while causal effort regularisation produces larger fairness gains with a clearer predictive trade-off.
Third, the financial effects are finite and explicitly measured: feature-independent regularisation leaves loss measures broadly stable, whereas causal regularisation raises them relative to baseline; RAROC declines but remains positive in both cases.
Fourth, optimising for effort fairness often improves predictive-parity metrics as a by-product, suggesting that effort-based objectives can complement traditional fairness criteria.

\zc{Our contributions are fourfold. First, we formalise feature-independent and causal effort parity as complements to predictive parity and connect group-level diagnostics to applicant-level pathways. Second, we develop a tractable in-processing method based on local analytical surrogates, with an exact logistic-regression specialisation. Third, we derive finite bounds on changes in the considered credit-risk measures under the stated regularity conditions. Fourth, using mortgage data, we document masked inequality under predictive parity and quantify the fairness, predictive, and risk-return consequences of mitigating it.}

The remainder of this paper is organised as follows. Section~\ref{sec:literature} reviews related work on algorithmic fairness and credit risk modelling. Section~\ref{sec:framework} introduces the effort-centric fairness framework, defining effort-based measures of inequality. Section~\ref{sec:approximate_learning} develops tractable approximations that embed effort parity constraints into gradient-based model training. Section~\ref{sec:risk_bounds} demonstrates the controllability of the associated credit risk, and Section~\ref{sec:discrete_recourse} shows how the framework yields actionable explanations for rejected borrowers. Section~\ref{sec:experiment} reports the empirical analysis. Section~\ref{sec:discussions} discusses managerial and regulatory implications, and Section~\ref{sec:conclusions} concludes.

\section{Related Literature}
\label{sec:literature}

Our work intersects three streams of research: algorithmic fairness in management science, ML for credit scoring, and counterfactual explanations for recourse.

\subsection{Algorithmic Fairness in Management Science}

The growing deployment of ML in consequential decisions has generated substantial interest in algorithmic fairness across management and operations research. \citet{de2022algorithmic} provides a comprehensive review oriented toward business analytics, identifying legal compliance, social responsibility, and organisational performance as key motivations for fairness and charting research directions for the field. They emphasise that unfair systems threaten not only societal welfare but also firms' competitiveness and survival.

Within management science, a series of papers examines the strategic implications of fairness constraints. \citet{fu2022fair} demonstrate that ``fair'' ML algorithms requiring impact parity can paradoxically make everyone worse off, including the protected group, by reducing firms' incentives to invest in model accuracy. \citet{cohen2022price} integrates fairness notions into explicit economic models of credit markets to study interest rate setting and to characterise the welfare implications of different fairness constraints in a competitive credit market. Building on this foundation, \citet{cohen2025dynamic} extends the analysis to a dynamic setting in which loan demand and the composition of the applicant pool evolve over time in response to lending policies, showing how static fairness constraints can have unintended long-run distributional consequences.

Other work addresses fairness in specific operational contexts. \citet{ganju2020role} demonstrates that decision support systems can attenuate racial biases in healthcare delivery. \citet{rea2021unequal} develops operational allocation models incorporating distributive justice, showing that unequal but fair allocations can be optimal under explicit equity considerations. \citet{lambrecht2019algorithmic} provides empirical evidence on gender-based discrimination in algorithmic advertising, illustrating how neutral optimisation can produce disparate outcomes.

A complementary literature examines fairness from data and methodological perspectives. \citet{zhang2018assessing} addresses the practical challenge that race and ethnicity are typically unobserved in mortgage data, evaluating proxy methods such as Bayesian Improved Surname Geocoding (BISG) and demonstrating that proxy choice can substantially influence the estimation of pricing disparities. \citet{kallus2022assessing} generalises this problem, deriving sharp bounds on fairness measures when protected attributes are unobserved in the primary dataset but available in auxiliary data. Beyond measurement, \citet{hu2025human} investigates how algorithmic systems trained on historical human decisions can simultaneously mitigate and propagate bias, showing that outcomes depend critically on market conditions and the feedback loop between algorithmic outputs and future training data. Together, these papers highlight that fairness assessment requires attention not only to the choice of criterion but also to data provenance, proxy quality, and human-algorithm dynamics.

\subsection{Machine Learning in Credit Scoring}

Credit scoring has long served as a testbed for predictive modelling, with ML methods offering improved accuracy over traditional scorecards \citep{lessmann2015benchmarking, baesens2003benchmarking}. \citet{desai1996comparison} shows that multilayer neural networks deliver modest but systematic gains in discriminatory power over linear scoring models in credit union settings, whilst  \citet{west2000neural} compares several neural architectures and finds that carefully tuned networks generally outperform logistic regression. Synthesising subsequent developments, \citet{dastile2020statistical} reviews statistical and ML approaches and concludes that ensemble-based classifiers tend to dominate traditional techniques when combined with appropriate preprocessing.

These predictive advances intensify regulatory concerns around fairness.
\citet{fuster2022predictably} embed ML credit risk models in an equilibrium mortgage market and show that, whilst ML improves prediction and expands credit overall, it can simultaneously widen interest rate disparities between racial groups.
\citet{kozodoi_fairness_2022} benchmark fairness notions and mitigation strategies on credit data, quantifying the frontier among accuracy, profit, and fairness and showing that different fairness definitions induce markedly different profitability and inclusion outcomes.
\citet{hurlin2024fairness} formalise statistical tests for disparate treatment and disparate impact, introducing Fairness Partial Dependence Plots to identify variables driving detected unfairness.

Interpretability has also received increasing attention. \citet{de2024explainable} surveys explainable AI for operational research, covering inherently interpretable models, transparency-enhancing feature engineering, and visual analytics. \citet{bucker2022transparency} provides a comprehensive framework, integrating SHAP, LIME, BreakDown, and PDPs, tailoring global and local explanations to different stages (e.g., application, behavioural scoring, collections) of the credit scoring process and stakeholder (clients, loan officers, regulators) needs.

However, this literature focuses exclusively on predictive parity, that is, equalising outcomes or error rates at the point of decision, leaving unexplored whether rejected applicants face equitable opportunities for \emph{future} approval. A model satisfying statistical parity may impose systematically higher improvement burdens on one group, a disparity invisible to existing fairness diagnostics.

\subsection{Counterfactual Explanations and Algorithmic Recourse}

Counterfactual explanations identify minimal changes to input features that would alter a model's prediction, providing ``what-if'' scenarios for rejected applicants \citep{wachter2018counterfactual}. Recent work formalises their generation as tractable optimisation problems \citep{kurtz2025counterfactual, de2024model, bogetoft2024counterfactual, dastile2022model} and surveys the resulting methods \citep{verma2024counterfactual,karimi2022survey}. In credit scoring, \citet{dastile2022model} proposes a model-agnostic evolutionary approach that yields compact counterfactuals for both loan rejections and approvals, while \citet{carrizosa2024generating} develops collective counterfactuals, aligning recommendations with batch decision workflows in lending \citep{carrizosa2024mathematical}.

A key development is the distinction between counterfactual explanations and algorithmic recourse. \citet{karimi2021algorithmic} argue that counterfactual explanations tell individuals where they need to reach but not how to get there, and that causal reasoning is essential for generating recommendations that achieve the desired outcome when acted upon. Building on this insight, \citet{karimi2020algorithmic} develops methods for causal recourse under imperfect causal knowledge using a probabilistic approach.
Several papers connect recourse to fairness considerations. \citet{ustun2019actionable} provide a canonical formulation of actionable recourse in linear classification, \citet{heidari2019long} formalise effort unfairness and its long-run consequences, and \citet{hu2019disparate} show that groups can face different costs of strategic adaptation. More directly, \citet{gupta2019equalizing} train classifiers with recourse-equalisation constraints, \citet{ross2021learning} learn models that support actionable recourse, and \citet{von2022fairness} introduce fairness criteria for causal algorithmic recourse. These studies establish that recourse can be embedded in model development rather than treated only as a post-hoc explanation problem. What remains comparatively underdeveloped is an in-processing framework tailored to credit scoring that combines effort-parity training with mixed credit features, a tractable approximation to the applicant-level recourse problem, and an explicit analysis of credit risk and profitability.

\subsection{Positioning and Contribution}

Our work addresses this credit-scoring gap. First, we formalise feature-independent and causal effort parity for rejected credit applicants, linking recourse fairness to the operational setting in which lenders must evaluate mutable financial attributes, immutable applicant characteristics, and dependent credit features. The formalisation elegantly delivers both fairness and explainability. Second, we derive local closed-form surrogates for general differentiable classifiers and exact weighted-distance expressions for our model, collapsing the applicant-level bi-level recourse problem into a single-level formulation. Third, we embed the applicable closed forms in model training and provide empirical evidence on mortgage credit data, documenting both the existence of masked inequality under predictive parity and the effectiveness of our mitigation approach. Fourth, we connect effort-parity training to credit-risk management by deriving risk bounds and reporting expected loss, unexpected loss, interest revenue, and RAROC. Thus, relative to prior recourse-fairness work, the contribution is the integration of effort parity, tractable closed-form training objectives, and risk-return analysis in credit scoring.

\section{The Effort-Centric Fairness Framework}
\label{sec:framework}

This section develops our effort-centric approach to fairness in credit scoring. We first establish notation and then demonstrate why predictive parity criteria, whilst valuable, fail to ensure equitable opportunities for future loan approval. We then formalise effort quantification and define effort parity.

\subsection{Setup and Notation}
\label{subsec:setup}
Let $\bmX = (X_1, \ldots, X_d)$ denote a $d$-dimensional random vector of observed features, taking values $\bmx = (x_1, \ldots, x_d) \in \mathcal{X} \subseteq \mathbb{R}^d$. These features typically include applicant characteristics such as income and debt ratio. Let $S \in \{0,1\}$ denote a binary protected attribute representing group membership (e.g., gender), where $S=0$ indicates membership in a protected group. We write $\bmV = (\bmX, S)$ for the full feature vector including the protected attribute, with realisations $\bmv = (\bmx, s)$. Let $Y \in \{0,1\}$ represent the binary outcome, where $Y=1$ indicates loan repayment and $Y=0$ indicates default.

We consider a training dataset $\mathcal{D} = \{(\bmv^i, y^i)\}_{i=1}^N$ consisting of $N$ observations. Our objective is to learn a scoring function $\hat{h}: \mathcal{X} \to [0,1]$ parametrised by $\bmtheta$, which outputs a continuous score reflecting the predicted probability of repayment. A binary decision is obtained by thresholding: $\hat{y} = \ind\{\hat{h}(\bmx; \bmtheta) \geq \tau\}$, where $\tau \in (0,1)$ is the decision threshold and $\ind\{\cdot\}$ denotes the indicator function. Applicants with $\hat{h}(\bmx; \bmtheta) \geq \tau$ are approved; those below are rejected.

In practice, features are not independent; changing one feature (e.g., income) may causally affect others (e.g., debt ratio). We model these dependencies using a linear structural causal model (SCM), which encodes the data-generating process through a system of structural equations. Let $\bmA \in \mathbb{R}^{d \times d}$ denote the \emph{structural matrix}, where entry $A_{ij}$ quantifies the direct causal effect of $X_j$ on $X_i$, with $A_{ij} = 0$ if $X_j$ does not directly cause $X_i$. The propagation matrix $(\mathbb{I}_d - \bmA)^{-1}$ captures the total effect of an additive shift to a structural equation, including all downstream pathways; its $(i,j)$-th entry is the total change in $X_i$ induced by a unit additive shift to the equation for $X_j$. The formal definition of SCMs and the derivation of causal effects are provided in Appendix~\ref{sec:appendix_scm}.

\subsection{Effort Quantification}
\label{subsec:effort_quantification}

Standard fairness criteria in credit scoring, such as statistical parity (equal approval rates across groups), equalised odds (equal error rates conditional on a positive outcome), and positive predictive value parity (equal precision amongst those approved), focus on equalising predictions or outcomes at the point of decision. As argued in Section~\ref{sec:introduction}, satisfying these criteria does not preclude systematic disparities in the effort required for rejected applicants to achieve future approval. We therefore develop an effort-based framework that directly addresses this gap. For completeness, formal definitions of these standard criteria used as benchmarks in our empirical analysis in Section~\ref{subsec:predictive_fairness_limitations} are provided in Appendix~\ref{sec:sm_pp_effort}.

For effort evaluation, we require a consistent basis for comparing the burden faced by different groups. The appropriate benchmark is the \emph{minimal effort} required to achieve approval, as this provides a fair basis for cross-group comparison; if one group requires systematically higher minimal effort than another, this reveals inequality in access to credit.
Let $\mathcal{D}^- = \{\bmv^i \mid (\bmv^i, y^i) \in \mathcal{D}, \; \hat{h}(\bmx^i; \bmtheta) < \tau\}$ denote the set of rejected applicants. For a rejected applicant with features $\bmx^{\texto}$, we seek the minimum-cost set of feasible changes that would result in approval.

\subsubsection{Feature-Independent Effort}

Let $\bm{\delta} \in \mathbb{R}^d$ denote a vector of changes applied to the factual features $\bmx^{\texto}$. The \emph{minimal feature-independent effort} is defined as:
\begin{align}\label{eq:feature_effort}
c^*(\bmx^{\texto}) = & \min_{\bm{\delta} \in \mathcal{F}(\bmx^{\texto})} \; \mathrm{cost}(\bm{\delta}; \bmx^{\texto}) \\
\text{s.t.} & \quad \hat{h}(\bmx^{\texto} \oplus \bm{\delta}; \bmtheta) \geq \tau, \nonumber
\end{align}
where $\mathrm{cost}(\cdot): \mathbb{R}^d \to \mathbb{R}^{+}$ measures the effort associated with implementing changes $\bm{\delta}$, and $\mathcal{F}(\bmx^{\texto})$ denotes the feasible set incorporating mutability constraints (e.g., age cannot change), boundedness (features have realistic ranges), and discreteness (some features take integer values).
We adopt a weighted Euclidean norm $\mathrm{cost}(\bm{\delta}; \bmx) = \|\bm{\delta}\|_{\bmW} = \|\bmW^{1/2}\bm{\delta}\|_2$, where the positive definite weight matrix $\bmW$ encodes feature-specific costs. Immutable features receive prohibitively large weights.
The operator $\bmx^{\texto} \oplus \bm{\delta}$ applies addition to continuous and ordinal features, with $\delta_j\in\mathbb Z$ for an ordinal feature, and value replacement to categorical features. This formulation assumes that features can be modified independently. In reality, interventions often propagate through causal relationships.

\subsubsection{Causal Effort}

When an applicant takes action to change certain features, downstream features may change as a consequence. These downstream changes may reinforce, leave unchanged, or offset the direct effect of an intervention on the approval score. To account for such dependencies, we define effort in terms of \emph{causal interventions} rather than direct feature changes.

Let $\bmxi_{\bmX} \in \mathbb{R}^d$ denote a vector of \emph{additive structural shifts} to the equations for the non-protected features, with zero entries for features on which no direct action is taken. Writing $\mathbf{P}:=(\mathbb{I}_d-\bmA)^{-1}$, the factual and shifted systems are $\bmX=\bmA\bmX+\mathbf B$ and $\bmX^{\bmxi}=\bmA\bmX^{\bmxi}+\mathbf B+\bmxi_{\bmX}$, respectively. Hence the induced total feature change is $\bm{\delta}=\bmX^{\bmxi}-\bmX=\mathbf P\bmxi_{\bmX}$. This is a soft, shift intervention: it modifies the relevant structural mechanisms but does not delete their incoming edges. It is therefore distinct from a hard intervention $do(X_j=x_j')$, whose propagation operator generally depends on the set of intervened nodes. For a factual profile $\bmv^{\texto}=(\bmx^{\texto},s^{\texto})$, define the corresponding action-based counterfactual profile as $\bmv_{\bmxi}^{\texto}:=(\bmx^{\texto}+\mathbf P\bmxi_{\bmX},s^{\texto})$. Appendix~\ref{sec:appendix_scm} gives the formal distinction and derivation.

The \emph{minimal causal effort} is therefore:
\begin{align}\label{eq:causal_effort}
r^*(\bmv^{\texto}) = & \min_{\bmxi_{\bmX} \in \mathcal{F}(\bmx^{\texto})} \; \mathrm{cost}(\bmxi_{\bmX}; \bmx^{\texto}) \\
\text{s.t.} & \quad \hat{h}\!\left(\bmx^{\texto}+\mathbf P\bmxi_{\bmX};\bmtheta\right) \geq \tau. \nonumber
\end{align}
Here $\mathcal F(\bmx^{\texto})$ constrains the direct shifts and the resulting profile, including mutability, direction, and range restrictions.

The key distinction between Eqs.~\eqref{eq:feature_effort} and \eqref{eq:causal_effort} is that causal effort charges only for the direct structural shifts $\bmxi_{\bmX}$, recognising that components of $\mathbf P\bmxi_{\bmX}$ occur as downstream consequences. Since those consequences may help or hinder movement towards approval, causal and feature-independent effort do not have an unconditional ordering.

\subsection{Effort Parity Definitions}
\label{subsec:effort_parity}

Having formalised effort quantification, we now define fairness criteria based on equalising effort. We develop two complementary notions: feature-independent effort parity and causal effort parity.

\subsubsection{Feature-Independent Effort Parity}

For each protected group $s \in \{0,1\}$, define $\Omega_s^- = \{\bmv^i \in \mathcal{D}^- \mid s^i = s\}$ as the subset of rejected applicants belonging to group $s$. The \emph{average feature-independent effort} for group $s$ is:
\begin{equation}
\bar{c}^*(\Omega_s^-) = \frac{1}{|\Omega_s^-|} \sum_{\bmv^i \in \Omega_s^-} c^*(\bmx^i).
\end{equation}

\begin{definition}[Feature-Independent Effort Parity]
\label{def:fi_parity}
A classifier satisfies \emph{feature-independent effort parity} if the average minimal efforts are equal across protected groups:
\begin{equation}
\bar{c}^*(\Omega_0^-) = \bar{c}^*(\Omega_1^-).
\end{equation}
\end{definition}

Violations of this criterion indicate that one group faces systematically higher barriers to achieving approval. We quantify such violations via the \emph{feature-independent effort disparity}:
\begin{equation}\label{eq:fi_disparity}
\Delta^{\text{FI}}(\bmtheta) = \left| \bar{c}^*(\Omega_0^-) - \bar{c}^*(\Omega_1^-) \right|.
\end{equation}

\subsubsection{Causal Effort Parity}

Analogously, the \emph{average causal effort} for group $s$ is:
\begin{equation}
\bar{r}^*(\Omega_s^-) = \frac{1}{|\Omega_s^-|} \sum_{\bmv^i \in \Omega_s^-} r^*(\bmv^i).
\end{equation}

\begin{definition}[Causal Effort Parity]
\label{def:causal_group_parity}
A classifier satisfies \emph{causal effort parity} if the average minimal causal efforts are equal across protected groups:
\begin{equation}
\bar{r}^*(\Omega_0^-) = \bar{r}^*(\Omega_1^-).
\end{equation}
\end{definition}

The corresponding \emph{causal effort disparity} is:
\begin{equation}\label{eq:gc_disparity}
\Delta^{\text{GC}}(\bmtheta) = \left| \bar{r}^*(\Omega_0^-) - \bar{r}^*(\Omega_1^-) \right|.
\end{equation}

\subsubsection{Why Equalise Effort?}
An effort gap is not by itself evidence of unfair treatment: rejected applicants in one group may lie farther from a risk-based boundary because of repayment-relevant characteristics. We therefore treat effort disparity as a diagnostic, asking whether the burden can be reduced without materially weakening legitimate underwriting objectives. Effort parity complements, rather than replaces, risk-based assessment and should therefore be evaluated jointly with predictive performance and portfolio risk.

\subsubsection{Integration into Learning}

The disparity metrics defined above serve dual purposes; regulators can use them to audit deployed models for masked inequality, whilst lenders use them as training objectives to mitigate disparities during model development. We focus on the latter.

The fairness literature distinguishes three approaches to unfairness mitigation \citep{kozodoi_fairness_2022}; pre-processing methods modify training data before model fitting \citep{calders_building_2009, kamiran_classifying_2009}; post-processing methods adjust predictions after model training \citep{hardt_equality_2016, pleiss_fairness_2017}; and in-processing methods embed fairness constraints directly into the learning objective \citep{kamishima2012fairness, celis_classification_2019}. We adopt an in-processing approach for two reasons. First, effort disparity depends jointly on the decision boundary and feature distributions near that boundary; pre-processing methods that transform data without knowledge of the eventual classifier cannot directly target this joint dependence. Second, post-processing methods adjust only predictions, leaving the underlying decision boundary unchanged; since effort is determined by the distance to this boundary, post-hoc prediction adjustments cannot reduce effort disparities without retraining the model.

Accordingly, we augment the standard predictive loss with a fairness regularisation term:
\begin{equation}\label{eq:regularised_learning}
\min_{\bmtheta} \; \mathcal{L}_{\text{acc}}(\bmtheta) + \lambda \cdot \mathcal{L}_{\text{fair}}(\bmtheta),
\end{equation}
where $\mathcal{L}_{\text{acc}}(\bm{\theta}) = \mathbb{E}_{(\bmv, y) \sim
\mathcal{D}}[\ell(y, \hat{h}(\bmx; \bm{\theta}))] + \mu \|\bm{\theta}\|^2 /2$
is the ridge-regularised predictive loss (e.g., cross-entropy with weight
decay), with $\mu>0$ controlling the strength of $\ell_2$ regularisation. When the data-fit term is convex and the ridge penalty applies to all components of $\bmtheta$, $\mathcal L_{\mathrm{acc}}$ is $\mu$-strongly convex. $\lambda > 0$ controls the trade-off between predictive performance and fairness, and the fairness loss $\mathcal{L}_{\text{fair}}(\bmtheta)$ is instantiated using one of the disparity metrics:
\begin{equation}
\mathcal{L}_{\text{fair}}(\bmtheta) \in \left\{ \Delta^{\text{FI}}(\bmtheta), \; \Delta^{\text{GC}}(\bmtheta) \right\}.
\end{equation}

A key challenge in Eq.~\eqref{eq:regularised_learning} is that the effort terms $c^*(\cdot)$ and $r^*(\cdot)$ are themselves solutions to optimisation problems that depend on $\bmtheta$. This creates a computationally intractable bi-level structure.

\section{Efficient Learning with Effort Parity Constraints}
\label{sec:approximate_learning}

This section develops local analytical surrogates, together with exact expressions for logistic regression, that collapse the bi-level structure into a tractable single-level formulation.

\subsection{Computational Challenge}
\label{subsec:computational_challenge}

To illustrate the computational challenge, consider the fairness loss based on feature-independent effort disparity. Substituting the definition of $\Delta^{\text{FI}}(\bmtheta)$ from Eq.~\eqref{eq:fi_disparity} into the learning objective yields a bi-level optimisation problem:
\begin{align}\label{eq:bilevel}
\min_{\bmtheta} \quad & \mathcal{L}_{\text{acc}}(\bmtheta) + \lambda \left| \bar{c}^*(\Omega_0^-) - \bar{c}^*(\Omega_1^-) \right| \\
\text{where} \quad & c^*(\bmx^i) = \min_{\bm{\delta} \in \mathcal{F}(\bmx^i)} \|\bm{\delta}\|_{\bmW} \nonumber \\
\text{s.t.} &\quad \hat{h}(\bmx^i \oplus \bm{\delta}; \bmtheta) \geq \tau, \quad \forall \, \bmv^i \in \mathcal{D}^-. \nonumber
\end{align}
The causal-effort formulation is analogous, with the corresponding causal effort and group aggregation.

The outer level optimises the model parameters $\bmtheta$, whilst the inner level solves individual-specific recourse problems for every rejected applicant. Two features make this problem computationally prohibitive. First, the rejection set $\mathcal{D}^-$ depends on $\bmtheta$, as the decision boundary shifts during training and applicants move between approved and rejected status. Second, the inner optimisation problems are coupled to the outer problem through $\bmtheta$, which determines both the decision boundary and the gradient landscape.
Solving this problem naively, by recomputing optimal interventions at each gradient step, is infeasible for practical training.

\subsection{Local Surrogates and Exact Effort}
\label{subsec:analytical_approximation}

We formalise the effort through three assumptions.

\begin{assumption}[Local Linearity]
\label{ass:local_linearity}
For each considered rejected applicant, $\hat h(\cdot;\bmtheta)$ is differentiable near $\bmx^i$ and satisfies $\hat h(\bmx^i+\bm\delta;\bmtheta)=\hat h(\bmx^i;\bmtheta)+\nabla_{\bmx}\hat h(\bmx^i;\bmtheta)^\top\bm\delta+o(\|\bm\delta\|_2)$.
\end{assumption}

The local surrogate drops the $o(\|\bm\delta\|_2)$ term. It is exact for affine scores; logistic regression has a smooth input gradient and admits the same local expansion.

\begin{assumption}[Feasible Local Actions]
\label{ass:boundary_concentration}
Each relevant exact effort problem attains a minimum. Its feasible action set is star-shaped about zero, and the unconstrained linearised minimiser used below is feasible and lies in the neighbourhood of Assumption~\ref{ass:local_linearity}.
\end{assumption}

\begin{assumption}[Continuous Action Space]
\label{ass:continuous}
For the local closed forms, directly changed features are continuous and $\bmW\succ0$. Discrete and categorical actions are addressed separately in Section~\ref{sec:discrete_recourse}.
\end{assumption}

We first establish an exact property of the original constrained problem and then solve its unconstrained local linearisation.

\begin{lemma}[Optimality of Boundary Crossing]
\label{lem:boundary_crossing}
Under Assumptions~\ref{ass:local_linearity}--\ref{ass:continuous}, the optimal intervention $\bm{\delta}^*$ for a rejected applicant $\bmx^{\texto}$ satisfies the decision boundary constraint with equality, $\hat{h}(\bmx^{\texto}+\bm{\delta}^*;\bmtheta)=\tau$.
\end{lemma}

The proof, provided in Appendix~\ref{sec:appendix_proofs}, uses continuity, star-shaped feasibility, and positive homogeneity of the weighted norm; it does not require the classifier to be globally monotone.

\begin{theorem}[Feature-Independent Local Effort]
\label{thm:fi_approx}
Under Assumptions~\ref{ass:local_linearity} and~\ref{ass:continuous}, for a rejected applicant let $q^i:=\tau-\hat h(\bmx^i;\bmtheta)>0$ and $\bmg^i:=\nabla_{\bmx}\hat h(\bmx^i;\bmtheta)\neq\mathbf0$. The unconstrained first-order surrogate of Eq.~\eqref{eq:feature_effort} has value
\begin{equation}\label{eq:fi_effort_approx}
\tilde c(\bmx^i;\bmtheta)
=\frac{q^i}{\|\bmW^{-1/2}\bmg^i\|_2},
\end{equation}
and its unique minimum-norm action is
\begin{equation}\label{eq:fi_direction}
\tilde{\bm\delta}^{\,i}
=\tilde c(\bmx^i;\bmtheta)
\frac{\bmW^{-1}\bmg^i}{\|\bmW^{-1/2}\bmg^i\|_2}.
\end{equation}
When $\tilde{\bm\delta}^{\,i}$ is feasible and local as required by Assumption~\ref{ass:boundary_concentration}, Eqs.~\eqref{eq:fi_effort_approx}--\eqref{eq:fi_direction} approximate $c^*(\bmx^i)$ and its optimiser.
\end{theorem}

The proof is provided in Appendix~\ref{sec:appendix_proofs}. The local surrogate is proportional to the score gap and inversely proportional to the weighted gradient magnitude. Importantly, it is the exact solution of the stated linearised problem, not an assertion that the original constrained problem always has this closed form.

We next extend the local surrogate to causal effort, which accounts for downstream propagation of additive structural shifts.

\begin{theorem}[Causal Local Effort]
\label{thm:causal_approx}
Under the notation and conditions of Theorem~\ref{thm:fi_approx}, let $\mathbf P=(\mathbb I_d-\bmA)^{-1}$. The unconstrained first-order surrogate of Eq.~\eqref{eq:causal_effort} has value
\begin{equation}\label{eq:causal_effort_approx}
\tilde r(\bmv^i;\bmtheta)
=\frac{q^i}{\|\bmW^{-1/2}\mathbf P^{\top}\bmg^i\|_2},
\end{equation}
and its unique minimum-norm direct shift is
\begin{equation}\label{eq:causal_direction}
\tilde{\bmxi}_{\bmX}^{\,i}
=\tilde r(\bmv^i;\bmtheta)
\frac{\bmW^{-1}\mathbf P^{\top}\bmg^i}
{\|\bmW^{-1/2}\mathbf P^{\top}\bmg^i\|_2}.
\end{equation}
When this direct shift is feasible and local as required by Assumption~\ref{ass:boundary_concentration}, these expressions approximate $r^*(\bmv^i)$ and its optimiser.
\end{theorem}

The proof is provided in Appendix~\ref{sec:appendix_proofs}. The key difference is the transformed gradient $\mathbf P^{\top}\bmg^i$, which measures the score effect of a unit direct structural shift after downstream propagation. When $\bmA=\mathbf0$, $\mathbf P=\mathbb I_d$ and the two local problems coincide.

\begin{corollary}[Causal Amplification]
\label{cor:causal_amplification}
Under the notation and conditions of Theorem~\ref{thm:causal_approx}, define the causal amplification factor
\begin{equation}
\label{eq:causal_amplification}
\zc{\gamma^i}=
\frac{\|\bmW^{-1/2}\mathbf{P}^{\top}\bmg^i\|_2}
{\|\bmW^{-1/2}\bmg^i\|_2}.
\end{equation}
Then $\tilde{r}(\bmv^i;\bmtheta)=\tilde{c}(\bmx^i;\bmtheta)/\gamma^i$. Consequently, causal propagation lowers, preserves, or raises approximate effort according as $\gamma^i>1$, $\gamma^i=1$, or $\gamma^i<1$. Equivalently,
\begin{equation}
\begin{aligned}
&\tilde{r}(\bmv^i;\bmtheta)\leq\tilde{c}(\bmx^i;\bmtheta)\\[-2pt]
&\quad\Longleftrightarrow\quad {\bmg^i}^{\top}\!\left(
\mathbf{P}\bmW^{-1}\mathbf{P}^{\top}-\bmW^{-1}
\right)\!\bmg^i \geq0.
\end{aligned}
\end{equation}
\end{corollary}

\paragraph{Exact logistic specialisation.}
For the logistic model $\hat h(\bmx;\bmtheta)=\sigma(\bmw^\top\bmx+b)$, let $z^i:=\bmw^\top\bmx^i+b$, $\tau_{\mathrm{logit}}:=\log\{\tau/(1-\tau)\}$, and $[x]_+:=\max\{x,0\}$. Since the logistic link is strictly increasing, approval is equivalent to $\bmw^\top\bmx+b\ge\tau_{\mathrm{logit}}$. Thus, when the following unconstrained solutions satisfy the applicable feasible sets, the feature-independent and additive-shift causal efforts are exactly
\begin{equation}
\label{eq:exact_logistic_efforts}
c_{\mathrm{log}}^i
=\frac{[\tau_{\mathrm{logit}}-z^i]_+}{\|\bmW^{-1/2}\bmw\|_2},
\qquad
r_{\mathrm{log}}^i
=\frac{[\tau_{\mathrm{logit}}-z^i]_+}{\|\bmW^{-1/2}\mathbf P^\top\bmw\|_2}.
\end{equation}
Their minimum-norm actions are
\begin{align}
\label{eq:exact_logistic_directions}
\bm\delta_{\mathrm{log}}^i
&=\frac{[\tau_{\mathrm{logit}}-z^i]_+}{\|\bmW^{-1/2}\bmw\|_2^2}\bmW^{-1}\bmw,\\
\bmxi_{\bmX,\mathrm{log}}^i
&=\frac{[\tau_{\mathrm{logit}}-z^i]_+}{\|\bmW^{-1/2}\mathbf P^\top\bmw\|_2^2}
\bmW^{-1}\mathbf P^\top\bmw.
\end{align}
With common $\bmA$ and $\bmW$, define $\gamma(\bmw):=\|\bmW^{-1/2}\mathbf P^\top\bmw\|_2/\|\bmW^{-1/2}\bmw\|_2$. Let $\Delta_{\mathrm{log}}^{\mathrm{FI}}$ and $\Delta_{\mathrm{log}}^{\mathrm{GC}}$ denote Eqs.~\eqref{eq:fi_disparity} and~\eqref{eq:gc_disparity}, respectively, evaluated using $c_{\mathrm{log}}^i$ and $r_{\mathrm{log}}^i$. The exact logistic efforts and their group disparities then satisfy
\begin{equation}
\label{eq:logistic_effort_scaling}
r_{\mathrm{log}}^i=\frac{c_{\mathrm{log}}^i}{\gamma(\bmw)},
\qquad
\Delta_{\mathrm{log}}^{\mathrm{GC}}(\bmtheta)
=\frac{\Delta_{\mathrm{log}}^{\mathrm{FI}}(\bmtheta)}{\gamma(\bmw)}.
\end{equation}
Accordingly, the empirical logistic model can use Eqs.~\eqref{eq:exact_logistic_efforts}--\eqref{eq:exact_logistic_directions} without a Taylor approximation. If a feasibility constraint binds, these expressions remain unconstrained benchmarks and the corresponding constrained convex programme must be solved instead.

\subsection{Integration into Model Training}
\label{subsec:integration_training}

The local surrogates in Theorems~\ref{thm:fi_approx} and~\ref{thm:causal_approx}, and the exact logistic expressions in Eq.~\eqref{eq:exact_logistic_efforts}, are locally Lipschitz in $\bmtheta$ wherever their denominator norms are nonzero. They therefore yield objectives that are differentiable almost everywhere and can be handled by standard first-order methods.

\subsubsection{Soft Disparity Formulation}

Let $e_{\mathrm{FI}}^i(\bmtheta)$ and $e_{\mathrm{GC}}^i(\bmtheta)$ denote the effort quantities used for training. For a general differentiable classifier they are the local surrogates $\tilde c(\bmx^i;\bmtheta)$ and $\tilde r(\bmv^i;\bmtheta)$. For the logistic model used in our empirical analysis, they are instead the exact unconstrained quantities $c_{\mathrm{log}}^i$ and $r_{\mathrm{log}}^i$ from Eq.~\eqref{eq:exact_logistic_efforts}. For example, feature-independent group effort disparity becomes
\begin{equation}\label{eq:diff_fi_disparity}
\smalleq{\tilde{\Delta}^{\text{FI}}(\bmtheta) = \left| \frac{1}{|\Omega_0^-|} \sum_{\bmv^i \in \Omega_0^-} e_{\mathrm{FI}}^i(\bmtheta) - \frac{1}{|\Omega_1^-|} \sum_{\bmv^i \in \Omega_1^-} e_{\mathrm{FI}}^i(\bmtheta) \right|},
\end{equation}
with the analogous substitution of $e_{\mathrm{GC}}^i$ for causal effort disparity.

To avoid discontinuities with respect to the rejection set, we employ a soft rejection indicator. The challenge arises since the rejection set $\mathcal{D}^-$ is defined by a hard threshold; applicant $i$ is rejected if and only if $\hat{h}(\bmx^i; \bmtheta) < \tau$. This creates two problems for gradient-based optimisation. First, the indicator function $\ind\{\hat{h}(\bmx^i; \bmtheta) < \tau\}$ is non-differentiable at the decision boundary, with an undefined gradient. Second, as model parameters $\bmtheta$ evolve during training, applicants may cross the threshold and discretely enter or exit the rejection set, causing discontinuous jumps in the disparity.

\zc{We replace the hard indicator with the soft rejection weight
$\omega^i(\bmtheta):=\sigma\bigl(-\kappa[\hat h(\bmx^i;\bmtheta)-\tau]\bigr)$, where $\kappa>0$ controls sharpness. It approaches one for clear rejection and zero for clear approval. For finite $\kappa$, the resulting disparity is locally Lipschitz and differentiable almost everywhere:}
\begin{equation}\label{eq:soft_disparity}
\smalleq{\tilde{\Delta}^{\text{FI}}(\bmtheta) = \left| \frac{\sum_{i: s^i=0} \omega^i e_{\mathrm{FI}}^i(\bmtheta)}{\sum_{i: s^i=0} \omega^i} - \frac{\sum_{i: s^i=1} \omega^i e_{\mathrm{FI}}^i(\bmtheta)}{\sum_{i: s^i=1} \omega^i} \right|}.
\end{equation}
The causal version, obtained by replacing $e_{\mathrm{FI}}^i$ with $e_{\mathrm{GC}}^i$, is reported in Appendix~\ref{sec:appendix_loss}.

\subsubsection{Training Procedure}

Algorithm~\ref{alg:training} summarises the complete training procedure. The algorithm integrates effort-based fairness into standard mini-batch gradient descent; at each iteration, it computes both the predictive loss and the soft effort disparity, then updates parameters to minimise their weighted combination.
For a locally Lipschitz function $f$, $\partial_C f(\bmtheta)$ denotes its Clarke subdifferential: the convex hull of limiting gradients at differentiability points approaching $\bmtheta$.

\begin{algorithm}[t]
\caption{Effort-Fair Classifier Training}
\label{alg:training}
\begin{algorithmic}[1]
\Require Training data $\mathcal{D}$, regularisation weight $\lambda$, learning rate $\eta$, sharpness parameter $\kappa$, threshold $\tau$, weight matrix $\bmW$, causal matrix $\bmA$ (optional)
\Ensure Trained parameters $\bmtheta$
\State Initialise $\bmtheta$ randomly
\For{each epoch}
    \For{each mini-batch $\cmH \subseteq \mathcal{D}$}
        \State Compute predictive loss: $\mathcal{L}_{\text{acc}} \gets \frac{1}{|\cmH|} \sum_{(\bmv^i, y^i) \in \cmH} \ell(y^i, \hat{h}(\bmx^i; \bmtheta))$
        \State Compute soft rejection weights: $\omega^i \gets \sigma(-\kappa[\hat{h}(\bmx^i; \bmtheta) - \tau])$ for all $i$
        \State Compute efforts: $e_{\mathrm{FI}}^i(\bmtheta)$ or $e_{\mathrm{GC}}^i(\bmtheta)$ for all $i$
        \State Compute soft disparity: $\tilde{\Delta}(\bmtheta)$ via Eq.~\eqref{eq:soft_disparity}
        \State Compute total loss: $\mathcal{L} \gets \mathcal{L}_{\text{acc}} + \lambda \cdot \tilde{\Delta}(\bmtheta)$
        \State Select $\bmzeta \in \partial_C\mathcal{L}(\bmtheta)$ and update $\bmtheta \gets \bmtheta - \eta\bmzeta$
    \EndFor
\EndFor
\State \Return $\bmtheta$
\end{algorithmic}
\end{algorithm}

\paragraph{Computational Complexity.}
The overhead relative to standard training is modest. General differentiable classifiers require the input gradient already available through automatic differentiation. Logistic regression uses $\bmw$ directly in the exact expressions. Causal effort additionally multiplies by $\mathbf P^\top$, which costs $O(d^2)$ and can be precomputed with respect to the fixed causal structure.

The regularisation weight $\lambda$ controls the trade-off between accuracy and effort fairness. We recommend selecting $\lambda$ via cross-validation, monitoring both predictive metrics (AUC, accuracy, F1) and fairness metrics (effort disparity). The sharpness parameter $\kappa$ for the soft rejection indicator should be set large enough to approximate hard thresholding (we use $\kappa = 10$ in experiments) whilst maintaining numerical stability. The decision threshold $\tau$ is typically determined by business considerations (e.g., target approval rate, risk tolerance) and held fixed during training.

The effort quantification framework developed above serves a dual purpose beyond fairness-constrained training. The same optimisation machinery that computes minimal effort for regularisation also identifies the specific feature changes required for a rejected applicant to gain approval. This connection is not coincidental; effort parity ensures that the burden of improvement is equitably distributed, whilst the optimal intervention vector $\bm{\delta}^*$
or $\bmxi_{\bmX}^* $ specifies precisely what that improvement entails.

\section{Risk Controllability under Fairness Regularisation}
\label{sec:risk_bounds}

A natural concern when augmenting the training objective with a fairness penalty is whether doing so materially degrades portfolio credit risk.

\subsection{Credit Risk Measures}
\label{subsec:risk_measures}

Let $\text{LGD} \in (0,1]$ denote the loss given default and $\text{EAD}^i > 0$ the exposure at default for loan $i$. Conditional on the model parameters $\bmtheta$, the probability of
default for applicant $i$ is $\text{PD}^i(\bmtheta) = 1 - \hat{h}(\bmx^i;\bmtheta)$, and the loan-level loss is $L^i = \text{LGD} \cdot \text{EAD}^i \cdot D^i$, where
$D^i \sim \text{Bernoulli}(\text{PD}^i(\bmtheta))$.

\paragraph{Expected Loss.}
\zc{Define the complementary soft approval weight
$\bar\omega^i(\bmtheta):=1-\omega^i(\bmtheta)=
\sigma(\kappa[\hat{h}(\bmx^i;\bmtheta)-\tau])$. Let $N$ denote the number of loans. The portfolio expected loss is}
\begin{equation}\label{eq:EL}
  \EL(\bmtheta)
  = \text{LGD} \sum_{i=1}^{N}
    \bar\omega^i(\bmtheta)\cdot\text{PD}^i(\bmtheta)\cdot\text{EAD}^i.
\end{equation}

\paragraph{Unexpected Loss.}
Under the Basel ASRF single-factor model with asset-return correlation
$\rho \in (0,1)$, the regulatory capital charge for approved loan $i$ is
\begin{equation}
  \smalleq{K^i(\bmtheta)
  = \Phi\!\left[\frac{
      \Phi^{-1}(\text{PD}^i(\bmtheta)) + \sqrt{\rho}\,\Phi^{-1}(0.999)
    }{\sqrt{1-\rho}}\right] - \text{PD}^i(\bmtheta)}
\end{equation}
where $\Phi$ denotes the standard normal CDF. The portfolio unexpected
loss is

\begin{equation}\label{eq:UL}
  \UL(\bmtheta)
  = (
      \sum_{i=1}^{N} \bigl[\bar\omega^i(\bmtheta)\cdot\text{LGD}\cdot\text{EAD}^i\cdot K^i(\bmtheta)\bigr]^2
    )^{\!1/2},
\end{equation}
\zc{where $\bar\omega^i(\bmtheta)$ is the soft approval weight introduced above.}
The soft weights make the vector inside the Euclidean norm smooth; hence
$\UL$ is locally Lipschitz (and differentiable wherever that vector is
nonzero). The hard-indicator version reported in our empirical analysis is
recovered in the limit $\kappa \to \infty$.

\subsection{Bounding the Impact of Fairness Regularisation}
\label{subsec:risk_bounds}

Recall from Eq.~\eqref{eq:regularised_learning} that the fairness-regularised
objective is $\mathcal{L}_{\text{acc}}(\bmtheta) + \lambda\,\mathcal{L}_{\text{fair}}(\bmtheta)$.
We denote its minimiser by $\bmtheta_\lambda$ and the minimiser of the
fairness-unregularised predictive objective $\mathcal{L}_{\text{acc}}$ (equivalently, $\lambda=0$) by $\bmtheta_0$.
Under the convex data-fit setting specified above, $\mathcal{L}_{\text{acc}}$ is $\mu$-strongly convex.
All risk bounds flow from the following displacement lemma, which
controls how far $\bmtheta_\lambda$ can deviate from $\bmtheta_0$.

\begin{lemma}[Parameter Displacement]
\label{lem:displacement}
Let $\mathcal{L}_{\textup{acc}}$ be continuously differentiable and
$\mu$-strongly convex, and let $\mathcal{L}_{\textup{fair}}$ be locally
Lipschitz near a local minimiser $\bmtheta_\lambda$ of the regularised
objective. Then there exists
$\bmzeta_\lambda\in\partial_C\mathcal{L}_{\textup{fair}}(\bmtheta_\lambda)$
such that
\begin{equation}\label{eq:displacement}
  \nm{\bmtheta_\lambda - \bmtheta_0}
  \;\leq\;
  \frac{\lambda}{\mu}\,\nm{\bmzeta_\lambda}.
\end{equation}
\end{lemma}

The bound \eqref{eq:displacement} is informative when the Clarke
subgradients of $\mathcal{L}_{\textup{fair}}$ at $\bmtheta_\lambda$ are
finite. The following lemma supplies an explicit bound.

\begin{lemma}[Fairness Gradient Bound]
\label{lem:kfair}
Consider the logistic model and the soft fairness loss in Eq.~\eqref{eq:soft_disparity}, evaluated with the exact effort in Eq.~\eqref{eq:exact_logistic_efforts}. Let $\tilde{\bmx}^{\,i}:=(\bmx^i,1)$ and $\tilde R:=\max_i\|\tilde{\bmx}^{\,i}\|_2$. Suppose $\bmW\succ0$, $\kappa>0$, $\tau\in(0,1)$, and, writing $\bmtheta_\lambda=(\bmw_\lambda,b_\lambda)$, suppose that $c_0\leq\|\bmw_\lambda\|_2$ and $\|\bmtheta_\lambda\|_2\leq C_0$ for some $c_0,C_0>0$. For $s\in\{0,1\}$, let $n_s^-:=\#\{i:s^i=s,\ \hat h(\bmx^i;\bmtheta_\lambda)<\tau\}$ and assume $n_s^-\geq1$.

Set $M:=C_0\tilde R$, $u_{\max}:=|\tau_{\mathrm{logit}}|+M$, and $d_{\min}:=\sigma_{\min}(\bmW^{-1/2})c_0>0$. For each applicant $i$, set $\bar B^i:=\sigma_{\max}(\bmW^{-1/2})C_0\|\tilde{\bmx}^{\,i}\|_2$ and $\bar G:=\sigma_{\max}(\bmW^{-1/2})$. For each group $s\in\{0,1\}$, define
\begin{equation}
\label{eq:Phis}
\Psi_s:=\frac{2}{n_s^-}\sum_{i:\,s^i=s}
\left[
\frac{\kappa u_{\max}}{8d_{\min}}\|\tilde{\bmx}^{\,i}\|_2
+\frac{\bar B^i+u_{\max}\bar G}{d_{\min}^2}
\right].
\end{equation}
Then every $\bmzeta\in\partial_C\mathcal L_{\mathrm{fair}}(\bmtheta_\lambda)$ satisfies
\begin{equation}
\label{eq:kfair}
\|\bmzeta\|_2\leq K_{\mathrm{fair}}:=\Psi_0+\Psi_1<\infty.
\end{equation}

For causal effort, replace $\bmW^{-1/2}$ throughout by $\bmT:=\bmW^{-1/2}(\mathbb I_d-\bmA)^{-\top}$; equivalently, use $d_{\min}^c=\sigma_{\min}(\bmT)c_0$ and replace $\sigma_{\max}(\bmW^{-1/2})$ in $\bar B^i$ and $\bar G$ by $\sigma_{\max}(\bmT)$.
\end{lemma}

Lemmas~\ref{lem:displacement} and~\ref{lem:kfair} together establish that $\nm{\bmtheta_\lambda - \bmtheta_0} \leq (\lambda/\mu)\,K_{\text{fair}}$. A Lipschitz bound for a portfolio-risk functional on the segment joining $\bmtheta_0$ and $\bmtheta_\lambda$ then converts this displacement into a bound on its change.

\begin{theorem}[Risk Bounds under Fairness Regularisation]
\label{thm:risk_bounds}
Let
\[
\Theta_\lambda
:=
\left\{
\bmtheta_0+t(\bmtheta_\lambda-\bmtheta_0):t\in[0,1]
\right\}.
\]
For $R\in\{\EL,\UL\}$, suppose that $R$ is locally Lipschitz on a
neighbourhood of $\Theta_\lambda$, and define
\[
K_R
:=
\sup_{\bmtheta\in\Theta_\lambda}
\sup_{\bm\varphi\in\partial_C R(\bmtheta)}
\|\bm\varphi\|_2
<\infty.
\]
Then the following bounds hold simultaneously:
\begin{equation}\label{eq:bound_el}
  \bigl|\EL(\bmtheta_\lambda) - \EL(\bmtheta_0)\bigr|
  \;\leq\; K_{\EL}\cdot\frac{\lambda}{\mu}\cdot K_{\textup{fair}},
\end{equation}
\begin{equation}\label{eq:bound_ul}
  \bigl|\UL(\bmtheta_\lambda) - \UL(\bmtheta_0)\bigr|
  \;\leq\; K_{\UL}\cdot\frac{\lambda}{\mu}\cdot K_{\textup{fair}},
\end{equation}

\end{theorem}

Theorem~\ref{thm:risk_bounds} has a direct operational interpretation. For each risk measure $R\in\{\EL,\UL\}$, the bound combines the displacement factor $\lambda/\mu$, the fairness subgradient bound $K_{\mathrm{fair}}$ at the fitted model, and the risk-sensitivity constant $K_R$ along the segment joining $\bmtheta_0$ and $\bmtheta_\lambda$. Holding the two sensitivity constants fixed, reducing $\lambda$ or increasing $\mu$ tightens the bound. Since the constants depend on the fitted model and parameter segment, the result is best interpreted as a worst-case post-training certificate for each chosen fairness weight; it does not assert uniform linear growth or continuity of the fitted solution path in $\lambda$. The realised risk changes reported in Section~\ref{sec:experiment} are substantially smaller than these
worst-case guarantees.

\section{From Diagnosis to Remedy: Actionable Explanations}
\label{sec:discrete_recourse}

The effort-centric framework developed in Sections~\ref{sec:framework} and~\ref{sec:approximate_learning} provides aggregate diagnostics of fairness by comparing how much effort different groups must exert to obtain credit. In many decision contexts, however, applicants and lenders are also interested in \emph{individual-level} guidance, which concrete changes to an application would increase their credit assessment, subject to institutional and behavioural constraints.

In this section, we show how the same effort formulation used for fairness assessment can be repurposed to construct \emph{effort-aware pathways} for individual applicants. We pose a constrained optimisation problem over mutable features, embedding practical constraints on admissible changes, and interpret its solution as a recommended sequence of adjustments. These pathways complement our aggregate analysis by making explicit the concrete burdens that different applicants face and by highlighting how effort-regularised models alter the nature of the recommendations they receive.

\subsection{The Challenge of Credit Data}
\label{subsec:real_data_challenge}

The closed-form effort expressions developed in Section~\ref{sec:approximate_learning} concern unconstrained continuous actions, but real credit applications include continuous, discrete, and immutable features. This heterogeneity creates a practical challenge; a recommendation that an applicant ``increase the number of credit loans by 0.7'' is not meaningful.

We address this by partitioning features into three categories: \emph{continuous mutable features} that can take any value within a realistic range; \emph{ordinal mutable features} that take integer values and can change only by discrete amounts; and \emph{immutable features} that must remain fixed. This categorisation, determined by domain knowledge and regulatory constraints, ensures that generated pathways are feasible and interpretable.

\subsection{Illustrative Pathways to Improvement}
\label{subsec:computing_recommendations}

The computation requires three inputs: (i) the applicant's current feature values, (ii) a cost weight matrix encoding the relative difficulty of changing each feature, and (iii) feature-specific bounds and mutability constraints. The weight matrix may be set by the institution based on domain expertise, calibrated using empirical evidence on typical applicant trajectories, or specified by the applicant to reflect personal circumstances. This flexibility allows pathways to be tailored to individuals.

Given these inputs, we formulate the problem as finding the minimum weighted cost combination of hypothetical changes such that the modified profile would satisfy the approval threshold. The presence of discrete features transforms this into a mixed-integer programme. For linear classifiers, this formulation yields a mixed-integer quadratic programme that commercial solvers can handle efficiently. We develop a greedy algorithm that constructs pathways incrementally, prioritising changes with the highest impact per unit of difficulty. Technical details and pseudocode are provided in Appendix~\ref{sec:appendix_discrete}.

\subsection{Incorporating Causal Relationships}
\label{subsec:causal_recommendations}

The pathways discussed above treat features as independently adjustable. In practice, modifying one feature may induce systematic changes in others through underlying causal dependencies. Our causal effort explicitly accounts for these relationships by distinguishing between the interventions an individual undertakes and the resulting downstream feature changes.
In constructing pathways, we therefore specify interventions rather than final feature values. This allows the resulting pathways to reflect the natural propagation of changes implied by the structural model, yielding representations that are both more parsimonious and faithful to the underlying data-generating process.
Methodologically, we incorporate causal structure by replacing the direct gradient used in the feature-independent setting with the corresponding causal gradient when evaluating marginal benefits. This adjustment ensures that the greedy procedure respects the causality and correctly attributes the relative contribution of each intervention.

\subsection{Beyond Traditional Explanations}
\label{subsec:beyond_xai}

Post-hoc explanation methods such as feature-attribution techniques, e.g., SHAP \citep{lundberg2017unified} and LIME \citep{ribeiro2016should} are now routinely used to rationalise credit decisions. These tools identify which features contributed most to a particular prediction and can be helpful for auditing models and communicating high-level drivers of risk. However, they are fundamentally \emph{diagnostic}; they describe why an applicant was classified as high risk, but not which concrete changes would be sufficient to improve the credit assessment under realistic constraints.

For individual applicants and credit managers, this diagnostic focus is often inadequate. Delinquency history, for example, is crucial for the final decision but is inherently historical and cannot be improved through future effort. Answering such questions requires moving from explanations of the \emph{current} decision to recommendations about \emph{feasible future} profiles, taking behavioural, institutional, and causal constraints into account.

The effort-centric framework developed in Sections~\ref{sec:framework} and~\ref{sec:approximate_learning} provides precisely the modules needed for such recommendations. For each rejected applicant, we have defined an effort function that quantifies the minimum burden required to reach the approval region, either under a feature-independent or a causal notion of change. In this subsection, we repurpose this quantity to construct \emph{effort-aware pathways}; constrained sequences of adjustments to mutable features that improve their credit assessment.

Formally, we pose a constrained optimisation problem over the vector of mutable features. The objective minimises the effort required to reach the decision boundary of the learnt classifier, while the constraints encode immutability (e.g., age, past defaults) and institutional rules (e.g., upper bounds on loan amounts). The solution yields a recommended combination of changes, together with its associated effort, and these pathways complement feature-attribution explanations by translating model behaviour into concrete, feasible changes.

\section{Empirical Analysis}
\label{sec:experiment}
In this section, we demonstrate that, in the mortgage area, (i) substantial effort disparities persist even when predictive parity is satisfied, (ii) our in-processing method effectively reduces these disparities whilst preserving predictive performance, (iii) the financial cost of achieving effort fairness is modest, and (iv) predictive parity emerges as a by-product of our framework.

\subsection{Data and Experimental Setup}
\label{subsec:setup_experiment}

\paragraph{Dataset.}
We construct our dataset by linking two complementary sources: the \emph{Home Mortgage Disclosure Act} (HMDA) dataset and the \emph{Freddie Mac Single Family Loan-Level} dataset. HMDA, mandated by federal law, requires financial institutions to publicly disclose loan-level information on mortgage applications, including applicant demographics such as race, ethnicity, and gender, as well as loan characteristics such as amount, purpose, and geographic identifiers. However, HMDA lacks detailed credit risk variables and post-origination performance information. The Freddie Mac Single Family dataset, by contrast, provides rich loan-level origination characteristics, including interest rate, loan-to-value ratio, debt-to-income ratio, and credit score, together with monthly performance records that enable the construction of default outcomes. Crucially, however, Freddie Mac does not report borrower demographics. Linking these two sources,  following \citet{saadi2020role,kielty2023simplifying}, therefore allows us to combine the demographic information available in HMDA with the credit risk and performance data available from Freddie Mac, yielding a dataset suitable for studying fairness in mortgages. The matched sample comprises $16{,}250$ loan-level observations for the 2018 origination cohort. Restricting the fairness analysis to applicants with reported binary gender (female or male) yields a final analytic sample of $9{,}195$ observations.

The outcome variable is mortgage default, defined as the borrower becoming $90+$ days past due at any point during the observation window. We code $Y = 1$ for non-default, that is, loans for which no $90+$ day delinquency event is recorded throughout the entire observation period, and $Y = 0$ for observed default. Since the matched Freddie Mac performance data contain originated loans, our empirical sample does not include the subsequent repayment outcomes of applicants who were actually denied credit.
We follow standard practice in the credit scoring fairness literature, where models are trained and evaluated on the observed outcomes of originated loans \citep{kozodoi_fairness_2022}, and applicants predicted to carry a higher probability of default are treated as high-risk and denied credit; the recourse analysis should accordingly be read as measuring effort relative to this approval proxy rather than as observing real denied applicants.
We treat gender as the protected attribute $S$, with male applicants constituting the privileged group ($S = 1$) and female applicants the unprivileged group ($S = 0$). Table~\ref{tab:description_dataset} summarises the full set of features used in our analysis; we distinguish mutable features (e.g., income, debt-to-income ratio) from immutable ones (e.g., loan purpose), with immutable features assigned prohibitively large modification costs to prevent algorithmic manipulation.
Full dataset details are provided in Appendix~\ref{sec:appendix_data}.

\begin{table}[H]
\TABLE
{Feature summary for the dataset.\label{tab:description_dataset}}
{\begin{tabular}{lccc}
\hline
\up\down Symbol & Feature & Type & Mutable\\
\hline
\up $S$ & Gender & binary & $\times$\\
$X_1$ & Credit score & continuous & \checkmark\\
$X_2$ & Number of units & ordinal & $\times$\\
$X_3$ & Combined loan to value (CLTV) & continuous & \checkmark\\
$X_4$ & Debt to income ratio (DTI) & continuous & \checkmark\\
$X_5$ & Unpaid balance (UPB) & continuous & \checkmark\\
$X_6$ & Interest rate & continuous & $\times$\\
$X_7$ & Loan purpose & nominal & $\times$\\
$X_8$ & Number of borrowers & ordinal & \checkmark\\
$X_9$ & Income & continuous & \checkmark\\
\hline
\end{tabular}}
{}
\end{table}

\paragraph{Causal Structure.}
For causal effort calculations, we learn SCM over the feature space using a two-stage procedure combining constraint-based skeleton discovery (PC-stable algorithm) with score-based edge orientation (BIC-scored greedy search). Domain knowledge priors ensure plausible edge directions (e.g., income affects debt ratio). The resulting adjacency matrix $\bmA$ encodes direct causal effects, and the propagation matrix $(\mathbb{I}_d - \bmA)^{-1}$ is precomputed for use in Theorem~\ref{thm:causal_approx}. Full details of the learning procedure are provided in Appendix~\ref{sec:sm_causal}.

\paragraph{Classifier and Training.}
We employ logistic regression as the base classifier $\hat{h}(\bmx;\bmtheta)$, trained to minimise cross-entropy loss augmented with the effort fairness regulariser $\lambda\mathcal L_{\mathrm{fair}}(\bmtheta)$. Since the classifier is logistic, effort in the regulariser is evaluated with the exact logit-scale expressions in Eq.~\eqref{eq:exact_logistic_efforts}. Training uses the Adam optimiser (learning rate $10^{-2}$, batch size 128, up to 100 epochs with early stopping). The decision threshold is $\tau=0.5$, for which \zc{$\tau_{\mathrm{logit}}=0$}. We vary $\lambda$ to trace the fairness--accuracy trade-off.

\paragraph{Effort Weights.}
The diagonal weight matrix $\bmW$ encodes feature-specific modification costs. Immutable features receive weights $W_{jj} = 10^8$ to prevent change. For mutable features, we consider multiple weight configurations reflecting different assumptions about the relative difficulty of changing each feature; specifications are detailed in Appendix~\ref{sec:sm_weights}.

\paragraph{Evaluation Metrics.}
Predictive performance is measured by AUC, accuracy, and F1. We further evaluate each decision along two complementary financial dimensions, both computed on a one-year horizon over the set of approved loans: credit risk, measured by EL and UL as in Eqs.~\eqref{eq:EL} and~\eqref{eq:UL}, and profitability, measured by interest revenue,
\[
\zc{\text{Revenue}^i = \text{EAD}^i \times \text{IR}^i},
\]
\zc{where $\text{IR}^i$ is the annual interest rate.} We additionally report risk-adjusted return on capital,
\[
\text{RAROC} = \frac{\text{Revenue} - \text{EL}}{\text{UL}},
\]
which jointly summarises the risk-return trade-off across these two dimensions. All financial metrics are reported as means per approved applicant. Effort fairness is quantified via Eqs.~\eqref{eq:fi_disparity}--\eqref{eq:gc_disparity}. For actionable explanations, we report: (i) \emph{Feature Change Ratio} (FCR), the fraction of features requiring intervention; (ii) \emph{Effort}, $\|\bm{\delta}^*\|_{\bmW}$ or $\|\bmxi_{\bmX}^*\|_{\bmW}$, which is the weighted distance to the boundary (DtB); and (iii) \emph{Validity}, the share of rejected applicants whose action crosses the decision boundary, so $\mathrm{Val}=100\%$ by construction. All experiments are repeated over five random runs; we report means and standard errors.

\subsection{Descriptive Statistics and Group Differences}
\label{subsec:statistical_analysis}

Table \ref{tab:statistical_analysis} provides a diagnostic of group heterogeneity in both outcomes and the financial attributes that typically drive credit decisions. The outcome rate is high and does not differ significantly by group: 0.931 for female borrowers versus 0.932 for male borrowers. This suggests that the two groups may satisfy statistical parity in outcomes, while leaving open whether fairness holds among the rejected population or with respect to effort-based assessment.

Despite this similarity in outcomes, the covariate distributions exhibit several economically meaningful differences across groups. Male borrowers take out substantially larger loans, as reflected in a higher unpaid principal balance and a higher combined loan-to-value ratio, indicating that they finance a larger share of the property value. They also report markedly higher income and are charged slightly higher interest rates, and are somewhat more likely to purchase multi-unit properties. The debt-to-income ratio is somewhat lower for male borrowers, though the difference is only marginally significant. Loan purpose and the number of borrowers on the loan do not differ significantly between groups, and credit scores are statistically indistinguishable, suggesting that female and male borrowers are, on average, comparably creditworthy at origination. These imply that gender is associated with a distinct joint financial profile, so that even models that formally exclude sensitive attributes may still encode group membership indirectly through correlated predictors.

These distributional differences are consequential for our effort analysis, which is defined relative to a borrower's position in feature space and the feasible directions of change. Group gaps in loan size, leverage, and pricing therefore translate mechanically into group differences in distance to the decision boundary and in the minimum-cost improvement path, particularly among rejected applicants. In particular, a group carrying a higher loan-to-value ratio and a larger loan balance will typically require a larger reduction in leverage or a larger increase in income to reach approval, even when acceptance rates or error rates are equalised across groups.

The large standard deviation for income reflects the heavy-tailed nature of reported borrower income in the matched mortgage data. To limit the influence of extreme observations on model training and effort calculations, continuous features are winsorised at the 99th percentile in preprocessing, as detailed in Appendix~\ref{subsec:appendix_preprocessing}.

\begin{table}[H]
\footnotesize
\TABLE
{Summary Statistics for all, group $S=0$ and group $S=1$ borrowers.\label{tab:statistical_analysis}}
{\begin{tabular}{lllll}
\hline
 & All & $S=0$ & $S=1$ & $p$-value \\
\midrule
$Y$ & 0.931 (0.253) & 0.931 (0.254) & 0.932 (0.252) & $0.891$ \\
Credit score & 743.250 (49.665) & 742.556 (51.461) & 743.738 (48.363) & $0.266$ \\
Number of units & 1.098 (0.424) & 1.062 (0.332) & 1.124 (0.477) & $< 0.0001$ \\
CLTV & 67.513 (19.636) & 64.624 (20.816) & 69.543 (18.497) & $< 0.0001$ \\
DTI & 36.318 (9.543) & 36.569 (9.431) & 36.142 (9.618) & $< 0.05$ \\
UPB & 210454.051 (134621.820) & 192509.225 (128391.417) & 223059.619 (137449.024) & $< 0.0001$ \\
Interest rate & 4.932 (0.538) & 4.914 (0.534) & 4.944 (0.541) & $< 0.01$ \\
Loan purpose & 0.805 (0.771) & 0.823 (0.766) & 0.792 (0.774) & $0.057$ \\
Number of borrowers & 1.162 (0.388) & 1.154 (0.372) & 1.168 (0.399) & $0.107$ \\
Income & 113.615 (759.507) & 88.455 (544.868) & 131.289 (879.121) & $< 0.01$ \\
\hline
\end{tabular}}
{The table reports the p-values from two-sided t-tests to test for a difference in the mean of the focal variable between group $S=0$ and group $S=1$. Entries report mean (std).}
\end{table}

\subsection{Research Question 1: Does Masked Inequality Exist?}
\label{subsec:predictive_fairness_limitations}
We first examine whether satisfying predictive parity guarantees equitable opportunities for future approval. We investigate this by training classifiers that enforce three predictive parity criteria: Statistical Parity (SP), Equalised Odds (EO), and Positive Predictive Value Parity (PPV), alongside a baseline model that optimises only the accuracy loss. For each method, we select the configuration that offers the best trade-off and report its corresponding disparity. Table~\ref{tab:predictive_parity} confirms that predictive parity can be substantially improved without materially sacrificing predictive performance.

\begin{table}[H]
\TABLE
{Performance under predictive parity criteria.\label{tab:predictive_parity}}
{\begin{tabular}{lcccc}
\hline
Method & AUC & Accuracy & F1 & Pred.\ disparity \\
\hline
Baseline & $0.724\,(0.011)$ & $0.719\,(0.098)$ & $0.823\,(0.074)$ & --\\
SP \citep{kamishima2012fairness} & $0.724\,(0.011)$ & $0.719\,(0.098)$ & $0.823\,(0.074)$ & $1.73\times 10^{-5}\,(2.02\times 10^{-5})$ \\
EO \citep{hardt_equality_2016}& $0.722\,(0.012)$ & $0.711\,(0.010)$ & $0.817\,(0.076)$ & $0.013\,(0.005)$\\
PPV \citep{chouldechova2017fair} & $0.724\,(0.011)$ & $0.719\,(0.098)$ & $0.822\,(0.074)$ & $0.005\,(0.002)$\\
\hline
\end{tabular}}
{Entries report mean (std) across five runs.}
\end{table}

\begin{figure}[H]
\FIGURE
{\includegraphics[width=0.48\textwidth]{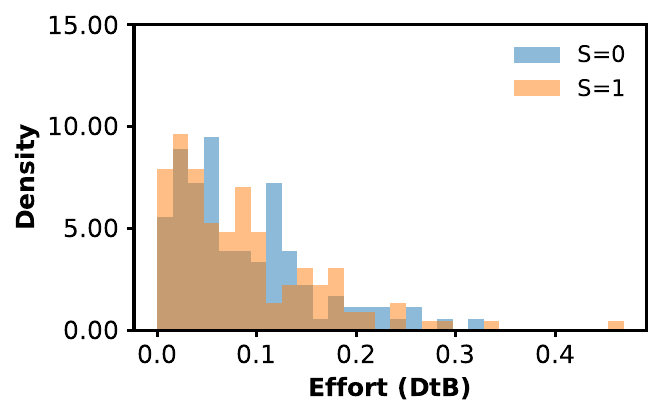}}
{Distribution of effort (distance to boundary) under the SP model.\label{fig:pp_dist_hist}}
{DtB denotes distance to boundary. The figure displays the effort distribution for rejected borrowers in each demographic group ($S=0$ and $S=1$). To quantify distributional differences beyond a single run, we report complementary distance metrics with mean (std) across five independent runs under SP: Kolmogorov--Smirnov (KS) statistic = 0.125 (0.093); Cram\'er--von Mises (CVM) statistic = 0.299 (0.291); Total Variation (TV) distance = 0.205 (0.110); Jensen--Shannon (JS) divergence = 0.044 (0.047); and Hellinger (HE) distance = 0.194 (0.123).
}
\end{figure}

\paragraph{Finding 1: Effort Disparities Persist Despite Predictive Parity.}
Despite achieving predictive parity, these criteria do not guarantee that female and male borrowers face comparable effort when attempting to overturn a loan rejection. Fig.~\ref{fig:pp_dist_hist} displays the distribution of distances to the decision boundary for rejected applicants under the model satisfying SP. Female applicants are systematically further from the boundary, indicating they must exert greater effort to obtain loans.
Due to space limitations, we present the histogram for a single run; however, we report multiple distance metrics computed across all independent runs. These aggregated metrics confirm that the observed separation is not driven by a particular experiment or by a specific choice of distance metric.

\paragraph{Answer to RQ1.}
These results provide clear evidence of masked inequality: satisfying predictive parity alone does not eliminate disparities in recourse effort. Rejected female applicants face greater effort even though the raw default rates and credit-score distributions of female and male borrowers are statistically similar (Table~\ref{tab:statistical_analysis}).
The same effort-gap pattern holds across all predictive parity criteria considered, with complete results reported in Appendix~\ref{sec:sm_pp_effort}.

\subsection{Research Question 2: Can Our Method Address Masked Inequality?}
\label{subsec:our_fairness}

We now evaluate how effectively the framework reduces these disparities and what predictive, risk, and profitability trade-offs result.

\subsubsection{Feature-independent Effort Parity}
\label{subsubsec:group_effort_parity}

\paragraph{Finding 2: Effort Parity Substantially Reduces Disparities.}
Fig.~\ref{fig:fi_group_tradeoff} displays the trade-off between predictive performance and effort disparity reduction as the regularisation weight $\lambda$ increases. Effort disparities decrease substantially, specifically by more than 50\% at higher $\lambda$ values, whilst AUC, accuracy, and F1 remain stable.

\begin{figure}[H]
\centering
\FIGURE
{
\subcaptionbox{Fairness improvement.}{
\includegraphics[height=5.0cm,keepaspectratio]{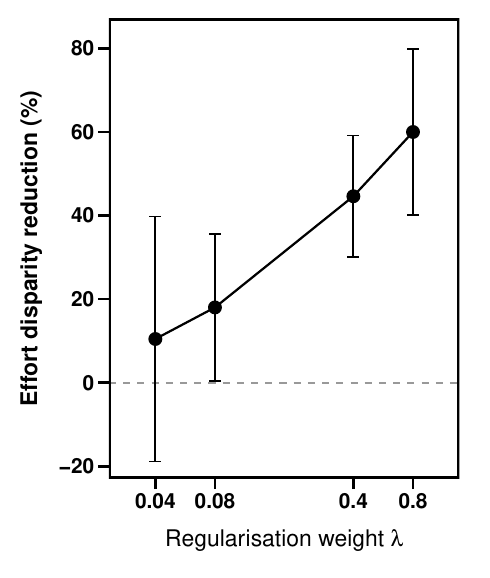}}
\subcaptionbox{Model performance stability.}{
\includegraphics[height=5.0cm,keepaspectratio]{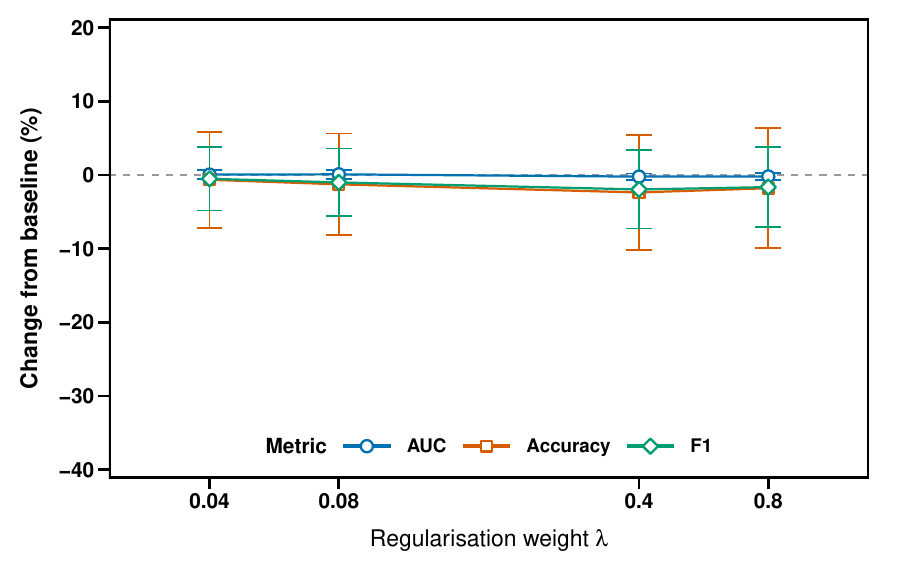}}
}
{Trade-off between predictive performance and effort disparity reduction under feature-independent effort parity.
\label{fig:fi_group_tradeoff}}
{Effort disparity reduction is measured as a percentage relative to the baseline model. Each point corresponds to a different value of $\lambda$; error bars show standard errors over five runs.}
\end{figure}

Fig.~\ref{fig:fi_group_dist_evolution} shows how the effort distributions change as $\lambda$ increases. At baseline ($\lambda = 0.00$), female applicants systematically require more effort than male applicants. The gap is visible at both ends of the distribution: among applicants requiring little effort, male borrowers account for a noticeably larger share, indicating that many of them can obtain loans with little change; meanwhile, female applicants are over-represented in the upper tail of required effort, implying substantially greater required changes. Once effort-parity regularisation is introduced, the female distribution begins to shift towards the male one, and this convergence becomes more pronounced as $\lambda$ increases towards 0.8.

\begin{figure}[H]
\FIGURE
{\includegraphics[width=0.98\textwidth]{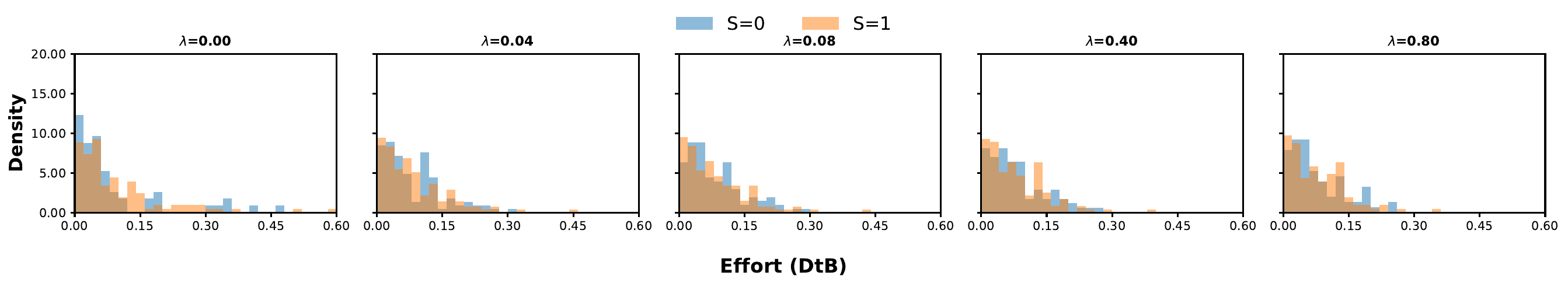}}
{Evolution of effort (distance to boundary) distributions under feature-independent effort parity.
\label{fig:fi_group_dist_evolution}}
{Each panel displays the distribution of effort for rejected borrowers in each demographic group ($S=0$ and $S=1$) at a given regularisation strength $\lambda$.}
\end{figure}

Specifically, we quantify the separation between the two distributions using diversified distance metrics in Table~\ref{tab:fi_dist_metrics_lambda}. It shows a broadly systematic reduction in cross-group discrepancy as $\lambda$ increases. The KS distance falls from 0.126 at baseline to approximately 0.075 once the smallest fairness weight is applied, and declines further to 0.066 at $\lambda = 0.8$. This pattern is broadly consistent with CVM, which drops sharply and monotonically from 0.314 at baseline to 0.076 at $\lambda = 0.8$. The histogram-based metrics provide complementary evidence. TV, JS, and HE quantify discrepancies in bin-wise probability mass under a fixed discretisation. The simultaneous decline in both distribution-based and histogram-based metrics suggests that cross-group differences narrow broadly across the effort distribution, rather than being driven solely by an extreme tail or a localised discrepancy. The consistency further indicates that this reduction is robust to the choice of distance measure.

\begin{table}[H]
\TABLE
{Distributional distance metrics between rejected female and male borrowers.\label{tab:fi_dist_metrics_lambda}}
{
\begin{tabular}{lccccc}
\hline
$\lambda$
& KS
& CVM
& TV
& JS
& HE \\
\hline
0.00 & 0.126 (0.095) & 0.314 (0.313) & 0.144 (0.086) & 0.021 (0.023) & 0.131 (0.077) \\
0.04 & 0.075 (0.019) & 0.218 (0.183) & 0.133 (0.071) & 0.015 (0.014) & 0.113 (0.057) \\
0.08 & 0.076 (0.019) & 0.174 (0.092) & 0.123 (0.055) & 0.013 (0.011) & 0.104 (0.047) \\
0.40 & 0.074 (0.029) & 0.106 (0.040) & 0.123 (0.061) & 0.015 (0.013) & 0.113 (0.053) \\
0.80 & 0.066 (0.035) & 0.076 (0.039) & 0.109 (0.047) & 0.010 (0.008) & 0.094 (0.039) \\
\hline
\end{tabular}
}
{Entries report mean (std) across five independent runs. Lower values indicate more similar distributions between groups for all metrics shown, including Kolmogorov--Smirnov (KS), Cram\'er--von Mises (CVM), Total Variation (TV), Jensen--Shannon (JS), and Hellinger (HE).}
\end{table}

\paragraph{Finding 3: Effort Parity Preserves Risk Controllability and Profitability.}

Fig.~\ref{fig:fi_risk_profit_tradeoff} reports the change in credit risk and profitability metrics relative to the unconstrained baseline as the feature-independent effort-parity penalty $\lambda$ increases. Both EL and UL remain broadly stable, with UL declining modestly as $\lambda$ increases (Panel a). This is because equalising the average distance to the decision boundary across the two groups may cause the model to tighten its approval standard for marginal applicants, so that some borrowers who would have been approved under the baseline model are now rejected. Since these excluded applicants are, on average, closer to the margin of creditworthiness, their exclusion lowers both the portfolio's expected loss and its exposure to correlated default risk (unexpected loss).
Interest revenue declines for the same reason (Panel b): tightening the approval standard shrinks the pool of approved loans.
Panel (c) reports that RAROC declines modestly as $\lambda$ increases, as the reduction in revenue is not fully offset by the accompanying reduction in expected and unexpected loss; risk-adjusted profitability is therefore somewhat lower under the fairness constraint, but it remains positive throughout the range of $\lambda$ we consider.
Feature-independent effort parity therefore reduces the lender's risk exposure, while RAROC remains positive despite its modest decline.
This demonstrates that substantial fairness gains are achievable under feature-independent effort parity without materially sacrificing financial profitability or losing control of credit risk.

\begin{figure}[H]
\centering
\FIGURE
{
\subcaptionbox{Expected and unexpected loss.}{
\includegraphics[width=0.32\textwidth,keepaspectratio]{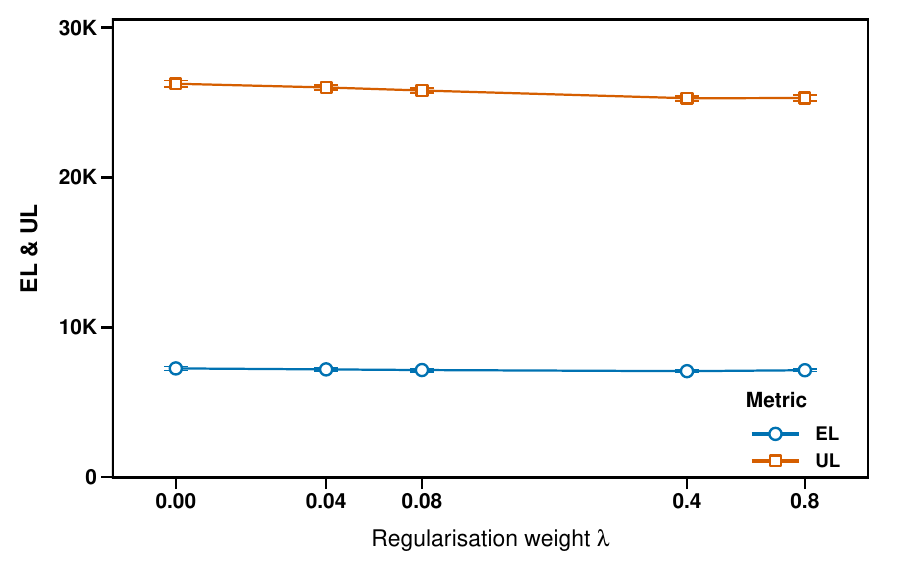}}
\subcaptionbox{Revenue.}{
\includegraphics[width=0.32\textwidth,keepaspectratio]{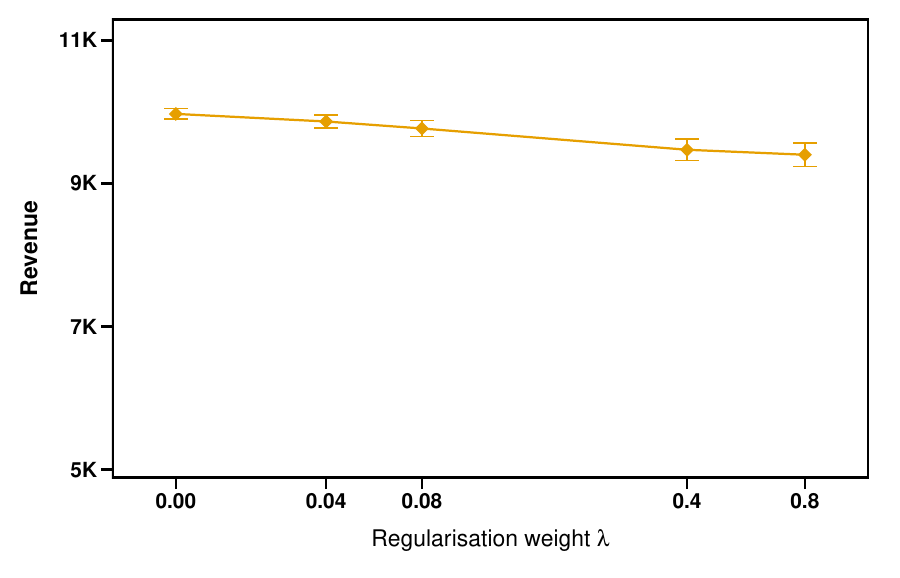}}
\subcaptionbox{Risk-adjusted return on capital.}{
\includegraphics[width=0.32\textwidth,keepaspectratio]{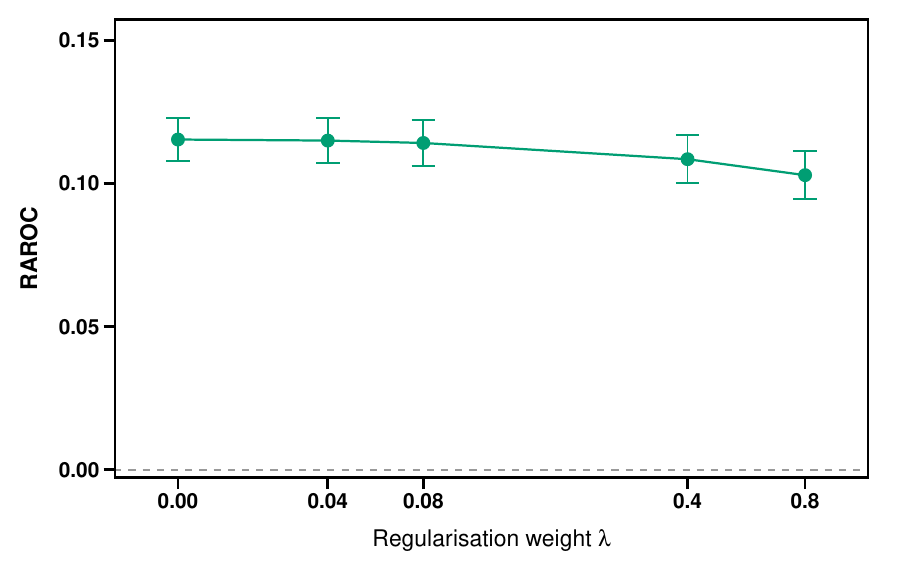}}
}
{Trade-off between credit risk and effort disparity reduction under feature-independent effort parity.
\label{fig:fi_risk_profit_tradeoff}}
{Each point corresponds to a different value of $\lambda$; error bars show standard errors over five runs.}
\end{figure}

\subsubsection{Causal Effort Parity}
\label{subsubsec:causal_group_effort_parity}

\paragraph{\zc{Finding 4: Causal Effort Parity Produces Larger Fairness Gains with Larger Trade-offs.}}

\zc{Causal effort parity responds strongly even at low regularisation.
Fig.~\ref{fig:cg_tradeoff} shows that the reduction exceeds 90\% at $\lambda=0.04$, peaks near 99\% at $\lambda=0.4$, and declines slightly at $\lambda=0.8$. These gains come with a more pronounced predictive trade-off than under feature-independent effort parity, including a lower AUC and, at larger $\lambda$, lower accuracy and F1.}

Fig.~\ref{fig:cg_group_dist_evolution} shows that, in the baseline model, female borrowers exhibit a heavier right tail than male borrowers, indicating that they require higher effort to obtain approval. As $\lambda$ increases, the two distributions move closer together, and once $\lambda \geq 0.4$, they overlap substantially across the support, including the right tail, suggesting that female and male borrowers face broadly comparable effort requirements to overturn a rejection. Table~\ref{tab:cg_dist_metrics_lambda} reports distance metrics under the causal effort parity, and it shows a similar pattern to that in Table~\ref{tab:fi_dist_metrics_lambda}.

\begin{figure}[H]
\centering
\FIGURE
{
\subcaptionbox{Fairness improvement.}{
\includegraphics[height=5.0cm,keepaspectratio]{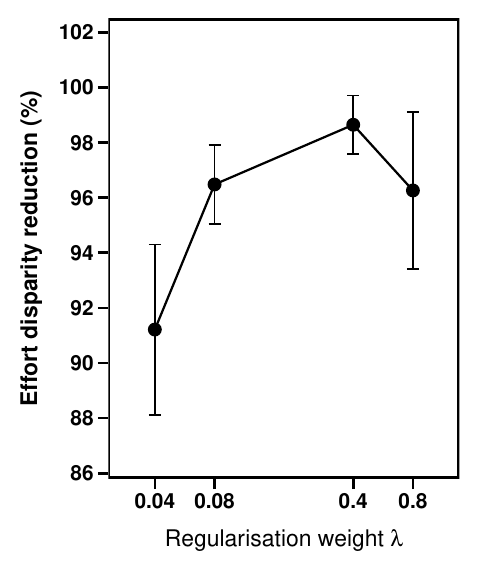}}
\subcaptionbox{Model performance stability.}{
\includegraphics[height=5.0cm,keepaspectratio]{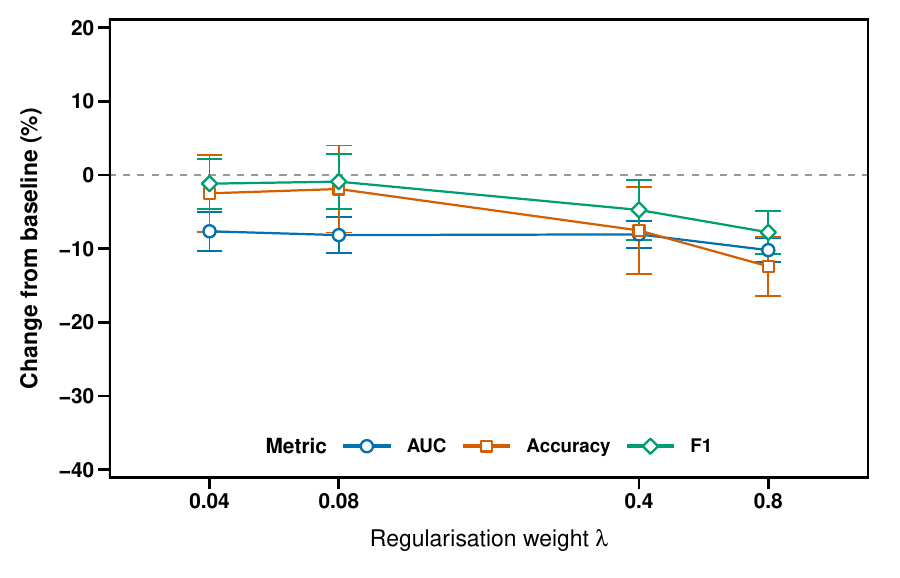}}
}
{Trade-off between predictive performance, profit, and effort disparity reduction under causal effort parity.
\label{fig:cg_tradeoff}}
{Effort disparity reduction is measured as a percentage relative to the baseline model. Each point corresponds to a different value of $\lambda$; error bars show standard errors over five runs.}
\end{figure}

\begin{figure}[H]
\FIGURE
{\includegraphics[width=0.98\textwidth]{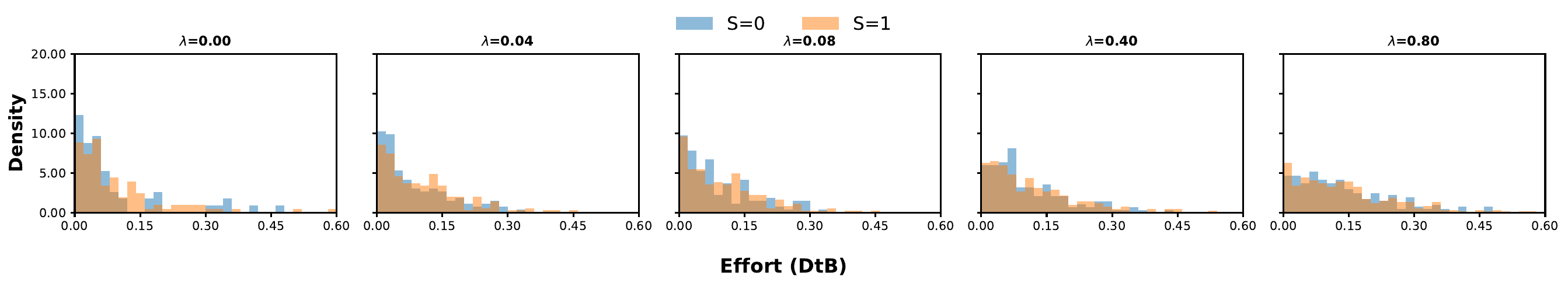}}
{Evolution of effort (distance to boundary) distributions under causal effort parity.
\label{fig:cg_group_dist_evolution}}
{Each panel displays the distribution of effort for rejected borrowers in each demographic group ($S=0$ and $S=1$) at a given regularisation strength $\lambda$.}
\end{figure}

\begin{table}[H]
\TABLE
{Distributional distance metrics between rejected female and male borrowers.\label{tab:cg_dist_metrics_lambda}}
{
\begin{tabular}{lccccc}
\hline
$\lambda$
& KS
& CVM
& TV
& JS
& HE \\
\hline
0.00 & 0.126 (0.095) & 0.314 (0.313) & 0.144 (0.086) & 0.021 (0.023) & 0.131 (0.077) \\
0.04 & 0.108 (0.058) & 0.375 (0.361) & 0.111 (0.051) & 0.011 (0.010) & 0.098 (0.046) \\
0.08 & 0.116 (0.083) & 0.471 (0.629) & 0.122 (0.071) & 0.013 (0.014) & 0.103 (0.056) \\
0.40 & 0.092 (0.043) & 0.285 (0.271) & 0.122 (0.051) & 0.013 (0.010) & 0.107 (0.042) \\
0.80 & 0.063 (0.022) & 0.172 (0.125) & 0.091 (0.037) & 0.007 (0.006) & 0.078 (0.031) \\
\hline
\end{tabular}
}
{Entries report mean (std) across five independent runs. Lower values indicate more similar distributions between groups for all metrics shown, including Kolmogorov--Smirnov (KS), Cram\'er--von Mises (CVM), Total Variation (TV), Jensen--Shannon (JS), and Hellinger (HE).}
\end{table}

We provide the detailed results in Appendix~\ref{sec:sm_effort_framework}. Our framework can also apply to alternative classifiers, and the results remain consistent with our findings. Moreover, it can be deployed at the individual level to identify applicants who face disproportionately high effort, enabling a more comprehensive diagnosis of effort disparities. We report these robustness tests in Appendix~\ref{sec:sm_svm}-\ref{sec:sm_individual_effort}.

Fig.~\ref{fig:cg_risk_profit_tradeoff} shows a risk and profitability pattern under causal effort parity that differs from the feature-independent case.
Both EL and UL increase relative to baseline.
Empirically, the gap relative to baseline is largely established at the smallest regularisation weight we consider ($\lambda = 0.04$) and does not widen further as $\lambda$ increases to 0.8.
This observed saturation is consistent with, but is not implied by, the finite-change guarantees.
Interest revenue remains close to baseline, but the higher losses reduce RAROC; it nevertheless remains positive throughout the considered range of $\lambda$.

\begin{figure}[H]
\centering
\FIGURE
{
\subcaptionbox{Expected and unexpected loss.}{
\includegraphics[width=0.32\textwidth,keepaspectratio]{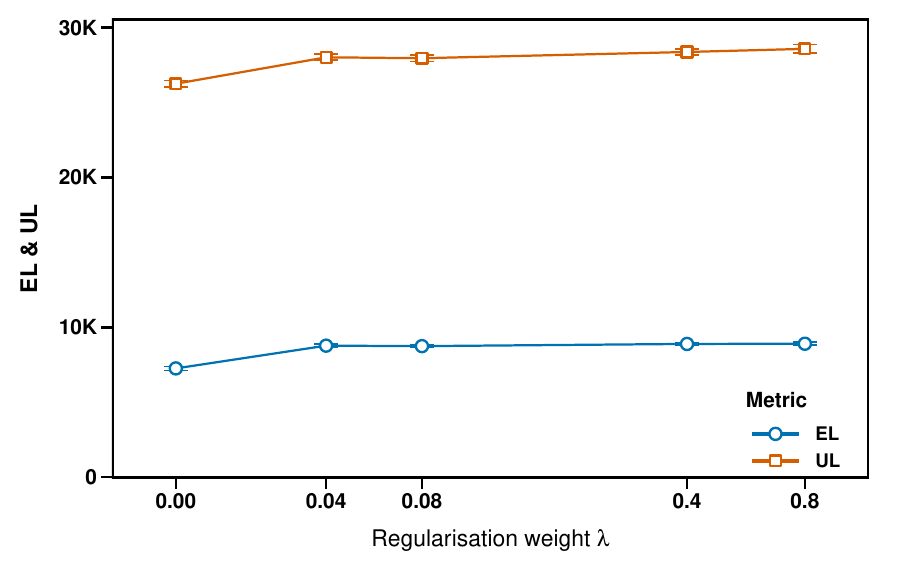}}
\subcaptionbox{Revenue.}{
\includegraphics[width=0.32\textwidth,keepaspectratio]{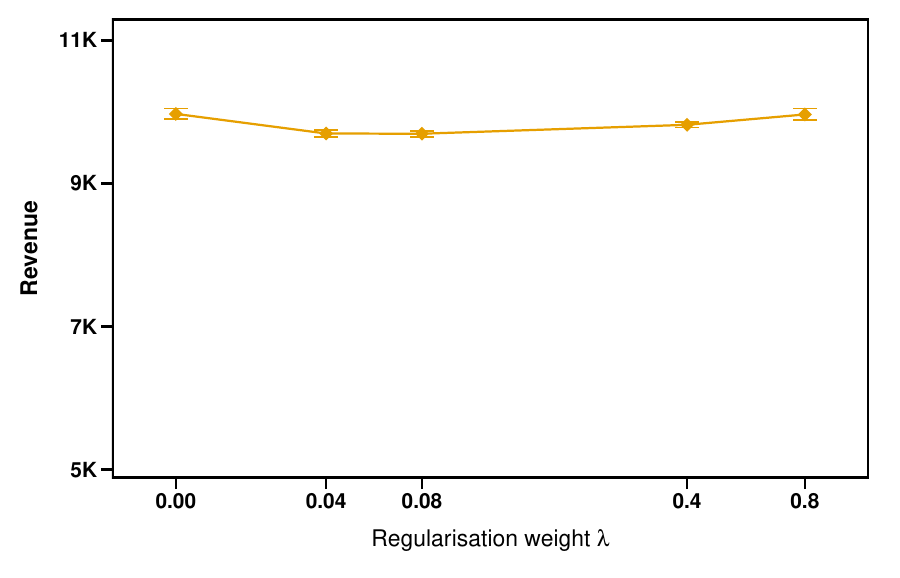}}
\subcaptionbox{Risk-adjusted return on capital.}{
\includegraphics[width=0.32\textwidth,keepaspectratio]{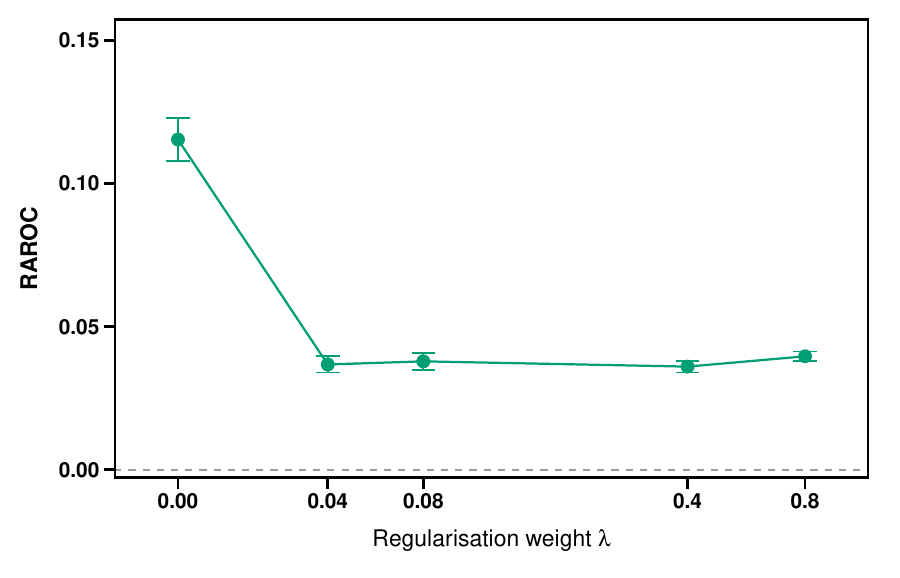}}
}
{Trade-off between credit risk, profitability, and effort disparity reduction under causal effort parity.
\label{fig:cg_risk_profit_tradeoff}}
{Each point corresponds to a different value of $\lambda$; error bars show standard errors over five runs.}
\end{figure}

\subsubsection{Effects on Predictive Parity Metrics}
\label{subsubsec:effort_pp_effects}

\paragraph{Finding 5: Effort Fairness Improves Predictive Parity as a By-Product.}
An interesting by-product of effort-centric fairness is its effect on standard predictive parity. Although our framework optimises solely for effort disparity, we observe that several predictive parity criteria improve as a consequence. Fig.~\ref{fig:pp_under_effort} illustrates this relationship under feature-independent effort parity; the causal results appear in Appendix~\ref{subsec:sm_pp_effects}. As $\lambda$ increases and effort disparities decline, EO disparity decreases. The absolute change in SP is modest, as the baseline model already starts from a relatively low SP disparity. PPV remains nearly flat, consistent with its dependence on both error rates and group-specific outcome prevalence.

\begin{figure}[H]
\FIGURE
{\includegraphics[width=0.4\textwidth]{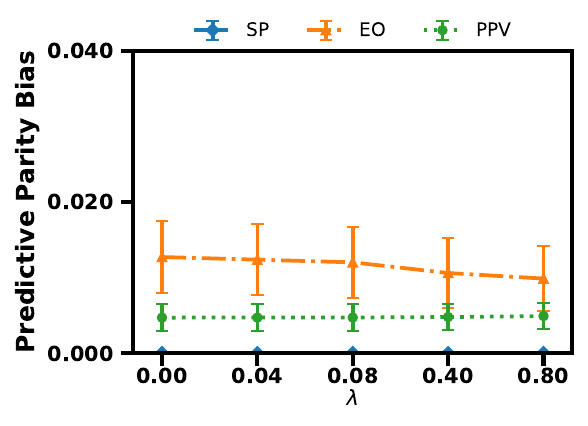}}
{Evolution of predictive parity metrics under feature-independent effort parity.
\label{fig:pp_under_effort}}
{Optimising for effort fairness often reduces predictive disparities as a by-product.}
\end{figure}

\paragraph{Answer to RQ2.}
Our effort-based regularisation method effectively addresses masked inequality. The framework substantially reduces effort disparities, achieving over 50\% reductions, with modest impact on predictive performance, profitability and risk control. The observed improvements in predictive parity under effort-based training suggest potential synergies between these fairness objectives, indicating that targeting effort-based disparities can simultaneously enhance conventional outcome-based fairness without explicitly constraining it.

\subsection{Illustrative Case Study: From ``Why Rejected'' to ``How to Improve''}
\label{subsec:case_study}
As part of our framework, each rejected applicant receives targeted feedback that enumerates actionable improvement plans for credit assessment. For every rejected individual, we compute feasible pathways under a classifier trained with causal effort parity by solving the optimisation formulation introduced in Section~\ref{subsec:beyond_xai}. We illustrate the procedure using a rejected applicant drawn from the test set; Table~\ref{tab:case_actions} reports five recommended improvement plans. The related feature change cost weights are provided in Appendix~\ref{sec:sm_weights}.

The individual is a rejected mortgage applicant with a credit score of 662, a single-unit property, an original combined loan-to-value ratio of $90\%$, an original debt-to-income ratio of $40\%$, an original unpaid balance of \$205K, an interest rate of $4.125\%$, and one borrower on the loan. The combination of a high loan-to-value ratio, an elevated debt burden, and a comparatively low credit score results in a risk score below the approval threshold. Several attributes are not mutable in the short run, including the number of units, interest rate, and loan purpose, which are fixed at origination. The recommended plans therefore combine an improved credit score with a lower debt-to-income ratio, a reduced unpaid balance through a smaller requested loan amount, higher income, and, in every plan, an additional co-borrower.

This case study complements the aggregate results in Section~\ref{subsec:our_fairness}, demonstrating how these reductions materialise at the level of specific applicants and concrete actions. In particular, it highlights that enforcing effort parity does not merely adjust summary statistics but can qualitatively alter the nature of the recommendations communicated to different groups of borrowers.
\begin{table}[H]
\TABLE
{Case study: rejected borrower and improvement plans.\label{tab:case_actions}}
{
\begin{tabular}{llccccc}
\hline
 & & Plan.\ 1 & Plan.\ 2 & Plan.\ 3 & Plan.\ 4 & Plan.\ 5 \\
\hline
Required changes & $\Delta$ Credit score $X_1$          & $+13.12$ & $+14.44$ & $+11.75$ & $+15.01$ & $+6.42$ \\
                  & $\Delta$ Number of units $X_2$       & 0       & 0       & 0       & 0       & 0       \\
                  & $\Delta$ CLTV (pp) $X_3$             & 0       & 0       & 0       & 0       & 0       \\
                  & $\Delta$ DTI (pp) $X_4$              & $-0.80$  & $-0.44$  & $-1.42$  & $-0.23$  & $-3.11$  \\
                  & $\Delta$ UPB (\$K) $X_5$             & $-1.74$ & $-1.92$ & $-1.56$ & $-2.00$ & $-0.85$ \\
                  & $\Delta$ Interest rate $X_6$         & 0       & 0       & 0       & 0       & 0       \\
                  & $\Delta$ Loan purpose $X_7$          & 0       & 0       & 0       & 0       & 0       \\
                  & $\Delta$ Number of borrowers $X_8$   & $+1$    & $+1$    & $+1$    & $+1$    & $+1$    \\
                  & $\Delta$ Income (\$K) $X_9$          & $+25.36$ & $+6.98$ & $+45.41$ & $+14.50$ & $+99.32$ \\
\hline
\multicolumn{2}{l}{FCR}            & 5/9    & 5/9    & 5/9    & 5/9    & 5/9    \\
\multicolumn{2}{l}{Effort (DtB)}   & 0.1586 & 0.1587 & 0.1588 & 0.1588 & 0.1602 \\
\hline
\end{tabular}
}
{The research object is a rejected mortgage applicant with a credit score of 662, one housing unit, CLTV 90\%, DTI 40\%, UPB \$205K, interest rate 4.125\%, one borrower, and income \$92K. Entries report the direct interventions required of the applicant; passively induced downstream changes are not shown. In particular, CLTV is not intervened upon directly (hence $\Delta$CLTV $=0$) but adjusts passively with UPB.}
\end{table}

\section{Discussion}
\label{sec:discussions}

The effort-centric perspective developed in this paper yields several implications for the design, evaluation, and governance of AI-based credit scoring.
The results support adding effort disparity to model-validation dashboards alongside predictive-parity measures.
For regulators, the measures expose masked unequal routes to future approval that predictive parity can miss. Applicant-level pathways translate that diagnosis into clearer customer communication and may inform targeted product design.

Implementation nevertheless requires both economic and legal judgement. Effort gaps should be assessed alongside risk, and $\lambda$ should be treated as a policy choice rather than a mandate for full equalisation. The theoretical risk analysis complements this governance process with a post-training certificate based on the fitted parameter segment. Institutions can compare this certificate with realised changes in expected and unexpected loss and with their risk tolerance.

In addition, the in-processing method uses protected-group membership during training even though the deployed score does not take it as an input. Since such use is operationally sensitive in regulated credit markets, institutions can first employ effort disparity for auditing and model validation, and use fairness-aware training only with appropriate legal authority and governance controls \citep{kozodoi_fairness_2022}.

We emphasise that the illustrative pathways generated by our framework are not guarantees of future approval; lending decisions depend on circumstances at the time of application, and applicant situations, model parameters, and institutional policies may all evolve. Rather, these pathways serve as transparency tools, helping applicants understand the factors most relevant to their assessment. Although we argue that effort-fair models produce more equitable pathways, we have not empirically validated whether applicants perceive these explanations as fairer. User studies or experiments would be needed. Lenders should treat pathways as informative illustrations rather than contractual commitments.

\section{Conclusions}
\label{sec:conclusions}

This paper develops an effort-centric complement to predictive fairness in credit scoring. We formalise feature-independent and causal effort parity, derive tractable training objectives and risk-change bounds, and use the same framework to generate applicant-level pathways to approval.

In the mortgage application, effort disparities persist despite predictive parity.
Feature-independent regularisation reduces the gap by about 60\% at the largest tested weight with modest predictive changes; causal regularisation yields reductions above 90\% at positive tested weights but entails larger predictive and risk-return trade-offs.
Conventional fairness metrics sometimes improve as a by-product.
Financial effects are explicitly accounted for: the theoretical analysis provides a finite post-training bound, while the empirical analysis reports realised trade-offs.
Under feature-independent effort parity, expected and unexpected losses remain broadly stable, while revenue and RAROC decline modestly.
Under causal effort parity, expected and unexpected losses increase at the smallest fairness weight but do not widen materially as the penalty increases; revenue remains close to baseline, and RAROC remains positive throughout.

The conclusions are limited by the estimated linear SCM and the static study design. Longitudinal borrower data are needed to evaluate implemented recommendations, repeated applications, feedback, and model drift. Future work should address these issues with dynamic models and longitudinal validation. Moreover, the 2018 cohort coincides with COVID‑era forbearance and reporting disruptions, so future work needs a more robust default label.

\clearpage
\appendix

\section{Differentiable Fairness Loss Formulations}
\label{sec:appendix_loss}

Using $e_{\mathrm{GC}}^i(\bmtheta)=r_{\mathrm{log}}^i$ for the empirical logistic model, and $e_{\mathrm{GC}}^i(\bmtheta)=\tilde r(\bmv^i;\bmtheta)$ for a general differentiable classifier, the causal soft disparity is
\begin{equation}
\label{eq:approx_group_causal_loss_appendix}
\tilde{\Delta}^{\text{GC}}(\bmtheta)
= \left| \frac{\sum_{i: s^i=0} \omega^i e_{\mathrm{GC}}^i(\bmtheta)}{\sum_{i: s^i=0} \omega^i} - \frac{\sum_{i: s^i=1} \omega^i e_{\mathrm{GC}}^i(\bmtheta)}{\sum_{i: s^i=1} \omega^i} \right|.
\end{equation}
which is locally Lipschitz and differentiable almost everywhere wherever the applicable denominator norm is nonzero. The soft rejection weights ensure continuous transitions as applicants cross the boundary.

\section{Discrete Feature Handling and Pathway Generation}
\label{sec:appendix_discrete}

\subsection{Problem Formulation}

Let us partition the feature space into two disjoint subsets: $\mathcal{I}_{\text{cont}}$ (indices of continuous features) and $\mathcal{I}_{\text{ord}}$ (indices of ordinal discrete features). For each discrete feature $j \in \mathcal{I}_{\text{ord}}$, we introduce a binary decision variable $a_j \in \{0, 1\}$ selecting between floor rounding ($a_j = 0$) and ceiling rounding ($a_j = 1$). The full feasible action vector $\bm{\delta}_{\text{feas}} \in \mathds{R}^d$ is defined component-wise as:
\begin{equation}
\label{eq:discrete_action_signed_appendix}
\delta_{\text{feas},j} =
\begin{cases}
\zc{p_j}, & j \in \mathcal{I}_{\text{cont}}, \\
\lfloor \delta^*_j \rfloor + a_j\left(\lceil \delta^*_j \rceil - \lfloor \delta^*_j \rfloor\right), & j \in \mathcal{I}_{\text{ord}}, \, g_j \geq 0, \\
\lceil \delta^*_j \rceil + a_j\left(\lfloor \delta^*_j \rfloor - \lceil \delta^*_j \rceil\right), & j \in \mathcal{I}_{\text{ord}}, \, g_j < 0,
\end{cases}
\end{equation}
\zc{where $\mathbf p = (p_j \in \mathds{R} : j \in \mathcal{I}_{\text{cont}})$ are continuous decision variables, $\bma = (a_j \in \{0, 1\} : j \in \mathcal{I}_{\text{ord}})$ are binary decision variables, and $g_j = [\nabla_{\bmx} \hat{h}(\bmx^i; \bmtheta)]_j$ is the $j$-th element of the gradient vector.}

The mixed-integer optimisation problem seeks:
\begin{equation}
\label{eq:mix_optimisation_appendix}
\zc{\min_{\bma, \mathbf p} \quad \|\bm{\delta}_{\text{feas}}\|_{\bmW} \quad \text{s.t. } \quad \nabla_{\bmx} \hat{h}(\bmx^i; \bmtheta)^\top (\bm{\delta}_{\text{feas}} - \bm{\delta}^*) \geq 0.}
\end{equation}

\subsection{Greedy Algorithm}

We propose a greedy algorithm that ranks features by their utility-per-effort ratio and explores adjustments to find feasible actions.

\begin{algorithm}[H]
\small
\caption{\small{Greedy Algorithm for Feasible Actions}}
\label{alg:greedy_mixed_appendix}
\begin{algorithmic}[1]
\Require Individual features $\bmx^i$, continuous optimal action $\bm{\delta}^*$, classifier $\hat{h}(\cdot; \bmtheta)$, gradient $\mathbf{g} = \nabla_{\bmx} \hat{h}(\bmx^i; \bmtheta)$, weight matrix $\bmW$, threshold $\tau$, safety margin $\varepsilon > 0$
\Ensure Set of feasible actions $\mathcal{A}$

\State \textbf{Initialisation:}
\State \zc{Set $a_j = 0$ for all $j \in \mathcal{I}_{\text{ord}}$ and $p_j = \delta^*_j$ for all $j \in \mathcal{I}_{\text{cont}}$}
\State \zc{Compute utility-per-effort: $\bm{\upsilon} = |\bmW^{-1} \mathbf{g}|$ with components $\upsilon_j$}
\State \zc{Sort all features in descending order of $\upsilon_j$: obtain ordering $(j_1, j_2, \ldots, j_{d})$}
\State Initialise candidate set: $\mathcal{A} = \emptyset$

\State \textbf{Greedy Feature Exploration:}
\For{$l = 1$ to $d$}
    \If{$j_l \in \mathcal{I}_{\text{cont}}$} \Comment{Continuous feature}
        \State \zc{Optimise $p_{j_l}$ subject to: $\mathbf{g}^\top (\bm{\delta}_{\text{feas}} - \bm{\delta}^*) \geq \varepsilon$}
        \If{feasible}
            \State \zc{Compute cost: $\mathsf c_l = \|\bm{\delta}_{\text{feas}}\|_{\bmW}$}
            \State \zc{Add to candidate set: $\mathcal{A} \leftarrow \mathcal{A} \cup \{(\bm{\delta}_{\text{feas}}, \mathsf c_l)\}$}
        \EndIf
        \State \zc{Reset: $p_{j_l} = \delta^*_{j_l}$}
    \ElsIf{$j_l \in \mathcal{I}_{\text{ord}}$} \Comment{Discrete feature}
        \State Set $a_{j_l} = 1$
        \State Construct $\bm{\delta}_{\text{feas}}$ via Eq.~\eqref{eq:discrete_action_signed_appendix}
        \If{$\mathbf{g}^\top (\bm{\delta}_{\text{feas}} - \bm{\delta}^*) \geq \varepsilon$}
            \State \zc{Compute cost: $\mathsf c_l = \|\bm{\delta}_{\text{feas}}\|_{\bmW}$}
            \State \zc{Add to candidate set: $\mathcal{A} \leftarrow \mathcal{A} \cup \{(\bm{\delta}_{\text{feas}}, \mathsf c_l)\}$}
        \EndIf
        \State Reset: $a_{j_l} = 0$
    \EndIf
\EndFor

\State \textbf{Validation:}
\State \zc{Filter candidate set: $\mathcal{A} \leftarrow \{(\bm{\delta}, \mathsf c) \in \mathcal{A} : \hat{h}(\bmx^i + \bm{\delta}; \bmtheta) > \tau \}$}

\State \textbf{Output:}
\State \zc{\Return candidate set $\mathcal{A}$ sorted by cost $\mathsf c$ in ascending order}
\end{algorithmic}
\end{algorithm}

The algorithm explores each feature independently by adjusting it from the baseline configuration whilst keeping all other features at their initial values. For continuous features, we optimise the specific component $p_{j_l}$ to satisfy the constraint. For discrete features, we test flipping $a_{j_l}$ to 1 and check if the constraint holds. All feasible configurations are collected in $\mathcal{A}$, allowing selection of the minimal-effort solution.

The validation step filters the candidate set to retain only action plans that achieve the required classifier output $\hat{h}(\bmx^i + \bm{\delta}_{\text{feas}}; \bmtheta) \geq \tau + \varepsilon$, where $\varepsilon$ provides a safety margin to account for approximation errors in the linearisation.

\subsection{Extension to Causal Effort}

For causal effort calculations, the algorithm applies directly with the gradient $\nabla_{\bmx} \hat{h}(\bmx^i; \bmtheta)$ replaced by the effective causal gradient $(\mathbb{I}_d - \bmA)^{-\top} \nabla_{\bmx} \hat{h}(\bmx^i; \bmtheta)$. All solution algorithms remain applicable with this modified gradient. The causal framework may yield different discrete actions compared to the feature-independent case, as the structural dependencies encoded in $\bmA$ alter the effective benefit of each feature change.

\section{Structural Causal Model}
\label{sec:appendix_scm}

This section collects the formal definitions of the structural causal model (SCM) and the associated intervention and counterfactual operators used in the main text. In the application, the endogenous variables specialise to $(\bmX, S)$, but we keep the presentation general here.

\subsection{General Structural Causal Model}
\label{subsec:appendix_scm}

Let $\bmV = (V_1, \ldots, V_{d+1})$ denote a vector of \emph{endogenous} variables (determined inside the modelled system) and let $\bmU = (U_1, \ldots, U_{d+1})$ denote a vector of \emph{exogenous} variables (unobserved noise terms external to the system). An SCM is defined as the pair
\begin{equation*}
\label{eq:scm_definition}
\mathcal{M} = (\bmB, \prob_{\bmU}),
\end{equation*}
where $\bmB = \{ V_i = f_i(\bmPA_i, U_i) \}_{i=1}^{d+1}$ is a collection of \emph{structural equations}. Here $\bmPA_i \subseteq \bmV \setminus \{V_i\}$ denotes the set of direct causal parents of $V_i$, and $f_i$ is a measurable function encoding how $V_i$ is generated from its parents and the exogenous disturbance $U_i$. The distribution $\prob_{\bmU}$ determines the joint law of the exogenous variables.

We assume \emph{causal sufficiency}, i.e.
\begin{equation*}
\prob_{\bmU} = \prod_{i=1}^{d+1} \prob_{U_i},
\end{equation*}
so that there are no unmodelled common causes of the endogenous variables. In the credit-scoring application, we instantiate $\bmV = (\bmX, S)$, where $\bmX$ collects the non-protected features and $S$ is the protected attribute.

The structural equations induce a directed graph $\mathcal{G}$ with vertex set $\bmV$ and an edge $V_j \to V_i$ whenever $V_j \in \bmPA_i$. We assume that $\mathcal{G}$ is a directed acyclic graph (DAG). Under acyclicity, the system of structural equations admits a unique solution for $\bmV$ given $\bmU$, and the induced observational distribution $\prob_{\bmV}$ factorises according to the causal Markov property \citep{pearl2009causality}:
\begin{equation*}
\prob_{\bmV}(\bmv)
= \prod_{i=1}^{d+1} \prob_{V_i \mid \bmPA_i}(v_i \mid \bmpa_i),
\end{equation*}
where each conditional distribution $\prob_{V_i \mid \bmPA_i}$ is implicitly defined by $f_i$ and $\prob_{U_i}$.

\subsection{Linear SCM and Propagation Matrix}
\label{subsec:appendix_scm_matrix}

For analytical tractability in the main text, we focus on a \emph{linear} SCM, while the general framework extends to non-linear specifications \citep{pearl2009causality}. A linear SCM over $\bmV \in \mathbb{R}^{d+1}$ is defined by
\begin{equation}
\label{eq:linear_scm_supp}
\bmV = \bmF \bmV + \bmU,
\end{equation}
where $\bmF \in \mathbb{R}^{(d+1)\times(d+1)}$ is the \emph{structural coefficient matrix}. The entry $F_{ij}$ quantifies the direct linear effect of $V_j$ on $V_i$, and by construction $F_{ij} = 0$ whenever $V_j \notin \bmPA_i$.

Under an appropriate ordering of the variables, acyclicity of the associated graph is equivalent to nilpotency of $\bmF$ and guarantees that $(\mathbb{I}_{d+1} - \bmF)$ is invertible. Equation~\eqref{eq:linear_scm_supp} then implies
\begin{equation*}
\bmV = (\mathbb{I}_{d+1} - \bmF)^{-1} \bmU.
\end{equation*}
We refer to $(\mathbb{I}_{d+1} - \bmF)^{-1}$
as the \emph{propagation matrix}. Its $(i,j)$-th entry captures the total (direct and indirect) causal effect on $V_i$ of a unit perturbation in $U_j$, aggregated along all directed paths in the graph.
In the main text we use the relevant block of $(\mathbb I_{d+1}-\bmF)^{-1}$ to propagate additive structural shifts applied to the equations for $\bmX$.

\subsection{Interventions}
\label{subsec:appendix_interventions}

\zc{An intervention replaces the structural equations of a subset of variables by externally imposed assignments. Let $\mathcal{I} \subseteq \{1, \ldots, d+1\}$ denote a set of indices and write $\bmV_{\mathcal{I}}$ for the associated subvector of endogenous variables. For a vector of imposed values $\bm{\chi}_{\mathcal{I}}$, the intervention $do(\bmV_{\mathcal{I}} = \bm{\chi}_{\mathcal{I}})$ is represented by the manipulated SCM}
\begin{equation*}
\zc{\mathcal{M}^{do(\bm{\chi}_{\mathcal{I}})} = \bigl( \bmB^{do(\bm{\chi}_{\mathcal{I}})}, \prob_{\bmU} \bigr)},
\end{equation*}
\zc{where, for each $i \in \mathcal{I}$, the structural equation $V_i = f_i(\bmPA_i, U_i)$ is replaced by $V_i = \chi_i$, whilst the equations for $i \notin \mathcal{I}$ remain unchanged. At the graphical level, this corresponds to deleting all incoming edges into nodes in $\mathcal{I}$.}

The interventional distribution
\begin{equation*}
\zc{\prob_{\bmV \mid do(\bm{\chi}_{\mathcal{I}})} := \prob_{\mathcal{M}^{do(\bm{\chi}_{\mathcal{I}})}}}
\end{equation*}
satisfies the \emph{truncated factorisation} \citep{pearl2009causality}:
\begin{equation}
\label{eq:truncated_factorisation_supp}
\zc{\prob_{\bmV_{-\mathcal{I}} \mid do(\bm{\chi}_{\mathcal{I}})}(\bmv_{-\mathcal{I}})}
= \prod_{j \notin \mathcal{I}} \prob_{V_j \mid \bmPA_j}(v_j \mid \bmpa_j),
\end{equation}
\zc{where $\bmV_{\mathcal{I}}$ is fixed at $\bm{\chi}_{\mathcal{I}}$ and $\bmV_{-\mathcal{I}}$ denotes the remaining variables. This is a \emph{hard} value-setting intervention. Its incoming edges are removed, so simultaneous hard interventions do not in general propagate through the original matrix $(\mathbb I-\bmF)^{-1}$.}

The causal-effort model in the main text instead uses a soft \emph{additive shift intervention}. For a sparse vector $\bmxi$ supported on actionable coordinates, it changes the linear SCM to $\bmV^{\bmxi}=\bmF\bmV^{\bmxi}+\bmU+\bmxi$ while preserving all incoming edges. This distinction is essential: the original propagation matrix applies to these additive shifts, whereas a hard intervention requires a propagation operator constructed after cutting the intervened rows.

\subsection{Counterfactuals Twin Constructions}
\label{subsec:appendix_counterfactuals}

Counterfactuals describe hypothetical values of $\bmV$ under interventions, conditional on a factual realisation. Let $\bmv^{\texto} = (\bmx^{\texto}, s^{\texto})$ denote the observed profile of an individual. Pearl’s abduction–action–prediction procedure constructs counterfactuals in three steps:
\begin{enumerate}
    \item \textbf{Abduction:} infer a (possibly set-valued) realisation $\bmu$ of the exogenous variables that is compatible with the factual observation $\bmV = \bmv^{\texto}$ under $\mathcal{M}$;
    \item \textbf{Action:} form the manipulated model $\mathcal{M}^{do(\bm{\chi}_{\mathcal{I}})}$ corresponding to the desired intervention $do(\bmV_{\mathcal{I}} = \bm{\chi}_{\mathcal{I}})$;
    \item \textbf{Prediction:} solve the structural equations of $\mathcal{M}^{do(\bm{\chi}_{\mathcal{I}})}$ using $\bmu$ to obtain the counterfactual vector $\bmV_{\bm{\chi}_{\mathcal{I}}}(\bmu)$.
\end{enumerate}
We write this compactly as
\begin{equation*}
\bmV_{\bm{\chi}_{\mathcal{I}}}(\bmu)
\equiv \bmV \mid do(\bmV_{\mathcal{I}} = \bm{\chi}_{\mathcal{I}}),\; \bmV = \bmv^{\texto}.
\end{equation*}

For the additive action used by causal effort, let $\mathbb J_{\mathcal I}\bmxi_{\mathcal I}$ insert the direct shifts into the full vector. After abduction, the action step adds this vector to the selected structural equations without deleting their incoming edges. We denote the resulting action-based counterfactual profile by
\begin{equation*}
\bmv_{\bmxi}^{\texto}
\equiv \bmV_{\operatorname{shift}(\mathbb J_{\mathcal I}\bmxi_{\mathcal I})}(\bmu).
\end{equation*}
This profile captures the direct shifts and all downstream changes. Hard $do$-counterfactuals remain available for value-setting questions, such as the $S$-twins considered later, but they are not the intervention semantics used in the causal-effort closed form.

\clearpage

\section{Proofs of Main Results}
\label{sec:appendix_proofs}

\subsection{Total Feature Change Under a Linear Additive Shift}
\label{subsec:sm_total_fi_scm}

Before proving the main results, we establish how an additive action propagates through the causal structure.

Following the discussion on linear SCM, we partition $\bmV$ into non-protected features $\bmX$ and the protected attribute $S$. The structural equations for $\bmX$ can be written as
\begin{equation}
\label{eq:linear_SCM_partition_appendix}
    \zc{\bmX = \bmF_{\bmX\bmX} \bmX + \bmF_{\bmX S} S + \bmU_{\bmX}},
\end{equation}
where $\bmF_{\bmX\bmX}$ encodes the dependencies among the features in $\bmX$, $\bmF_{\bmX S}$ captures the direct influence of the protected attribute $S$ on $\bmX$, and \zc{$\bmU_{\bmX}$} denotes the corresponding exogenous noise terms.

For causal effort, additive shifts target the non-protected features $\bmX$ while $S$ is held fixed. Hard interventions on $S$ are considered separately only when constructing $S$-twins. With $S$ fixed, Eq.~\eqref{eq:linear_SCM_partition_appendix} can be rewritten as
\begin{equation}
\label{eq:causal_equilibrium_appendix}
    \bmX
    = \big(\mathbb{I}_d - \bmF_{\bmX\bmX}\big)^{-1}
      \big( \bmF_{\bmX S} S + \bmU_{\bmX} \big)
    = \big(\mathbb{I}_d - \bmA\big)^{-1} \mathbf{B},
\end{equation}
where we define $\bmA := \bmF_{\bmX\bmX} \in \mathbb{R}^{d \times d}$ as the structural matrix and $\mathbf{B} := \bmF_{\bmX S} S + \bmU_{\bmX}$ for shorthand.

This representation allows us to derive the total feature change produced by an additive structural shift.

\begin{lemma}[Total Feature Change Under a Linear Additive Shift]
\label{lem:causal_change_appendix}
Let the factual linear SCM be $\bmX=\bmA\bmX+\mathbf B$. For a set $\bmI\subseteq\{1,\ldots,d\}$, let $\bmxi_{\bmI}\in\mathbb R^{|\bmI|}$ be additive shifts to the selected structural equations, and let $\mathbb J_{\bmI}$ be the corresponding selection matrix. The shifted system is $\bmX^{\bmxi}=\bmA\bmX^{\bmxi}+\mathbf B+\mathbb J_{\bmI}\bmxi_{\bmI}$. Its total change relative to the factual equilibrium is
\begin{equation}
\label{eq:total_feature_to_linear_causal_intervention_appendix}
\bm\delta
=\bmX^{\bmxi}-\bmX
=(\mathbb I_d-\bmA)^{-1}\mathbb J_{\bmI}\bmxi_{\bmI}.
\end{equation}
Equivalently, if $\bmxi_{\bmX}:=\mathbb J_{\bmI}\bmxi_{\bmI}\in\mathbb R^d$ is the sparse full action vector, then $\bm\delta=(\mathbb I_d-\bmA)^{-1}\bmxi_{\bmX}$.
\end{lemma}

\begin{proof}
Subtracting the factual equation $\bmX=\bmA\bmX+\mathbf B$ from the shifted equation gives $\bm\delta=\bmA\bm\delta+\mathbb J_{\bmI}\bmxi_{\bmI}$. Since acyclicity makes $\mathbb I_d-\bmA$ invertible, multiplying by its inverse proves Eq.~\eqref{eq:total_feature_to_linear_causal_intervention_appendix}. No equality between the propagation matrices of an original and an edge-cut graph is invoked.
\end{proof}

\subsection{Proof of Lemma~\ref{lem:boundary_crossing} (Optimality of Boundary Crossing)}

\begin{proof}
\zc{Let $\bm\delta^*$ be an optimal feasible action for $\bmx^{\texto}$. If $\hat h(\bmx^{\texto}+\bm\delta^*;\bmtheta)>\tau$, continuity of $t\mapsto\hat h(\bmx^{\texto}+t\bm\delta^*;\bmtheta)$, together with $\hat h(\bmx^{\texto};\bmtheta)<\tau$, gives some $t\in(0,1)$ for which $\hat h(\bmx^{\texto}+t\bm\delta^*;\bmtheta)=\tau$. Assumption~\ref{ass:boundary_concentration} makes $t\bm\delta^*$ feasible. Positive homogeneity of the weighted norm gives $\|t\bm\delta^*\|_{\bmW}=t\|\bm\delta^*\|_{\bmW}<\|\bm\delta^*\|_{\bmW}$, contradicting optimality. Hence every optimum reaches the boundary exactly.}
\end{proof}

\subsection{\texorpdfstring{\zc{Proof of Theorem~\ref{thm:fi_approx} (Feature-Independent Local Effort)}}{Proof of Theorem 1 (Feature-Independent Local Effort)}}

\begin{proof}
\zc{Write $q^i=\tau-\hat h(\bmx^i;\bmtheta)>0$ and $\bmg^i=\nabla_{\bmx}\hat h(\bmx^i;\bmtheta)$. For any action satisfying the linearised approval constraint, weighted Cauchy--Schwarz gives}
\begin{equation*}
\zc{q^i\le {\bmg^i}^{\top}\bm\delta
\le \|\bmW^{-1/2}\bmg^i\|_2\|\bm\delta\|_{\bmW}.}
\end{equation*}
\zc{Equality is attained uniquely by Eq.~\eqref{eq:fi_direction}, which proves Eq.~\eqref{eq:fi_effort_approx}. Assumption~\ref{ass:boundary_concentration} states when this unconstrained local solution is admissible for approximating the original constrained problem; otherwise, the omitted constraints must remain in the linearised optimisation.}
\end{proof}

\subsection{Proof of Theorem~\ref{thm:causal_approx} (Causal Local Effort)}

\begin{proof}
\zc{Lemma~\ref{lem:causal_change_appendix} gives $\bm\delta=\mathbf P\bmxi_{\bmX}$. Hence the linearised approval constraint is ${\bmg^i}^\top\mathbf P\bmxi_{\bmX}\ge q^i$. Applying the preceding weighted Cauchy--Schwarz argument to $\mathbf P^\top\bmg^i$ proves Eqs.~\eqref{eq:causal_effort_approx}--\eqref{eq:causal_direction}. The conclusion concerns additive structural shifts; it does not apply to simultaneous hard interventions that cut incoming edges.}
\end{proof}

\subsection{Derivation of the Exact Logistic Expressions}

\zc{For logistic regression, strict monotonicity of $\sigma$ makes $\hat h(\bmx^i+\bm\delta;\bmtheta)\ge\tau$ equivalent to $\bmw^\top\bm\delta\ge\tau_{\mathrm{logit}}-z^i$. The feature-independent expression in Eq.~\eqref{eq:exact_logistic_efforts} and its direction therefore follow from the same weighted Cauchy--Schwarz argument applied to $\bmw$. Under an additive causal shift, the constraint becomes $\bmw^\top\mathbf P\bmxi_{\bmX}\ge\tau_{\mathrm{logit}}-z^i$, so applying it to $\mathbf P^\top\bmw$ gives the causal expression and direction. These derivations are exact whenever the displayed unconstrained actions satisfy the feasible set.}

\subsection{Proof of Corollary~\ref{cor:causal_amplification} (Causal Amplification)}

\begin{proof}
\zc{Dividing the right-hand side of Eq.~\eqref{eq:causal_effort_approx} by that of Eq.~\eqref{eq:fi_effort_approx} gives}
\begin{equation}
\frac{\tilde{r}(\bmv^i;\bmtheta)}{\tilde{c}(\bmx^i;\bmtheta)}
=
\frac{\|\bmW^{-1/2}\bmg^i\|_2}
{\|\bmW^{-1/2}\mathbf{P}^{\top}\bmg^i\|_2}
=\zc{\frac{1}{\gamma^i}}.
\end{equation}
\zc{Since $\bmW\succ0$, $\mathbf{P}$ is invertible, and $\bmg^i\neq\mathbf{0}$, both norms are positive. The three directional cases therefore follow directly from the value of $\gamma^i$. Moreover,}
\begin{align}
&\|\bmW^{-1/2}\mathbf{P}^{\top}\bmg^i\|_2^2
-\|\bmW^{-1/2}\bmg^i\|_2^2 \nonumber\\
&\qquad=
{\bmg^i}^{\top}
\left(\mathbf{P}\bmW^{-1}\mathbf{P}^{\top}-\bmW^{-1}\right)
\bmg^i,
\end{align}
\zc{which establishes the equivalent condition in the corollary. For the exact logistic expressions, $c_{\mathrm{log}}^i$ and $r_{\mathrm{log}}^i$ have the same numerator $[\tau_{\mathrm{logit}}-z^i]_+$, while the ratio of their denominators is $\gamma(\bmw)$. Hence $r_{\mathrm{log}}^i=c_{\mathrm{log}}^i/\gamma(\bmw)$ for every applicant. Taking group means and then their absolute difference establishes Eq.~\eqref{eq:logistic_effort_scaling}.}
\end{proof}

\subsection{Proof of Lemma~\ref{lem:displacement} (Parameter Displacement)}

\begin{proof}
Since $\mathcal{L}_{\mathrm{fair}}$ is locally Lipschitz, Clarke's
necessary condition at the local minimiser $\bmtheta_\lambda$ gives a
subgradient
\zc{$\bmzeta_\lambda\in\partial_C\mathcal{L}_{\mathrm{fair}}(\bmtheta_\lambda)$}
such that
\begin{equation}
\label{eq:kkt_lambda}
\nabla\mathcal{L}_{\mathrm{acc}}(\bmtheta_\lambda)
+\zc{\lambda\bmzeta_\lambda}
=\mathbf 0.
\end{equation}
\zc{At the fairness-unregularised minimiser,}
\begin{equation}
\label{eq:kkt_0}
\nabla\mathcal{L}_{\mathrm{acc}}(\bmtheta_0)=\mathbf 0.
\end{equation}
Subtracting Eq.~\eqref{eq:kkt_0} from Eq.~\eqref{eq:kkt_lambda} and taking
the inner product with $\bmtheta_\lambda-\bmtheta_0$ yields
\begin{align*}
&\left\langle
\nabla\mathcal{L}_{\mathrm{acc}}(\bmtheta_\lambda)
-\nabla\mathcal{L}_{\mathrm{acc}}(\bmtheta_0),
\bmtheta_\lambda-\bmtheta_0
\right\rangle \\
&\qquad =
-\lambda\left\langle
\zc{\bmzeta_\lambda},
\bmtheta_\lambda-\bmtheta_0
\right\rangle.
\end{align*}
Strong convexity bounds the left-hand side below by
$\mu\|\bmtheta_\lambda-\bmtheta_0\|_2^2$, while Cauchy--Schwarz bounds the
absolute value of the right-hand side above by
\zc{$\lambda\|\bmzeta_\lambda\|_2\|\bmtheta_\lambda-\bmtheta_0\|_2$.} Cancelling
the displacement norm when it is nonzero proves
\begin{equation}
\label{eq:displacement_supp}
\|\bmtheta_\lambda-\bmtheta_0\|_2
\leq
\zc{\frac{\lambda}{\mu}\|\bmzeta_\lambda\|_2}.
\end{equation}
The result is immediate when $\bmtheta_\lambda=\bmtheta_0$.
\end{proof}

\subsection{Proof of Lemma~\ref{lem:kfair} (Fairness Gradient Bound)}

\begin{proof}
\zc{Write $\tilde{\bmx}^{\,i}=(\bmx^i,1)$, $z^i=\bmtheta^\top\tilde{\bmx}^{\,i}$, $u^i=[\tau_{\mathrm{logit}}-z^i]_+$, and $d_{\bmW}(\bmw)=\|\bmW^{-1/2}\bmw\|_2$. The exact feature-independent logistic effort is $e^i=u^i/d_{\bmW}(\bmw)$. We establish the bounds at $\bmtheta_\lambda$.}

\smallskip
\noindent \textit{Step 1: effort is bounded.}
\zc{The parameter bound implies $|z^i|\le C_0\tilde R=M$, hence $0\le u^i\le u_{\max}$.} Moreover,
\begin{equation*}
d_{\min}=\sigma_{\min}(\bmW^{-1/2})c_0
\le \zc{d_{\bmW}(\bmw_\lambda)}
\le\sigma_{\max}(\bmW^{-1/2})C_0.
\end{equation*}
\zc{Consequently, $0\le e^i(\bmtheta_\lambda)\le u_{\max}/d_{\min}$.}

\smallskip
\noindent \textit{Step 2: individual-effort subgradients are bounded.}
\zc{Every $\bm{\nu}^i\in\partial_Cu^i(\bmtheta_\lambda)$ satisfies $\|\bm{\nu}^i\|_2\le\|\tilde{\bmx}^{\,i}\|_2$, while $d_{\bmW}(\bmw_\lambda)\le\sigma_{\max}(\bmW^{-1/2})C_0$ and $\|\nabla_{\bmw}d_{\bmW}(\bmw_\lambda)\|_2\le\bar G$. The Clarke quotient rule therefore gives, for every $\bmzeta^i\in\partial_Ce^i(\bmtheta_\lambda)$,}
\begin{equation*}
\zc{\|\bmzeta^i\|_2}
\le\frac{\bar B^i+u_{\max}\bar G}{d_{\min}^2}.
\end{equation*}

\smallskip
\noindent \textit{Step 3: group-effort subgradients are bounded.}
\zc{Let $E_s=\sum_{i:s^i=s}\omega^ie^i$, $Z_s=\sum_{i:s^i=s}\omega^i$, and $\tilde F_s=E_s/Z_s$. The soft rejection weight satisfies $\|\nabla_{\bmtheta}\omega^i\|_2\le(\kappa/16)\|\tilde{\bmx}^{\,i}\|_2$. Every rejected applicant has $\omega^i>1/2$, so $Z_s>n_s^-/2$, and $0\le\tilde F_s\le u_{\max}/d_{\min}$. Applying the Clarke product, sum, and quotient rules with the bounds above yields, for every $\bmzeta_s\in\partial_C\tilde F_s(\bmtheta_\lambda)$,}
\begin{equation*}
\zc{\|\bmzeta_s\|_2}
\le\frac{2}{n_s^-}\sum_{i:s^i=s}
\left[
\frac{\kappa u_{\max}}{8d_{\min}}\|\tilde{\bmx}^{\,i}\|_2
+\frac{\bar B^i+u_{\max}\bar G}{d_{\min}^2}
\right]
=\Psi_s.
\end{equation*}

\smallskip
\noindent \textit{Step 4: the fairness-loss subgradient is bounded.}
\zc{The Clarke chain and sum rules for $\mathcal L_{\mathrm{fair}}=|\tilde F_0-\tilde F_1|$ give $\|\bmzeta\|_2\le\Psi_0+\Psi_1=K_{\mathrm{fair}}$ for every $\bmzeta\in\partial_C\mathcal L_{\mathrm{fair}}(\bmtheta_\lambda)$. For exact causal logistic effort, replace $\bmW^{-1/2}$ by $\bmT=\bmW^{-1/2}(\mathbb I_d-\bmA)^{-\top}$ throughout. Invertibility of $\bmT$ follows from $\bmW\succ0$ and acyclicity, and the same argument gives the causal constants stated in Lemma~\ref{lem:kfair}.}
\end{proof}

\subsection{Proof of Theorem~\ref{thm:risk_bounds} (Risk Bounds under Fairness Regularisation)}

\begin{proof}
For a locally Lipschitz function $R$, Lebourg's mean value theorem and the
definition of $K_R$ give
\begin{equation}
\label{eq:risk-lipschitz-supp}
|R(\bmtheta_\lambda)-R(\bmtheta_0)|
\leq
K_R\|\bmtheta_\lambda-\bmtheta_0\|_2.
\end{equation}
The parameter segment is compact. For the logistic model, every predicted
probability is strictly between zero and one along this segment. Thus the
summands in $\EL$ are smooth there. The components inside the Euclidean norm
defining $\UL$ are also smooth there, so $\UL$ is locally Lipschitz even at
a point where that vector might vanish. Consequently, the constants in the
theorem are finite under its stated neighbourhood assumption.

Lemma~\ref{lem:displacement} supplies a Clarke subgradient
\zc{$\bmzeta_\lambda\in\partial_C\mathcal L_{\mathrm{fair}}(\bmtheta_\lambda)$}
such that
\begin{equation*}
\|\bmtheta_\lambda-\bmtheta_0\|_2
\leq
\zc{\frac{\lambda}{\mu}\|\bmzeta_\lambda\|_2}.
\end{equation*}
Lemma~\ref{lem:kfair} gives
\zc{$\|\bmzeta_\lambda\|_2\leq K_{\mathrm{fair}}$.} Substituting these two bounds
into Eq.~\eqref{eq:risk-lipschitz-supp}, first with $R=\EL$ and then with
$R=\UL$, proves Eqs.~\eqref{eq:bound_el} and~\eqref{eq:bound_ul}.
\end{proof}

\section{Dataset Detail}
\label{sec:appendix_data}

\subsection{Data Source and Description}
\label{subsec:sm_data_description}

We construct our dataset by linking two complementary data sources: the
\emph{Home Mortgage Disclosure Act} (HMDA) dataset and the \emph{Freddie
Mac Single Family Loan-Level} dataset.

\paragraph{HMDA Dataset.}
HMDA, mandated under the Home Mortgage Disclosure Act of 1975 and administered by the Consumer Financial Protection Bureau, requires financial institutions to publicly disclose loan-level information on
mortgage applications. It is the only publicly available dataset containing loan-level applicant demographics, including gender, race, and ethnicity, for both approved and denied mortgage applications in the United States. HMDA also records loan characteristics such as loan amount, loan purpose,
occupancy type, and geographic identifiers at the census tract level. However, HMDA lacks detailed credit risk variables and contains no post-origination
performance information, making it insufficient on its own for default prediction modelling.

\paragraph{Freddie Mac Single Family Loan-Level Dataset.}
The Freddie Mac Single Family Loan-Level dataset is a publicly available dataset covering fixed-rate mortgages purchased by Freddie Mac. It provides
rich origination-level characteristics, including interest rate, loan-to-value ratio, debt-to-income ratio, credit score, loan purpose, and occupancy type, together with monthly performance records that enable the construction of default outcomes at the loan level. Crucially, however, Freddie Mac does not report borrower demographics such as gender or race, precluding fairness analysis without an external demographic source.

\paragraph{Dataset Linkage.}
Linking HMDA and Freddie Mac allows us to combine applicant demographic information with credit risk characteristics and post-origination performance
data, yielding a dataset suitable for studying fairness in mortgage default prediction. We construct the matched dataset following \citet{saadi2020role,kielty2023simplifying}, who establish a deterministic record linkage procedure for merging HMDA and Freddie Mac data in the absence of a common
identifier. The matching is performed on five loan-level variables that appear in both datasets: the three-digit ZIP code prefix (\texttt{zip3}), occupancy type, loan purpose, unpaid balance, and interest rate. The geographic identifier is coarsened to the three-digit ZIP code prefix to align with the level of geographic granularity available in HMDA,
following \citet{saadi2020role}. To mitigate erroneous matches, we restrict the sample to 30-year fixed-rate mortgages and retain only \emph{uniquely matched} records, that is, observations for which the combination of matching keys appears exactly once in each dataset independently, ensuring a one-to-one correspondence between HMDA and Freddie Mac records. The final matched sample comprises $16{,}250$ loan-level observations for the 2018 origination cohort. Restricting the fairness analysis to applicants with reported binary gender information yields a final analytic sample of $9{,}195$ observations, consisting of female and male borrowers.

\paragraph{Label Definition.}
The outcome variable is mortgage default, defined as the borrower becoming $90{+}$ days past due at any point during the observation window. For our effort-centric fairness analysis, we use the complement of this label:
\begin{equation}
Y = \begin{cases}
1 & \text{if no default (approval)}, \\
0 & \text{if default (rejection)}.
\end{cases}
\end{equation}
This monotone relabelling preserves the learning problem and aligns with standard underwriting practice based on risk, where approval is granted when the predicted delinquency risk is sufficiently low.

\paragraph{Protected Attribute.}
We define gender as the protected attribute $S$, sourced from HMDA applicant records:
\begin{equation}
S = \begin{cases}
1 & \text{if male (privileged group)}, \\
0 & \text{if female (unprivileged group)}.
\end{cases}
\end{equation}
Gender disparities in mortgage lending have received sustained regulatory and academic attention. Under the Equal Credit Opportunity Act and the Fair Housing Act, gender is an explicitly protected characteristic in credit decisions. Empirical evidence documents that borrowers of different genders face different treatment in mortgage markets. The disparities motivate examining whether credit scoring models impose disproportionate effort burdens on some applicants seeking mortgage approval.

\subsection{Feature Descriptions}
\label{subsec:sm_feature}

Table~\ref{tab:sm_feature_details} provides detailed descriptions of all features used in our analysis, including their economic interpretation and source dataset.

\begin{table}[H]
\renewcommand{\arraystretch}{1.0}
\footnotesize
\centering
\caption{Detailed feature descriptions for the matched HMDA--Freddie Mac
dataset.}
\label{tab:sm_feature_details}
\begin{tabular}{>{\raggedright\arraybackslash}p{1.2cm}
>{\raggedright\arraybackslash}p{3.6cm}
>{\raggedright\arraybackslash}p{5.0cm}
>{\raggedright\arraybackslash}p{2.0cm}
>{\raggedright\arraybackslash}p{1.5cm}}
\toprule
\textbf{Symbol} & \textbf{Feature Name} & \textbf{Description} &
\textbf{Source} & \textbf{Range}\\
\midrule
$S$ & Gender & Binary indicator: 1 if male, 0 if female & HMDA & $\{0,1\}$\\
$X_1$ & Credit Score & Borrower's credit score at origination & Freddie Mac
& $[300, 850]$\\
$X_2$ & Number of Units & Number of units in the mortgaged property &
Freddie Mac & $\mathbb{N}$\\
$X_3$ & CLTV & Combined loan-to-value ratio at origination
& Freddie Mac & $[0, \infty)$\\
$X_4$ & DTI & Debt-to-income ratio at origination &
Freddie Mac & $[0, \infty)$\\
$X_5$ & Unpaid Balance & Unpaid principal balance of the
loan at origination (USD) & Freddie Mac & $[0, \infty)$\\
$X_6$ & Interest Rate & Interest rate on the mortgage loan &
Freddie Mac & $[0, \infty)$\\
$X_7$ & Loan Purpose & Purpose of the mortgage loan (purchase, no-cash-out refinance, or cash-out refinance, refinance without specified reason) & Freddie Mac & nominal\\
$X_8$ & Number of Borrowers & Number of borrowers obligated on the mortgage note & Freddie Mac & $\mathbb{N}$\\
$X_9$ & Income & Annual gross income of the applicant (USD) & HMDA &
$[0, \infty)$\\
\bottomrule
\end{tabular}
\end{table}

\subsection{Mutable vs Immutable Features}
\label{subsec:sm_mutable}

We distinguish features based on whether they can plausibly be changed through applicant action:

\paragraph{Immutable Features.}
The following features are treated as immutable in our framework:
\begin{itemize}
    \item \textbf{Protected attribute} ($S$): Gender is a legally protected characteristic and cannot be the basis for required change.
    \item \textbf{Number of units} ($X_2$): The number of units in the mortgaged property is a fixed physical characteristic of the collateral
    and cannot be changed by the applicant.
    \item \textbf{Interest rate} ($X_6$): The contract interest rate is determined by the lender at origination and is not subject to applicant
    modification.
    \item \textbf{Loan purpose} ($X_7$): The purpose of the loan reflects the applicant's transaction intent and is fixed at the time of application.
\end{itemize}
In our cost function, immutable features receive weights $W_{jj} = 10^8$, effectively preventing any modification.
\paragraph{Mutable Features.}
The following features are treated as mutable:
\begin{itemize}
    \item \textbf{Credit score} ($X_1$): Credit score can be improved through concrete, verifiable actions (e.g., timely repayment of existing obligations), and is therefore treated as mutable.
    \item \textbf{CLTV} ($X_3$): The combined loan-to-value ratio can be reduced by increasing the down payment or paying down existing
    liens. Constrained to $[0, \infty)$.
    \item \textbf{DTI} ($X_4$): The debt-to-income ratio can be improved by reducing outstanding debt obligations or increasing income.
    Constrained to $[0, \infty)$.
    \item \textbf{Unpaid balance} ($X_5$): The requested loan amount can be reduced by increasing the down payment or applying for a smaller loan. Constrained to $[0, \infty)$.
    \item \textbf{Number of borrowers} ($X_8$): An applicant may add a
    co-borrower to strengthen the application. Integer-valued and constrained
    to $\mathbb{N}$.
    \item \textbf{Income} ($X_9$): Annual gross income can be increased
    through career advancement, additional employment, or other income
    sources. Constrained to $[0, \infty)$.
\end{itemize}

\subsection{Preprocessing}
\label{subsec:appendix_preprocessing}
\label{subsec:sm_proprocessing}

\paragraph{Outlier Treatment.}
Extreme values in continuous features (beyond the 99th percentile) are winsorised to reduce the influence of outliers on model training and effort calculations.

\paragraph{Normalisation.}
For effort calculations involving the weighted norm $\|\cdot\|_{\bmW}$, continuous features are standardised to have zero mean and unit variance (computed on the training set). This ensures that effort is measured on a comparable scale across features. The normalisation parameters are stored and applied consistently to the test set.

\paragraph{Group Composition.}
The sample comprises:
\begin{itemize}
    \item $3{,}794$ observations with $S = 0$ (female applicants), of which $262$ defaults within next two years ($Y = 0$) and $3{,}532$ have no defaults ($Y = 1$).
    \item $5{,}401$ observations with $S = 1$ (male applicants), of which $369$ defaults within next two years ($Y = 0$) and $5{,}032$ have no defaults ($Y = 1$).
\end{itemize}
The baseline rejection rate is thus $6.91\%$ for female applicants and $6.83\%$ for male applicants.

\paragraph{Training and Test Split.}
The data are split into training (80\%, $n = 7{,}356$) and testing (20\%, $n = 1{,}839$) sets using stratified sampling to preserve the joint distribution of labels $Y$ and protected attributes $S$ across splits.

\clearpage
\section{Causal Structure Learning}
\label{sec:sm_causal}

This section provides full details on the causal discovery procedure used to estimate the SCM.

\subsection{Learning Procedure}
\label{subsec:sm_learning_procedure}

We learn the causal graph $\mathcal{G}$ over the feature space $\bmV = (\bmX, S)$ using a two-stage hybrid approach that combines constraint-based and score-based methods.

\paragraph{Stage 1: Skeleton Discovery.}
We first learn the undirected skeleton of the causal graph using the PC-stable algorithm \citep{colombo2014order, kalisch2007estimating}. This constraint-based method identifies conditional independencies through partial correlation tests. We set the significance level to $\alpha = 0.01$ to balance sensitivity and specificity. The PC-stable variant ensures order-independent results, providing robustness against the arbitrary ordering of variables.

\paragraph{Stage 2: Edge Orientation.}
Given the learned skeleton, we orient edges using a forward and backward greedy search scored by the Bayesian Information Criterion (BIC) \citep{schwarz1978estimating, chickering2002optimal}. Local models are fitted using generalised linear models (GLM) \citep{nelder1972generalized} appropriate to each variable's type (Gaussian for continuous, Poisson for count variables).

\paragraph{Domain Knowledge Priors.}
To enhance interpretability and ensure economic plausibility, we enforce hard priors based on domain knowledge:
\begin{itemize}
    \item Unaware unfairness, Immutability constraints: The protected attribute $S$ (gender) may causally affect all other features \citep{chiappa_path-specific_2019, dwork_fairness_2012} but cannot be caused by them.
    \item Definitional relations: Income ($X_9$) affect DTI ($X_4$) and UPB ($X_5$) affects CLTV ($X_3$).
\end{itemize}

\paragraph{Bootstrap Aggregation.}
To assess stability, we generate $B = 100$ bootstrap replicates and retain edges appearing in at least 60\% of replicates. This aggregation procedure reduces sensitivity to sampling variability.

\subsection{SCM Parameter Estimation}
\label{subsec:sm_scm_parameter}

Given the learned graph $\mathcal{G}$, we estimate the linear SCM $\mathcal{M} = (\bmB, \prob_{\bmU})$ by fitting structural equations for each endogenous variable. For each feature $X_j$ with parents $\bmpa_j$ in the graph, we estimate:
\begin{equation}
X_j = \sum_{k \in \bmpa_j} A_{jk} X_k + U_j,
\end{equation}
using Bayesian Ridge Regression to provide regularisation and uncertainty quantification. The residuals are assumed to follow independent Gaussian distributions $U_j \sim \mathcal{N}(0, \sigma_j^2)$, with variances estimated from the data.
The resulting adjacency matrix $\bmA \in \mathbb{R}^{d \times d}$ will be used in its derived propagation matrix $(\mathbb{I}_d - \bmA)^{-1}$.

\clearpage

\section{Effort Disparities Under Predictive Parity: Full Results}
\label{sec:sm_pp_effort}

This section provides comprehensive statistical evidence that effort disparities persist across all predictive parity criteria examined in the main text.

\subsection{Predictive Parity Definitions}
\label{subsec:sm_pp_definitions}

For completeness, we provide formal definitions of the predictive parity criteria used as benchmarks in our empirical analysis:

\begin{itemize}
    \item \textbf{Statistical Parity} (SP) \citep{kamishima2012fairness} requires equal approval rates: $\prob(\hat{Y} = 1 \mid S = 0) = \prob(\hat{Y} = 1 \mid S = 1)$.

    \item \textbf{Equalised Odds} (EO) \citep{hardt_equality_2016} requires equal true positive and false positive rates across groups, conditional on the true outcome $Y$.

    \item \textbf{Positive Predictive Value Parity} (PPV) \citep{chouldechova2017fair} requires equal precision amongst those predicted positive.

\end{itemize}

\subsection{Experimental Protocol}
\label{subsec:sm_experiment_protocol}

For each predictive parity criterion (SP, EO, PPV), we train classifiers with different fairness regularisation weights $\lambda$. Each configuration is run with five independent random runs. We select the best balance between predictive performance (Accuracy, AUC, F1) and fairness (the disparity metric for each criterion), then analyse effort disparities for the selected model.

Since these classifiers are logistic, effort is measured using the exact logit-scale weighted distance in Eq.~\eqref{eq:exact_logistic_efforts}, with $\bmW=[1,10^8,10^8,1,1,10^8,10^8,1,1]$. CLTV ($X_3$) receives weight $10^8$ not because it is immutable in principle, but to exclude CLTV and UPB ($X_5$) as separate direct levers for the same underlying financing adjustment. In the causal specification, a direct shift to UPB may still propagate to CLTV through the SCM; in the feature-independent benchmark, CLTV is held fixed as a direct-action coordinate. The matrix $\bmW$ can be adapted to reflect heterogeneous modification difficulty. The expression is the exact minimum cost for the unconstrained continuous action problem; when an actionability constraint binds, it is an unconstrained benchmark and the constrained problem must be solved explicitly.

\subsection{Illustration of Effort Distributions}
\label{subsec:sm_effort_visualisation}
Table~\ref{tab:dist_metrics_pp_family} summarises distance metrics between the effort distributions of rejected female and male borrowers under three predictive parity criteria. The consistency across both distribution-based metrics (KS, CVM) and histogram-based metrics (TV, JS, HE) indicates that predictive parity does not eliminate the disparities of effort required to overturn a loan rejection.

\begin{table}[H]
\centering
\footnotesize
\caption{Distance metrics across predictive parity criteria.}
\label{tab:dist_metrics_pp_family}
\resizebox{\textwidth}{!}{
\begin{tabular}{lccccc}
\hline
Criterion
& KS
& CVM
& TV
& JS
& HE\\
\hline
SP  & 0.125 (0.093) & 0.299 (0.291) & 0.205 (0.110) & 0.044 (0.047) & 0.194 (0.123) \\
EO  & 0.097 (0.059) & 0.257 (0.350) & 0.149 (0.059) & 0.018 (0.012) & 0.130 (0.045) \\
PPV & 0.130 (0.100) & 0.313 (0.330) & 0.215 (0.113) & 0.048 (0.046) & 0.206 (0.117) \\
\hline
\end{tabular}
}
\parbox{\textwidth}{\footnotesize \textit{Notes:} Entries report mean (std) across five independent runs at $\lambda=0.8$. Lower values indicate more similar distributions between groups for all metrics shown, including Kolmogorov--Smirnov (KS), Cram\'er--von Mises (CVM), Total Variation (TV), Jensen--Shannon (JS), and Hellinger (HE).}
\end{table}

\clearpage

\section{Effort-Centric Fairness Framework: Full Results}
\label{sec:sm_effort_framework}

This section provides comprehensive results for our effort-centric fairness framework, including detailed performance metrics, statistical tests, financial profit, and analysis of predictive parity effects.

\subsection{Experimental Protocol}
\label{subsec:sm_experiment_protocol_2}

We train classifiers with different fairness regularisation weights $\lambda$. Each configuration is run with five independent runs. We evaluate:
\begin{itemize}
    \item \textbf{Feature-independent effort parity}: Penalises group-level differences in mean effort, assuming features can be modified independently.
    \item \textbf{Causal effort parity}: Penalises group-level differences in mean effort, accounting for causal propagation through the structural model.
\end{itemize}

\subsection{Feature-Independent Effort Parity}

\paragraph{Performance Metrics Across $\lambda$.}
Table~\ref{tab:sm_fi_group_performance} shows a gradual trade-off under
feature-independent effort parity. The effort gap falls by 59.98\% at the
largest $\lambda$, while AUC remains stable, the other predictive and revenue
measures decline modestly, and expected and unexpected loss decrease slightly.

\begin{table}[H]
\centering
\footnotesize
\caption{Performance metrics under feature-independent effort parity across $\lambda$ values. Mean (std) over five runs. Gap reduction is computed relative to $\lambda = 0$.}
\label{tab:sm_fi_group_performance}
\resizebox{\textwidth}{!}{
\begin{tabular}{lccccccccc}
\toprule
$\lambda$ & AUC & Accuracy & F1 & EL & UL & Revenue & RAROC & Gap Reduction (\%) \\
\midrule
0.00 & 0.724 (0.011) & 0.719 (0.098) & 0.823 (0.074) & 7251.48 (302.92) & 26258.26 (497.44) & 9972.07 (169.85) & 0.115 (0.017) & 0.00 \\
0.04 & 0.724 (0.010) & 0.714 (0.105) & 0.819 (0.079) & 7186.47 (260.49) & 26008.51 (381.10) & 9865.81 (210.54) & 0.115 (0.018) & 10.46 \\
0.08 & 0.724 (0.010) & 0.710 (0.111) & 0.815 (0.084) & 7143.39 (245.79) & 25805.41 (376.06) & 9768.19 (249.15) & 0.114 (0.018) & 18.01 \\
0.40 & 0.722 (0.007) & 0.702 (0.126) & 0.807 (0.098) & 7067.96 (196.80) & 25290.07 (387.50) & 9470.89 (337.28) & 0.108 (0.019) & 44.60 \\
0.80 & 0.722 (0.007) & 0.706 (0.131) & 0.809 (0.100) & 7130.55 (187.72) & 25302.13 (438.03) & 9401.26 (370.57) & 0.103 (0.019) & 59.98 \\
\bottomrule
\end{tabular}
}
\end{table}

\paragraph{Full Distribution Evolution.}
Fig.~\ref{fig:sm_fi_group_dist_full} displays the effort distributions across all $\lambda$ values and five independent runs. At $\lambda = 0.8$, the distributions for female and male borrowers move noticeably closer.

\begin{figure}[H]
\centering
{
\begin{tabular}{c}
{
\includegraphics[width=0.95\textwidth]{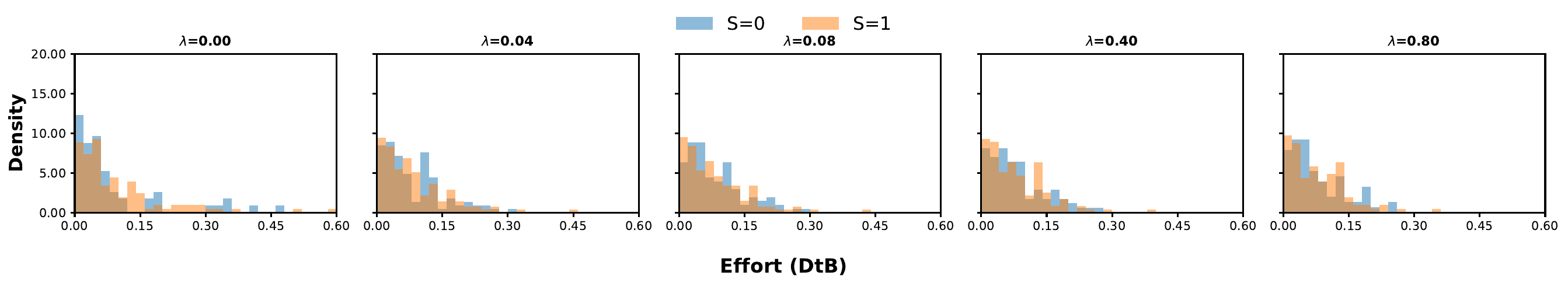}}\\[0.6ex]
{
\includegraphics[width=0.95\textwidth]{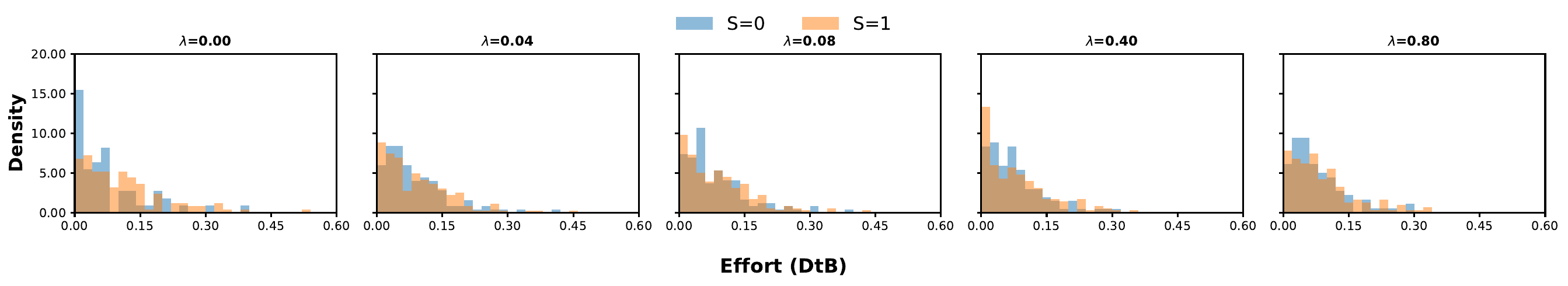}}\\[0.6ex]
{
\includegraphics[width=0.95\textwidth]{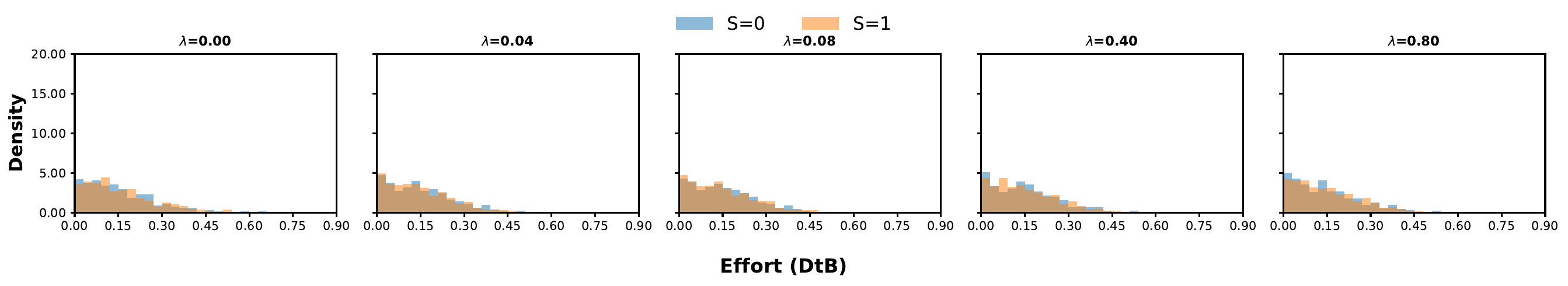}}\\[0.6ex]
{
\includegraphics[width=0.95\textwidth]{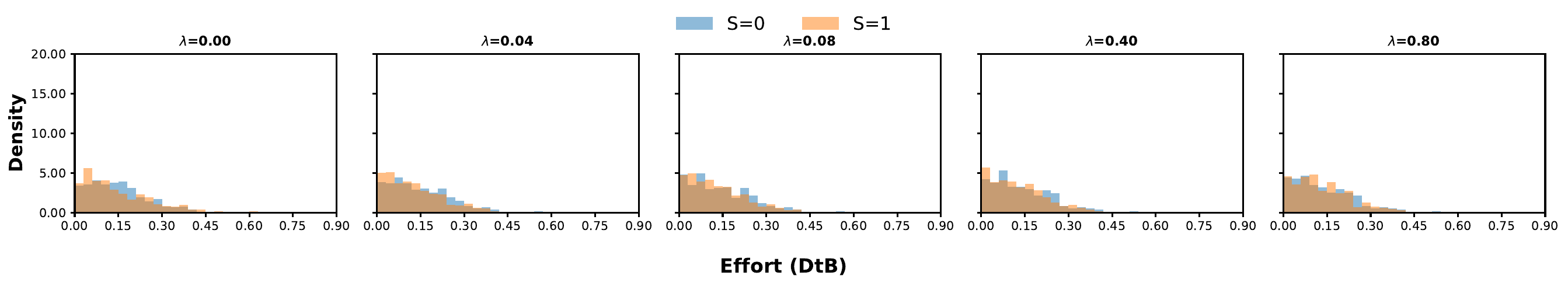}}\\[0.6ex]
{
\includegraphics[width=0.95\textwidth]{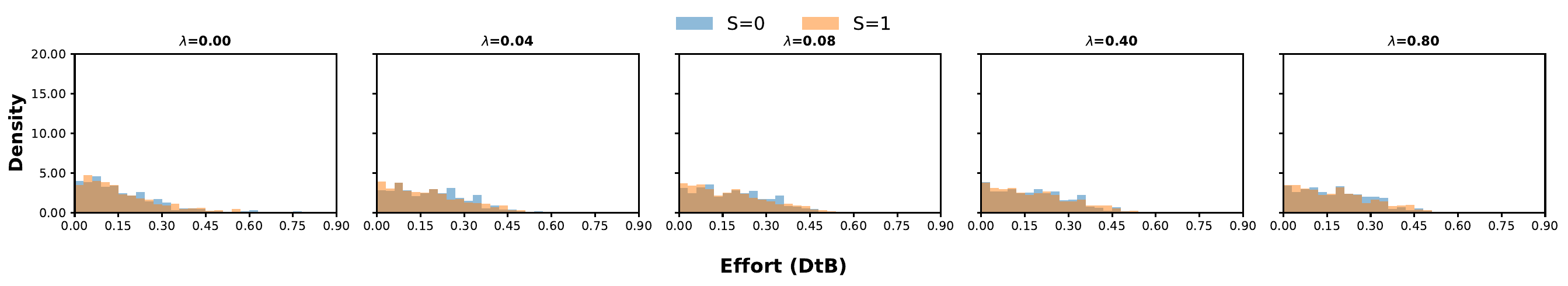}}
\end{tabular}
}
\caption{Evolution of effort distributions under feature-independent effort parity. Each row corresponds to an independent run.}
\label{fig:sm_fi_group_dist_full}
\end{figure}

\subsection{Causal Effort Parity}

\paragraph{Performance Metrics Across $\lambda$.}

Table~\ref{tab:sm_causal_group_performance} reports full performance metrics for causal effort parity.

\begin{table}[H]
\centering
\footnotesize
\caption{Performance metrics under causal effort parity across $\lambda$ values. Mean (std) over five runs. Gap reduction is computed relative to $\lambda = 0$.}
\label{tab:sm_causal_group_performance}
\resizebox{\textwidth}{!}{
\begin{tabular}{lccccccccc}
\toprule
$\lambda$ & AUC & Accuracy & F1 & EL & UL & Revenue & RAROC & Gap Reduction (\%) \\
\midrule
0.00 & 0.724 (0.011) & 0.719 (0.098) & 0.823 (0.074) & 7251.48 (302.92) & 26258.26 (497.44) & 9972.07 (169.85) & 0.115 (0.017) & 0.00 \\
0.04 & 0.669 (0.043) & 0.701 (0.084) & 0.813 (0.062) & 8770.43 (214.27) & 28021.17 (452.74) & 9697.82 (99.30) & 0.037 (0.006) & 91.21 \\
0.08 & 0.665 (0.039) & 0.705 (0.095) & 0.816 (0.069) & 8739.40 (219.38) & 27961.48 (479.60) & 9694.95 (97.00) & 0.038 (0.007) & 96.48 \\
0.40 & 0.665 (0.029) & 0.665 (0.095) & 0.784 (0.075) & 8885.24 (163.61) & 28376.81 (432.19) & 9819.56 (87.37) & 0.036 (0.005) & 98.65 \\
0.80 & 0.650 (0.026) & 0.630 (0.064) & 0.759 (0.053) & 8899.27 (220.62) & 28586.39 (639.48) & 9964.26 (184.29) & 0.040 (0.004) & 96.26 \\
\bottomrule
\end{tabular}
}
\end{table}

Causal effort parity responds more strongly at low regularisation: at
$\lambda=0.04$, the effort gap falls by approximately 90\%, compared with
about 10\% under feature-independent effort parity.

\paragraph{Full Distribution Evolution.}
Fig.~\ref{fig:sm_cg_group_dist_full} shows how the causal-effort
distributions evolve with $\lambda$.

\begin{figure}[H]
\centering
{
\begin{tabular}{c}
{
\includegraphics[width=0.95\textwidth]{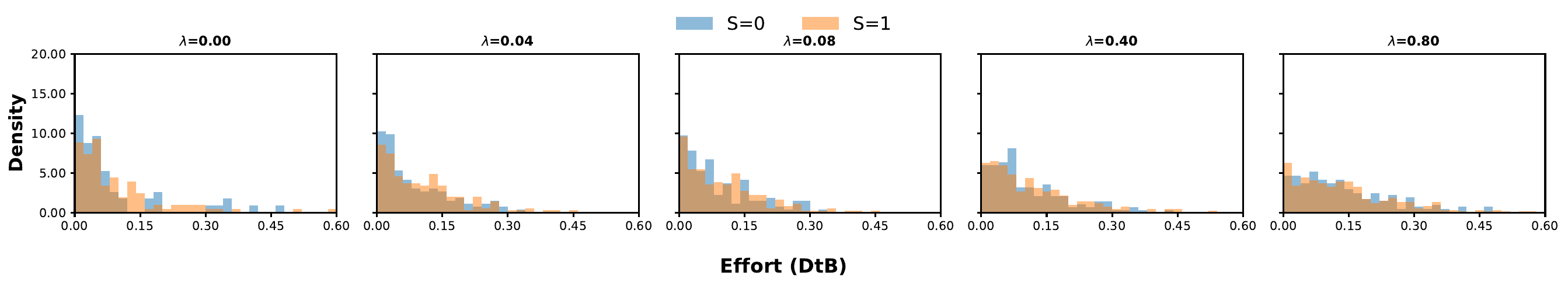}}\\[0.6ex]
{
\includegraphics[width=0.95\textwidth]{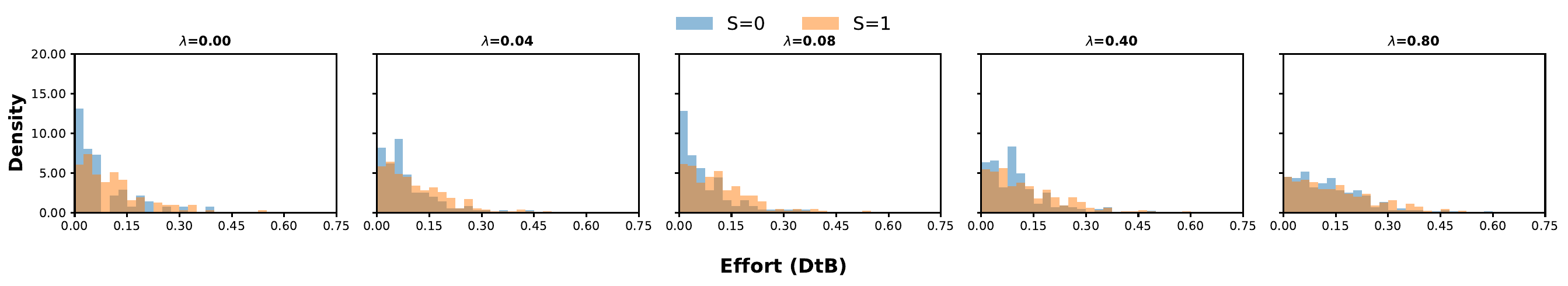}}\\[0.6ex]
{
\includegraphics[width=0.95\textwidth]{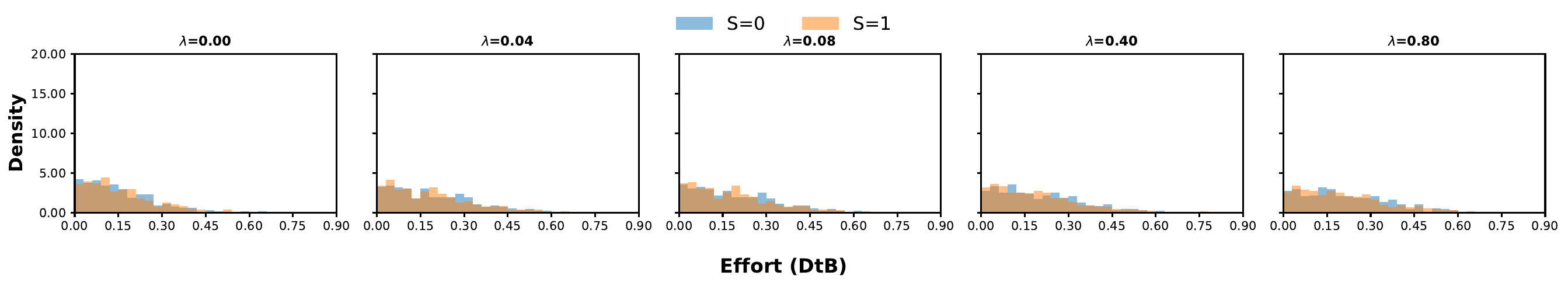}}\\[0.6ex]
{
\includegraphics[width=0.95\textwidth]{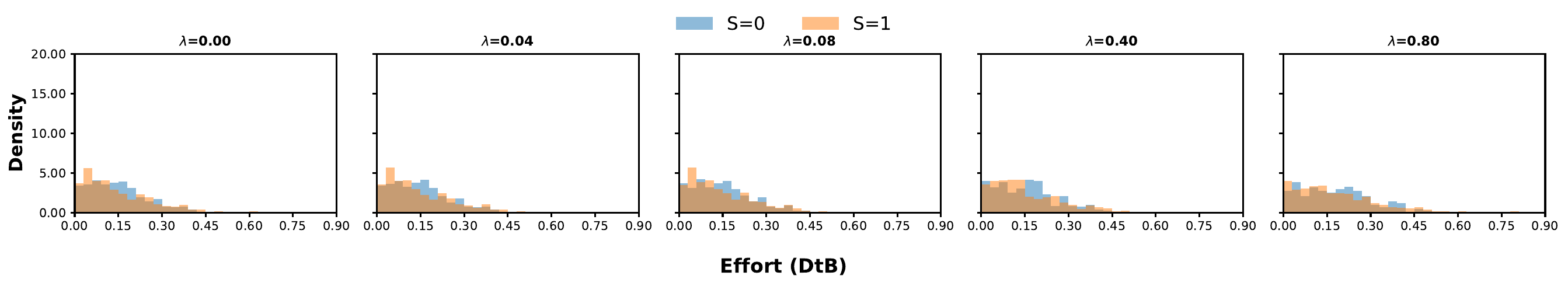}}\\[0.6ex]
{
\includegraphics[width=0.95\textwidth]{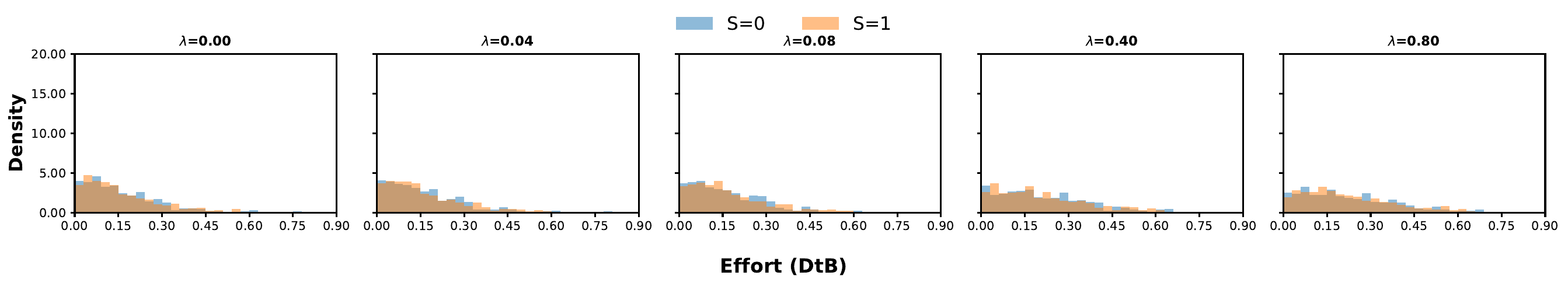}}
\end{tabular}
}
\caption{Evolution of effort distributions under causal effort parity. Each row corresponds to an independent run.}
\label{fig:sm_cg_group_dist_full}
\end{figure}

\subsection{Effects on Predictive Parity Metrics}
\label{subsec:sm_pp_effects}

Tables~\ref{tab:sm_fi_pp} and~\ref{tab:sm_cg_pp} indicate that EO disparity falls under feature‑independent effort parity, SP disparity falls under causal effort parity, and PPV disparity remains largely unchanged. It therefore improves several predictive fairness measures.

\begin{table}[H]
\centering
\footnotesize
\caption{Predictive parity metrics under feature-independent effort parity. Entries report mean (std) over five runs.}
\label{tab:sm_fi_pp}
\begin{tabular}{lccc}
\toprule
$\lambda$ & SP & EO & PPV \\
\midrule
0.00 & 0.0000173 (0.0000202) & 0.01273 (0.00472) & 0.00471 (0.00179) \\
0.04 & 0.0000174 (0.0000203) & 0.01238 (0.00472) & 0.00471 (0.00178) \\
0.08 & 0.0000180 (0.0000216) & 0.01202 (0.00469) & 0.00472 (0.00178) \\
0.40 & 0.0000236 (0.0000245) & 0.01061 (0.00458) & 0.00479 (0.00176) \\
0.80 & 0.0000245 (0.0000272) & 0.00988 (0.00430) & 0.00491 (0.00171) \\
\bottomrule
\end{tabular}
\end{table}

\begin{table}[H]
\centering
\footnotesize
\caption{Predictive parity metrics under causal effort parity. Entries report mean (std) over five runs.}
\label{tab:sm_cg_pp}
\begin{tabular}{lccc}
\toprule
$\lambda$ & SP & EO & PPV \\
\midrule
0.00 & 0.0000173 (0.0000202) & 0.01273 (0.00472) & 0.00471 (0.00179) \\
0.04 & 0.0000019 (0.0000014) & 0.00210 (0.00074) & 0.00572 (0.00141) \\
0.08 & 0.0000022 (0.0000016) & 0.00206 (0.00067) & 0.00572 (0.00141) \\
0.40 & 0.0000014 (0.0000013) & 0.00180 (0.00096) & 0.00567 (0.00141) \\
0.80 & 0.0000015 (0.0000015) & 0.00226 (0.00146) & 0.00560 (0.00138) \\
\bottomrule
\end{tabular}
\end{table}

\section{Cost Weight Specifications}
\label{sec:sm_weights}

The weight matrix $\bmW$ encodes the relative difficulty of changing each feature. We consider five configurations representing different applicant circumstances and preferences:

\begin{table}[H]
\centering
\footnotesize
\caption{Cost weight configurations. Higher values indicate greater difficulty of change. Immutable features receive weight $10^8$ (effectively infinite).}
\label{tab:sm_cost_weights}
\begin{tabular}{lccccc}
\toprule
Feature & Config.\ 1 & Config.\ 2 & Config.\ 3 & Config.\ 4 & Config.\ 5 \\
\midrule
$X_1$ (Credit score) & 0.2 & 0.1 & 0.1 & 0.4 & 0.2 \\
$X_2$ (Number of units) & $10^8$ & $10^8$ & $10^8$ & $10^8$ & $10^8$ \\
$X_3$ (CLTV) & $10^8$ & $10^8$ & $10^8$ & $10^8$ & $10^8$ \\
$X_4$ (DTI) & 0.2 & 0.2 & 0.4 & 0.05 & 0.1 \\
$X_5$ (UPB) & 0.2 & 0.1 & 0.1 & 0.4 & 0.2 \\
$X_6$ (Interest rate) & $10^8$ & $10^8$ & $10^8$ & $10^8$ & $10^8$ \\
$X_7$ (Loan purpose) & $10^8$ & $10^8$ & $10^8$ & $10^8$ & $10^8$ \\
$X_8$ (Number of borrowers) & 0.2 & 0.2 & 0.2 & 0.15 & 0.4 \\
$X_9$ (Income) & 0.2 & 0.4 & 0.2 & 0.05 & 0.1 \\
\bottomrule
\end{tabular}
\end{table}
\begin{itemize}
\item Configuration 1 represents a balanced view;
\item Configuration 2 penalises income changes;
\item Configuration 3 penalises DTI changes;
\item Configuration 4 penalises credit score and UPB changes;
\item Configuration 5 penalises changes to the number of borrowers.
\end{itemize}
\zc{The framework can be adapted to alternative cost environments; these configurations demonstrate that flexibility, while the treatment of CLTV and UPB follows Section~\ref{subsec:sm_experiment_protocol}.}

\clearpage

\section{Robustness Tests--SVM}
\label{sec:sm_svm}
In this section, we replace the classifier to assess robustness. The findings remain unchanged.

\subsection{Effort Parity Performance}

\subsubsection{Feature-independent Effort Parity}

Figs.~\ref{fig:sm_svm_fi_tradeoff} and~\ref{fig:sm_svm_fi_risk_profit_tradeoff} show that, despite mild non-monotonicity at low $\lambda$, feature-independent effort disparity achieves approximately 60--70\% at higher weights, while predictive performance, risk, revenue, and RAROC remain broadly stable.

\begin{figure}[H]
\centering
\subcaptionbox{Fairness improvement.}{
\includegraphics[height=5.0cm,keepaspectratio]{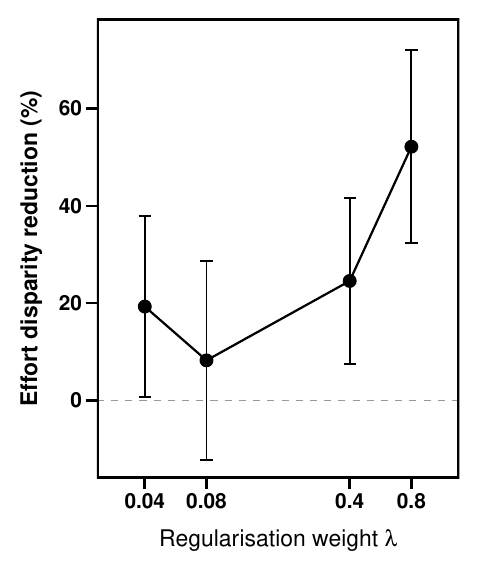}}
\subcaptionbox{Model performance stability.}{
\includegraphics[height=5.0cm,keepaspectratio]{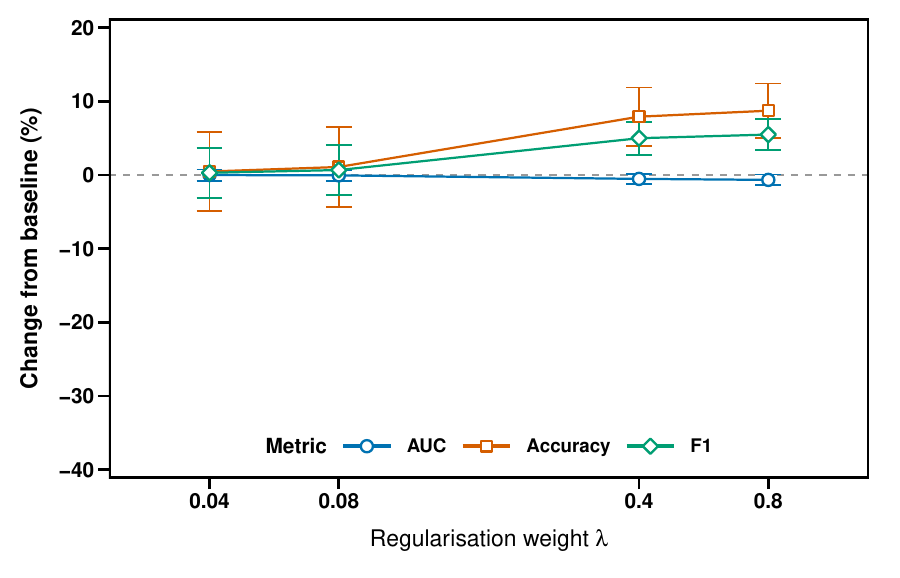}}
\caption{Relationship between predictive performance and effort disparity reduction under feature-independent effort parity. Effort disparity reduction is measured as a percentage relative to the baseline model. Each point corresponds to a different value of $\lambda$; error bars show standard errors over five runs.
\label{fig:sm_svm_fi_tradeoff}}
\end{figure}

\begin{figure}[H]
\centering
\subcaptionbox{Expected and unexpected loss.}{
\includegraphics[width=0.33\textwidth,keepaspectratio]{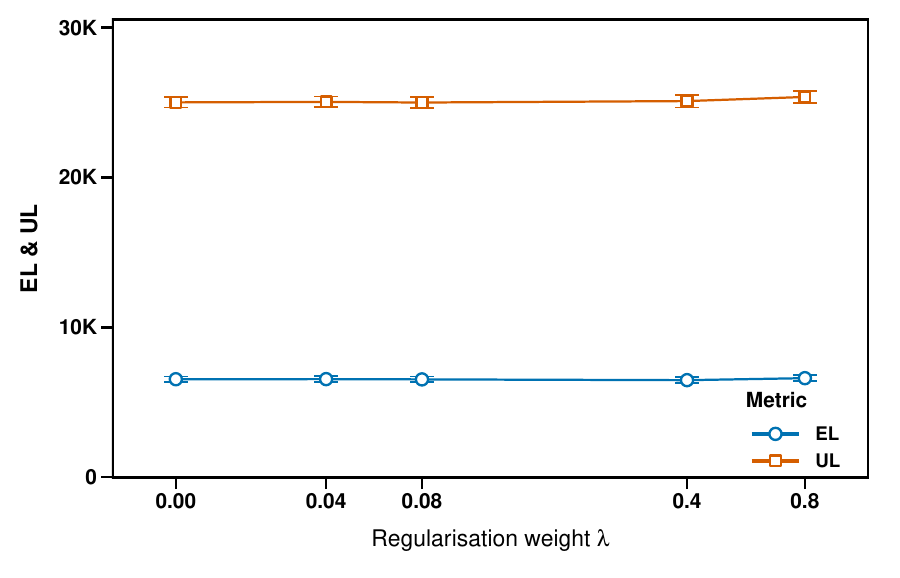}}
\subcaptionbox{Revenue.}{
\includegraphics[width=0.33\textwidth,keepaspectratio]{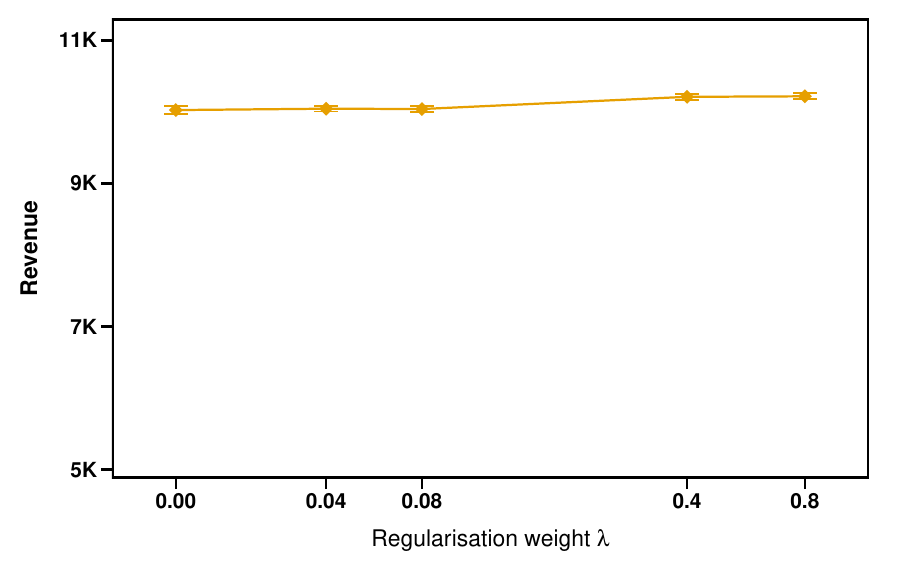}}
\subcaptionbox{Risk-adjusted return on capital.}{
\includegraphics[width=0.33\textwidth,keepaspectratio]{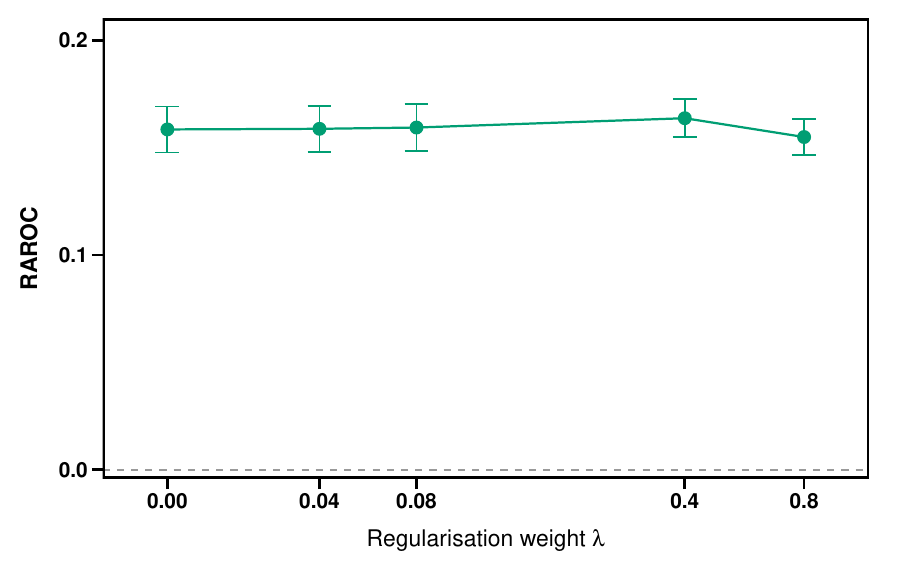}}
\caption{Relationship between credit risk, profitability, and effort disparity reduction under feature-independent effort parity. Each point corresponds to a different value of $\lambda$; error bars show standard errors over five runs.
\label{fig:sm_svm_fi_risk_profit_tradeoff}}
\end{figure}

\subsubsection{Causal Effort Parity}

The SVM results for causal effort parity show the same robustness pattern:
the disparity reduction rises approximately monotonically to 60--70\%, while
predictive performance and the reported financial measures remain broadly
stable across $\lambda$ (Figs.~\ref{fig:sm_svm_cg_tradeoff}
and~\ref{fig:cg_risk_profit_tradeoff_svm}).

\begin{figure}[H]
\centering
\subcaptionbox{Fairness improvement.}{
\includegraphics[height=5.0cm,keepaspectratio]{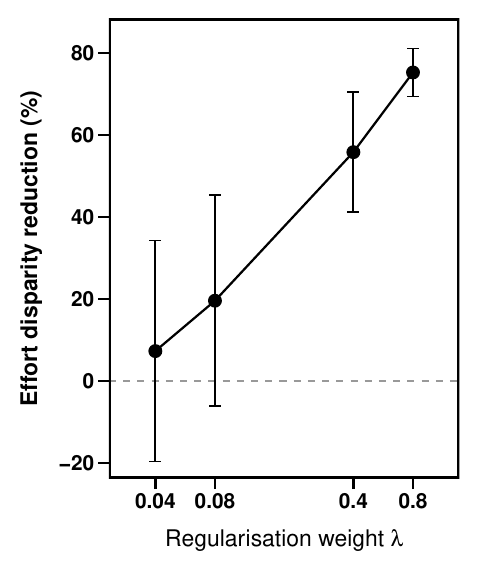}}
\subcaptionbox{Model performance stability.}{
\includegraphics[height=5.0cm,keepaspectratio]{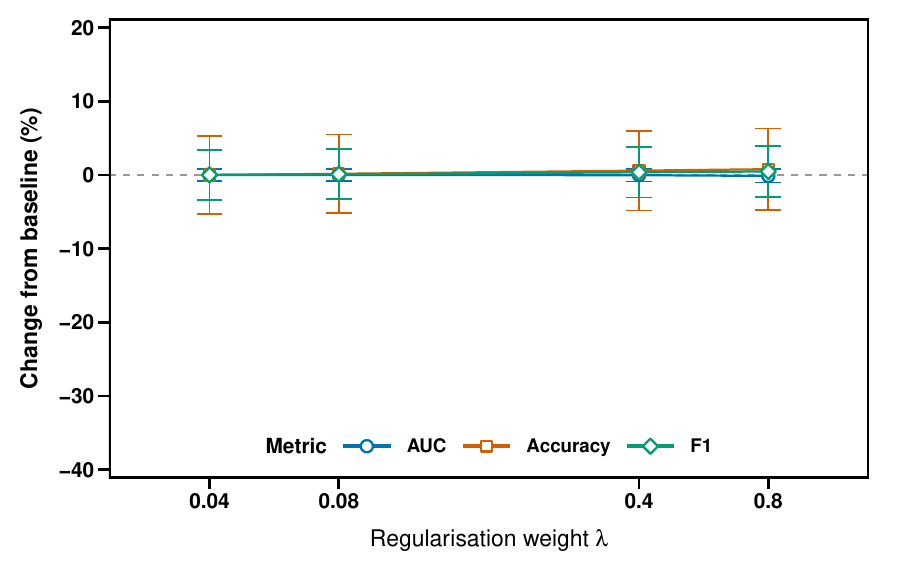}}
\caption{Relationship between predictive performance and effort disparity reduction under causal effort parity. Effort disparity reduction is measured as a percentage relative to the baseline model. Each point corresponds to a different value of $\lambda$; error bars show standard errors over five runs.
\label{fig:sm_svm_cg_tradeoff}}
\end{figure}

\begin{figure}[H]
\centering
\subcaptionbox{Expected and unexpected loss.}{
\includegraphics[width=0.33\textwidth,keepaspectratio]{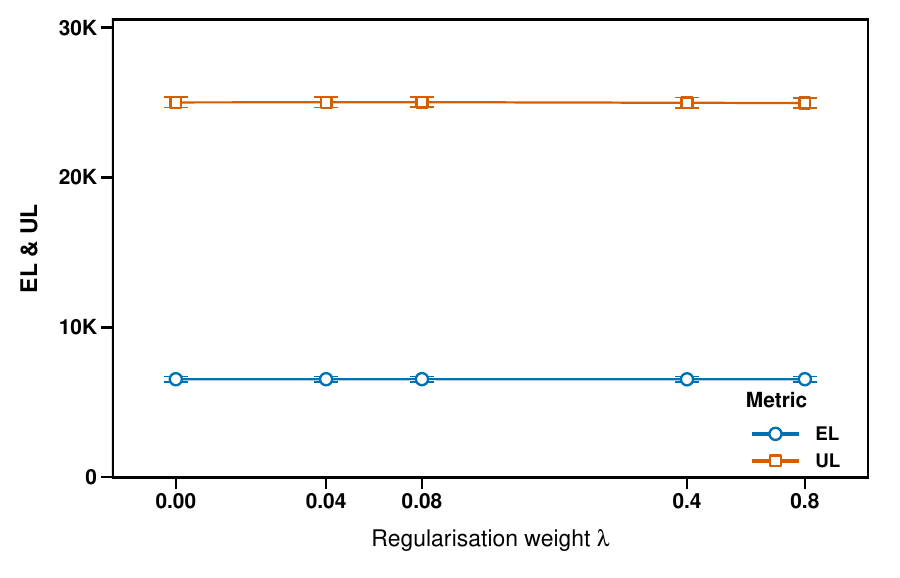}}
\subcaptionbox{Revenue.}{
\includegraphics[width=0.33\textwidth,keepaspectratio]{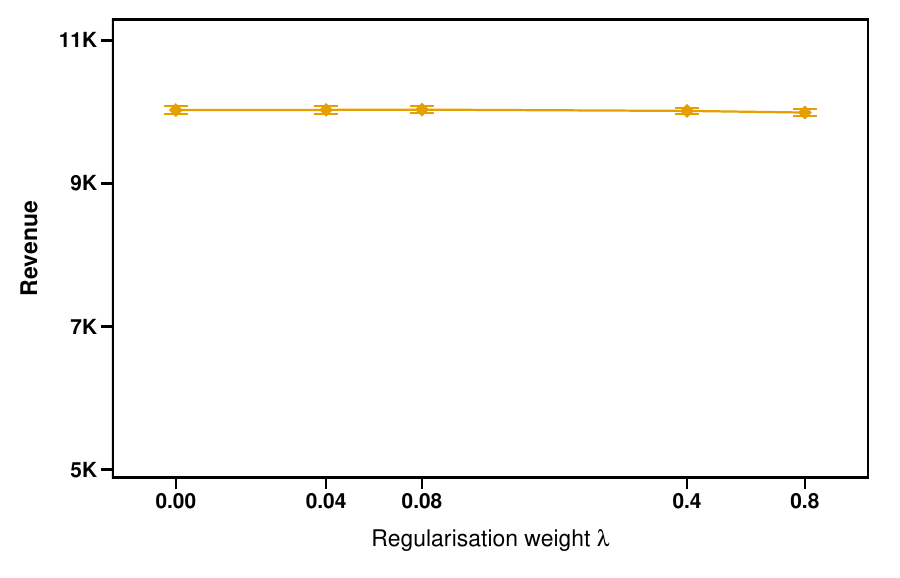}}
\subcaptionbox{Risk-adjusted return on capital.}{
\includegraphics[width=0.33\textwidth,keepaspectratio]{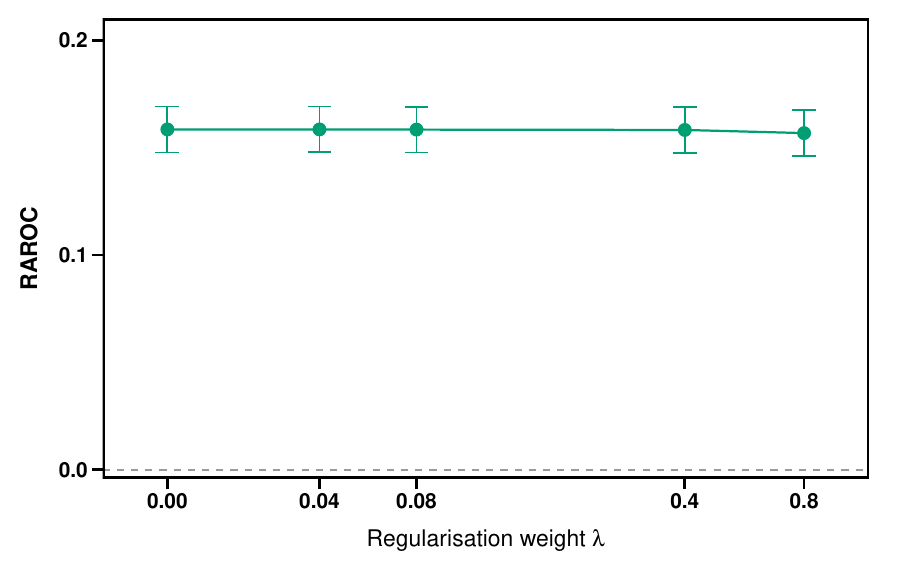}}
\caption{Relationship between credit risk, profitability, and effort disparity reduction under causal effort parity. Each point corresponds to a different value of $\lambda$; error bars show standard errors over five runs.
\label{fig:cg_risk_profit_tradeoff_svm}}
\end{figure}

\subsection{Effects on Predictive Parity Metrics}
Under the SVM classifier, improvements in predictive parity also emerge when effort parity is enforced (see Table~\ref{tab:sm_svm_fi_pp}-\ref{tab:sm_svm_cg_pp}).

\begin{table}[H]
\centering
\footnotesize
\caption{Predictive parity metrics under feature-independent effort parity. Entries report mean (std) over five runs.}
\label{tab:sm_svm_fi_pp}
\begin{tabular}{lccc}
\toprule
$\lambda$ & SP & EO & PPV \\
\midrule
0.00 & 0.00003 (0.00003) & 0.01430 (0.00551) & 0.00452 (0.00173) \\
0.04 & 0.00002 (0.00003) & 0.01418 (0.00548) & 0.00453 (0.00175) \\
0.08 & 0.00002 (0.00003) & 0.01399 (0.00533) & 0.00455 (0.00177) \\
0.40 & 0.00002 (0.00002) & 0.01265 (0.00474) & 0.00475 (0.00165) \\
0.80 & 0.00002 (0.00002) & 0.01166 (0.00396) & 0.00482 (0.00157) \\
\bottomrule
\end{tabular}
\end{table}

\begin{table}[H]
\centering
\footnotesize
\caption{Predictive parity metrics under causal effort parity. Entries report mean (std) over five runs.}
\label{tab:sm_svm_cg_pp}
\begin{tabular}{lccc}
\toprule
$\lambda$ & SP & EO & PPV \\
\midrule
0.00 & 0.00003 (0.00003) & 0.01430 (0.00551) & 0.00452 (0.00173) \\
0.04 & 0.00003 (0.00003) & 0.01434 (0.00550) & 0.00452 (0.00174) \\
0.08 & 0.00003 (0.00003) & 0.01437 (0.00550) & 0.00452 (0.00175) \\
0.40 & 0.00003 (0.00003) & 0.01447 (0.00546) & 0.00452 (0.00176) \\
0.80 & 0.00002 (0.00003) & 0.01444 (0.00527) & 0.00452 (0.00173) \\
\bottomrule
\end{tabular}
\end{table}

\section{Robustness Tests--Individual Effort Parity}
\label{sec:sm_individual_effort}
To test the robustness from comprehensive aspects, we also test the effort disparity and the mitigation at the individual level.

We define \emph{counterfactual $S$-twin}, denoted $\bmv_S^{\texto} \equiv \bmV_{\chi_S}(\bmu)$, as representing what an individual's features would be under a counterfactual change to their protected attribute, holding all exogenous factors constant. For example, the $S$-twin of a female applicant ($S=0$) is the hypothetical version of that individual had they been male ($S=1$), with all causally downstream features adjusted according to the structural equations. This construct underlies counterfactual fairness \citep{kusner_counterfactual_2017}, which requires that predictions be invariant to counterfactual changes in the protected attribute. Whilst counterfactual fairness concerns the \emph{prediction} an individual receives, our framework employs the $S$-twin to assess the \emph{effort} required to change that prediction.

\subsection{Individual Effort Parity}

Group-level criteria ensure that \emph{average} efforts are balanced but permit individual-level disparities. A stronger notion requires that each rejected applicant face the same effort as their counterfactual $S$-twin.

\begin{definition}[Individual Effort Parity]
\label{def:causal_individual_parity}
A classifier satisfies \emph{individual effort parity} if, for every rejected applicant $\bmv^{\texto} \in \mathcal{D}^-$, the minimal causal effort equals that of their counterfactual $S$-twin:
\begin{equation}
r^*(\bmv^{\texto}) = r^*(\bmv_S^{\texto}), \quad \forall \, \bmv^{\texto} \in \mathcal{D}^-.
\end{equation}
\end{definition}

This embodies a strong notion of fairness; an individual's path to approval should not depend on their demographic group membership, conditional on all other causally relevant factors. The corresponding \emph{individual effort disparity} is:
\begin{equation}\label{eq:ic_disparity}
\Delta^{\text{IC}}(\bmtheta) = \frac{1}{|\mathcal{D}^-|} \sum_{\bmv^i \in \mathcal{D}^-} \left| r^*(\bmv^i) - r^*(\bmv_S^i) \right|.
\end{equation}

\subsection{Individual Effort Parity Performance}
Fig.~\ref{fig:sm_ci_tradeoff} shows that individual effort regularisation achieves substantial disparity reductions, exceeding 80\% at $\lambda=0.40$, with minimal loss in predictive performance or profit. Profit decreases only modestly relative to baseline, demonstrating that even stringent individual-level fairness remains financially viable. Table~\ref{tab:sm_ci_performance} reports detailed performance metrics.

\begin{figure}[H]
\centering
\subcaptionbox{Fairness improvement.}{
\includegraphics[height=5.0cm,keepaspectratio]{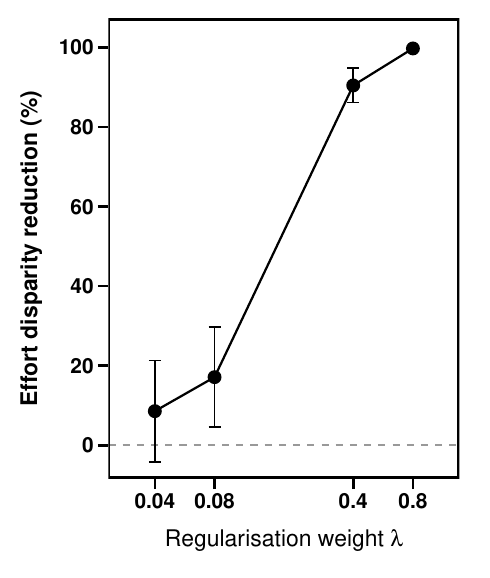}}
\subcaptionbox{Model performance stability.}{
\includegraphics[height=5.0cm,keepaspectratio]{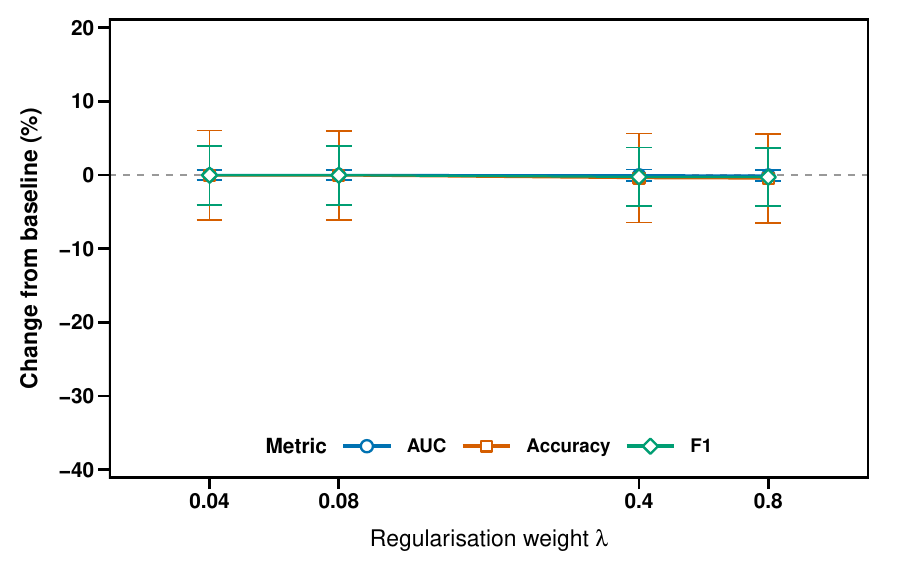}}
\caption{Relationship between predictive performance and effort disparity reduction under individual effort parity. Effort disparity reduction is measured as a percentage relative to the baseline model. Each point corresponds to a different value of $\lambda$; error bars show standard errors over five runs. For readability, the horizontal axis is displayed on a log scale.
\label{fig:sm_ci_tradeoff}}
\end{figure}

\begin{figure}[H]
\centering
\subcaptionbox{Expected and unexpected loss.}{
\includegraphics[width=0.33\textwidth,keepaspectratio]{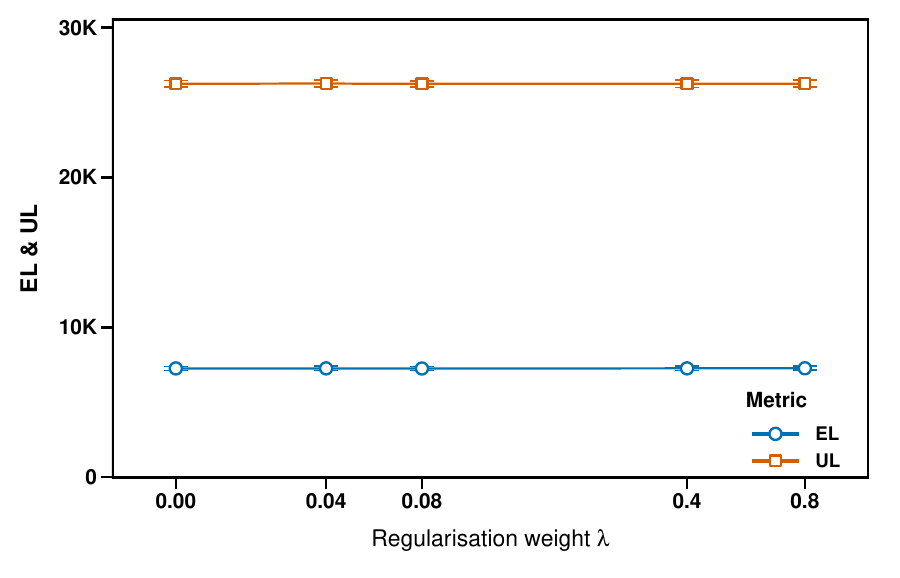}}
\subcaptionbox{Revenue.}{
\includegraphics[width=0.33\textwidth,keepaspectratio]{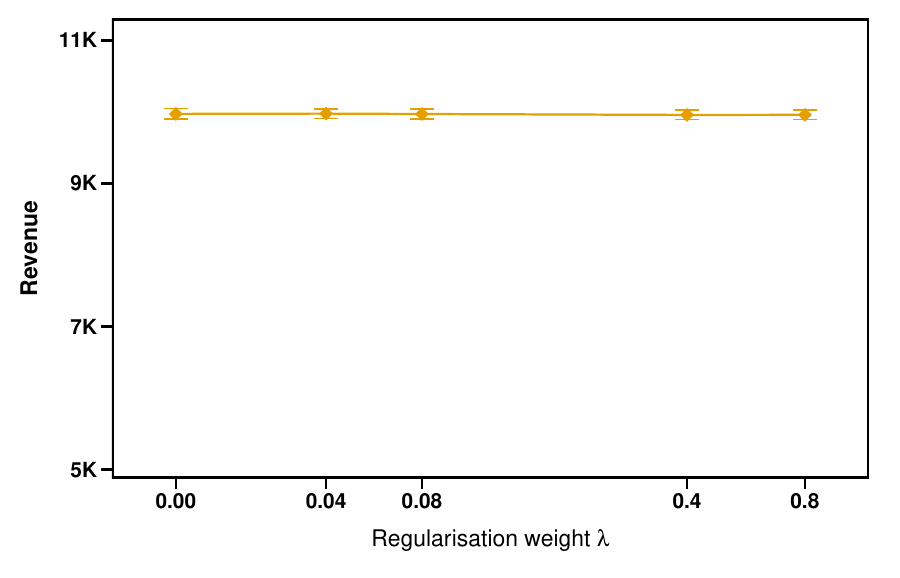}}
\subcaptionbox{Risk-adjusted return on capital.}{
\includegraphics[width=0.33\textwidth,keepaspectratio]{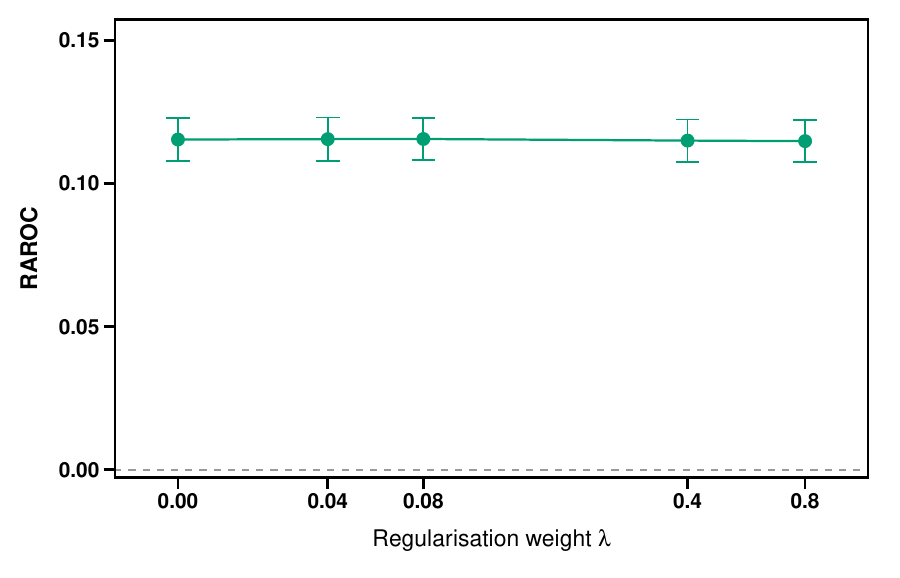}}
\caption{Relationship between credit risk, profitability, and effort disparity reduction under individual effort parity. Each point corresponds to a different value of $\lambda$; error bars show standard errors over five runs.
\label{fig:sm_ci_risk_profit_tradeoff}}
\end{figure}

\begin{table}[H]
\centering
\footnotesize
\caption{Performance metrics under individual effort parity across $\lambda$ values. Mean (std) over five runs. Gap reduction is computed relative to $\lambda = 0$.}
\label{tab:sm_ci_performance}
\resizebox{\textwidth}{!}{
\begin{tabular}{lccccccccc}
\toprule
$\lambda$ & AUC & Accuracy & F1 & EL & UL & Revenue & RAROC & Gap Reduction (\%) \\
\midrule
0.00 & 0.724 (0.011) & 0.719 (0.098) & 0.823 (0.074) & 7251.48 (302.92) & 26258.26 (497.44) & 9972.07 (169.85) & 0.115 (0.017) & 0.00 \\
0.04 & 0.724 (0.011) & 0.718 (0.098) & 0.823 (0.074) & 7253.06 (301.35) & 26269.14 (474.09) & 9976.00 (152.28) & 0.115 (0.017) & 8.56 \\
0.08 & 0.724 (0.012) & 0.718 (0.098) & 0.823 (0.074) & 7247.09 (292.00) & 26250.76 (469.50) & 9970.48 (161.08) & 0.116 (0.017) & 17.11 \\
0.40 & 0.724 (0.013) & 0.716 (0.097) & 0.821 (0.073) & 7255.69 (313.28) & 26250.93 (544.52) & 9957.38 (144.74) & 0.115 (0.017) & 90.43 \\
0.80 & 0.723 (0.013) & 0.715 (0.097) & 0.821 (0.073) & 7261.58 (313.82) & 26263.75 (545.54) & 9958.90 (146.37) & 0.115 (0.017) & 99.74 \\
\bottomrule
\end{tabular}
}
\end{table}

\subsection{Effects on Predictive Parity Metrics}

Table~\ref{tab:sm_ci_pp} reports predictive parity metrics under individual effort regularisation. The results are consistent with those in the main text.

\begin{table}[H]
\centering
\footnotesize
\caption{Predictive parity metrics under individual effort parity.}
\label{tab:sm_ci_pp}
\begin{tabular}{lccc}
\toprule
$\lambda$ & SP & EO & PPV \\
\midrule
0.00 & 0.00002 (0.00002) & 0.01273 (0.00472) & 0.00471 (0.00179) \\
0.04 & 0.00002 (0.00002) & 0.01284 (0.00467) & 0.00469 (0.00179) \\
0.08 & 0.00002 (0.00002) & 0.01295 (0.00462) & 0.00468 (0.00179) \\
0.40 & 0.00002 (0.00002) & 0.01382 (0.00414) & 0.00461 (0.00185) \\
0.80 & 0.00002 (0.00002) & 0.01389 (0.00399) & 0.00460 (0.00185) \\
\bottomrule
\end{tabular}
\end{table}

\bibliographystyle{plainnat}
\bibliography{reference}

\end{document}